\documentclass[a4paper, 11pt]{article}
\usepackage[utf8]{inputenc}
\usepackage{amsmath}
\usepackage{amssymb}
\usepackage{amsthm}
\usepackage[backend=biber,style=alphabetic]{biblatex}
\usepackage{url}
\usepackage{hyperref}
\usepackage{xcolor}
\usepackage{dsfont}
\usepackage{float}
\usepackage{mathrsfs}
\usepackage{yfonts}
\newtheorem{definition}{Definition}

\newtheorem{remark}[definition]{Remark}
\newtheorem{theorem}[definition]{Theorem}
\newtheorem{lemma}[definition]{Lemma}
\newtheorem{corollary}[definition]{Corollary}

\newtheorem{note}[definition]{Note}
\usepackage{graphicx}
 \usepackage{caption}
 \usepackage{subcaption}
 \usepackage{verbatim}
\usepackage{upgreek}
\usepackage{tikz}
\usetikzlibrary{shapes.geometric,plotmarks,backgrounds,fit,calc,circuits.ee.IEC}
\usetikzlibrary {arrows.meta}
\usetikzlibrary{patterns}
\usetikzlibrary{decorations.pathreplacing}
\usepackage{quantikz}

\usepackage{geometry}
\bibliography{references}
\newcommand{\bL}{\mathbf{L}}

\newcommand{\bD}{\mathbf{D}}

\newcommand{\cA}{{\cal A}}
\newcommand{\cB}{{\cal B}}

\newcommand{\cD}{{\cal D}}
\newcommand{\cE}{{\cal E}}
\newcommand{\cF}{{\cal F}}

\newcommand{\cI}{{\cal I}}
\newcommand{\cJ}{{\cal J}}

\newcommand{\cM}{{\cal M}}
\newcommand{\cN}{{\cal N}}

\newcommand{\cP}{{\cal P}}

\newcommand{\cR}{{\cal R}}
\newcommand{\cS}{{\cal S}}
\newcommand{\cT}{{\cal T}}
\newcommand{\cU}{{\cal U}}
\newcommand{\cV}{{\cal V}}
\newcommand{\cW}{{\cal W}}

\newcommand{\cZ}{{\cal Z}}

\newcommand{\CNOT}{\mathrm{CNOT}}

\newcommand{\ident}{\mathds{1}}
\newcommand{\dnorm}[1]{{\ensuremath{\|#1\|_\diamond}}}
\newcommand{\opn}[1]{{\ensuremath{\|#1\|_\infty}}}

\usepackage{braket}

\DeclareMathOperator{\tr}{Tr}
\DeclareMathOperator{\trans}{Trans}

\title{Fault-tolerant quantum input/output}
\author{Matthias Christandl$^1$, Omar Fawzi$^2$, and Ashutosh Goswami$^1$ \\[2mm]
    {\small $^1$Department of Mathematical Sciences, University of Copenhagen, Denmark} \\
    {\small $^2$Universit\'e de Lyon, Inria, ENS de Lyon, UCBL, LIP, France}\\
  {\small christandl@math.ku.dk, omar.fawzi@ens-lyon.fr, akg@math.ku.dk}}
\date{}

\begin{document}

\maketitle
\begin{abstract}
Usual scenarios of fault-tolerant computation are concerned with the fault-tolerant realization of quantum algorithms that compute classical functions, such as Shor's algorithm for factoring. In particular, this means that input and output to the quantum algorithm are classical. 
In contrast to stand-alone single-core quantum computers, in many distributed scenarios, quantum information might have to be passed on from one quantum information processing system to another one, possibly via noisy quantum communication channels with noise levels above fault-tolerant thresholds.  In such situations, quantum information processing devices will have quantum inputs, quantum outputs or even both, which pass qubits among each other. 

Working in the fault-tolerant framework of [Kitaev, 1997], we show that any quantum circuit with quantum input and output can be transformed into a fault-tolerant circuit that produces the ideal circuit with some controlled noise applied at the input and output. The framework allows the direct composition of the statements, enabling versatile future applications. We illustrate this with a concrete application, namely, communication over a noisy channel with faulty encoding and decoding operations [Christandl and M{\"u}ller-Hermes, 2024]. For communication codes with linear minimum distance, we construct fault-tolerant encoders and decoders for general noise (including coherent errors). For the weaker, but standard, model of local stochastic noise, we obtain fault-tolerant encoders and decoders for any communication code that can correct a constant fraction random errors.
%
%In the second application, we use our result for a state preparation circuit within the construction of [Gottesman, 2014] to establish that fault-tolerant quantum computation for general noise can be achieved with constant space overhead.
\end{abstract}

\section{Introduction}

\subsection{Background}

Many communication and computation scenarios involve black boxes, which implement a fixed quantum channel on an input-output space, while the user does not have control over their design and implementation. The user can only apply these boxes on a quantum input and receive the corresponding output. For example, in the communication setting, a noisy quantum channel, which is used as a resource to transmit information, acts as the black box~\cite{holevo1972mathematical, wilde2013quantum, watrous2018theory}. Another example is quantum learning, where an unknown quantum device, whose properties are to be determined, can be viewed as such a black box~\cite{giovannetti2011advances, montanaro2013survey,eisert2020quantum, arunachalam2017guest}. In general, a protocol with black boxes is a composition of many instances of black boxes and quantum circuits, designed by the user to optimize the performance of the protocol. In the communication setting, the performance of the protocol can be the number of bits/qubits that are reliably transmitted using the noisy quantum channel. In the setting of quantum learning, it can be the number of black box uses needed to accurately estimate the parameters of interest. 

\smallskip The design of black box protocols focuses typically on aspects of efficiency in terms of the performance, while it usually supposes that the protocols can be implemented perfectly without any circuit noise. However, due to the inherent noisy nature of quantum devices, this assumption is not realistic. As a consequence, one needs to implement such protocols in a fault-tolerant way. However, we note that the standard fault-tolerant methods for implementing quantum algorithms~\cite{shor1996fault, aharonov1997fault, kitaev, preskill1998fault, Aliferis2006quantum, gottesman2010introduction, gottesman2013fault, chamberland2016thresholds, yamasaki2024time} are not directly applicable to black box protocols because one can not implement black boxes in the logical subspace of an error correcting code. 

\medskip Standard fault-tolerant quantum computing deals with quantum circuits that have classical input and output. Such circuits consist of logic gates such as logical state preparation, unitary gates, and quantum measurement~\cite{preskill1998fault, Gottesman-book}. A fault-tolerant realization of such circuits is based on a quantum error-correcting code, which encodes a logical qubit in multiple physical qubits, such that the logical qubit is more resilient to noise than the physical qubits. The basic idea behind fault-tolerant quantum computing is that the logic  gates can be realized in the code space of an error correcting code with the help of fault-tolerant procedures, which make sure that errors do not propagate catastrophically, and, therefore, can be corrected by the error correcting code.

\smallskip In this article, working within Kitaev’s general framework for fault-tolerant quantum computation~\cite{kitaev}, we extend the standard model of fault-tolerant computation to circuits with quantum input and output. We then apply this extended framework to derive concrete results in the context of fault-tolerant quantum communication. We note that Kitaev’s framework is built around what we call a fault-tolerant scheme (see Definition~\ref{def:ft_scheme}), comprising three essential components: a family of quantum error-correcting codes, fault-tolerant implementations of logic gates, and encoding/decoding interfaces. It is worth emphasizing that this framework is general and does not prescribe a specific construction. Moreover, the third point, which implies the ability to interface between unencoded and encoded qubits is largely absent from traditional fault-tolerant constructions such as those in~\cite{aharonov1997fault, Aliferis2006quantum}.

\smallskip Roughly speaking, the interfaces allow to encode and decode arbitrary quantum state in an error correcting code, while ensuring that logical information is corrupted with an error rate $O(\delta)$, where $\delta$ is the physical error rate of circuit-level noise. Although Knill and Laflamme~\cite{knill1996concatenated} showed that this property holds for concatenated codes, their work focuses only on the encoding/decoding stage and does not provide a full treatment of circuits with quantum input/output. Our one of the main contributions is to provide an explicit, modular, and composable framework that integrates all three components of a fault-tolerant scheme into rigorous statements about fault-tolerant quantum circuits with quantum input and output. We demonstrate how this unified framework enables the design of fault-tolerant encoders and decoders for quantum communication over noisy channels.

To prove our results, we consider two kinds of independent circuit noise, namely general circuit noise, and circuit-level stochastic noise. For general circuit noise with rate $\delta$, a noisy gate is an arbitrary gate, which is $\delta$ close to the original gate in the diamond norm (quantum gates are seen as quantum channels in this framework).  For circuit-level stochastic noise with rate $\delta$, a noisy gate applies the ideal gate with probability $(1-\delta)$, and an arbitrary error channel with probability $\delta$. We suppose that the classical computation is noiseless.

\subsection{Our results} \label{sec:results}

We below present in an informal way our results including a key application in the context of quantum communication with noisy encoder and decoders.

\medskip Consider a circuit-level noise on a circuit with quantum input/output. Intuitively, noise in the quantum gates applied at the first and the last layer of the circuit is unavoidable. Therefore, the best fault-tolerant realization that is possible is that the noisy circuit realizes the original circuit up to independent weak noisy channels acting on the input and output systems. In this work, we show that this is essentially achievable. More precisely, we provide fault-tolerant realizations of circuits with quantum input/output, such that the noisy circuit realizes the original circuit up to some noisy channels acting on the input and output systems. Although the noisy channels thereof are not independent and can correlate quantum systems, we show that their detrimental effects can still be controlled by characterizing them as so called \emph{adversarial} and \emph{local stochastic} channels (see Def.~\ref{def:adver-can} and Def.~\ref{def:local-stoc}).

\medskip For any quantum circuit $\Phi$ with classical input and quantum output (i.e. state preparation circuits), we show that there exists another quantum circuit $\overline{\Phi}$, whose noisy version $[\overline{\Phi}]_{\delta}$, under general circuit noise with parameter $\delta$, realizes $\Phi$ up to an adversarial channel $\cW$ with parameter $O(\delta)$. Roughly speaking, an adversarial channel on $n$ qubits (i.e. a two-dimensional quantum system), with parameter $\delta$, affects non-trivially at most $O(n \delta)$ qubits. An informal version of our result for quantum circuits with classical input and quantum output for general circuit noise is given below.

\begin{theorem}[Informal, see Corollary~\ref{cor:dec} for details] \label{thm:main-ft-informal}
    Consider a quantum circuit $\Phi$ with classical input and  $n$ qubit output. 
    Let $\theta$ be the size (total number of locations) of $\Phi$, $k$ be any positive integer, and $\delta < \tfrac{1}{c}$ be a fixed noise rate for sufficiently large constant $c$, considering general circuit noise. Then, there exists a quantum circuit $\overline{\Phi}$, with the same input and output systems as $\Phi$ and of size $\theta \cdot \mathrm{poly}(k)$,
   such that for any noisy version $[\overline{\Phi}]_{\delta}$ of the circuit $\overline{\Phi}$, we have
    \begin{equation}
        \dnorm{[\overline{\Phi}]_{\delta} - \cW \circ \Phi} \leq  O(\theta \sqrt{(c\delta)^k}),
    \end{equation}
where $\cW$ is a $n$-qubit adversarial channel with parameter $O(\delta)$.
\end{theorem}
For the noisy version $[\overline{\Phi}]_{\delta}$, under circuit-level stochastic noise with parameter $\delta$, we show that it realizes $\Phi$ up to a local stochastic channel, with parameter $O(\delta)$. Roughly speaking, a local stochastic channel with parameter $\delta$ applies an arbitrary error channel on a subset of qubits containing $A \subseteq [n] := \{1, \dots, n\}$, with a probability that decreases exponentially in the size of $A$ as $\delta^{|A|}$. An informal version of our result for circuit-level stochastic noise is given below.  
\begin{theorem}[Informal, see Corollary~\ref{cor:dec-stc} for details] \label{thm:main-ft-informal-stc}
       Consider a quantum circuit $\Phi$ as in Theorem~\ref{thm:main-ft-informal}. Let $k$ be any integer and let $\delta < \tfrac{1}{c}$ be a fixed noise rate for sufficiently large constant $c$, considering circuit-level stochastic noise. Then, there exists a quantum circuit $\overline{\Phi}$, with the same input and output systems as $\Phi$ and of size $\theta \cdot \mathrm{poly}(k)$, 
   such that for any noisy version $[\overline{\Phi}]_{\delta}$ of the circuit $\overline{\Phi}$, we have
    \begin{equation}
        \dnorm{[\overline{\Phi}]_{\delta} - \cW \circ \Phi} \leq  O(\theta (c\delta)^k),
    \end{equation}
where $\cW$ is a $n$-qubit local stochastic channel with parameter $O(\delta)$.
\end{theorem}
We also provide analogous results to Theorems~\ref{thm:main-ft-informal} and \ref{thm:main-ft-informal-stc} for quantum circuits with quantum input and classical output (i.e. measurement circuits). In this case, an independent quantum channel acts on the input qubits, such that the independent  channel is a special case of adversarial or local stochastic channel depending on the circuit noise similarly to Theorems~\ref{thm:main-ft-informal} and \ref{thm:main-ft-informal-stc} (see Corollaries~\ref{cor:enc} and \ref{cor:enc-stc}). We note that for a Pauli circuit noise, a similar construction to Theorem~\ref{thm:main-ft-informal-stc} was used in~\cite{gottesman2013fault}.

\medskip 
%The quantum circuits with both quantum input and output are also considered.  
We also establish a theorem with both quantum input and output that is similar to Theorems~\ref{thm:main-ft-informal} and \ref{thm:main-ft-informal-stc} in the sense that noisy circuit $[\overline{\Phi}]_{\delta}$ gives a realization of $\Phi$ (with both qubit input and output), up to a noise channel with parameter $O(\delta)$ acting on the input and output qubits. In this case, however, the statement includes an auxiliary system, which acts as an environment to the input and output qubits. The noise channels are described in terms of an interaction with the auxiliary system, which can be arbitrarily correlated (see Fig.~\ref{fig:ft-picture-intro}, and Theorem~\ref{thm:main-ft2}).

%\medskip We provide two key applications: the first to quantum communication with noisy encoder and decoders, and the second to fault-tolerant computation with constant overhead for general noise.

%\paragraph{Applications} 
As an application of the fault-tolerant realizations of circuits with quantum input/output, we focus on a specific black box setting, namely fault-tolerant communication over noisy channels, which was previously considered in~\cite{christandl2022fault,belzig2023fault} for Pauli noise model. Fault-tolerant quantum communication is concerned with reliably sending quantum information, which is encoded in a computational code,  between the sender and the receiver that are connected by a noisy channel. The encoding and decoding circuits used for reliable fault-tolerant quantum communication are also considered to be noisy. Using fault-tolerant realizations for quantum circuits from Theorems~\ref{thm:main-ft-informal} and \ref{thm:main-ft-informal-stc} (and the analogous versions for circuits with quantum input and classical output), we show that a (standard) communication code, which is robust in the sense that it can reliably transmit information over a family of channels close to a communication channel $\Lambda$, can be used to construct a fault-tolerant communication code on $\Lambda$. This provides a general method for importing existing communication codes to the fault-tolerant scenario. An informal version of this result is given in the following theorem.

\begin{theorem} [Informal, see Theorems~\ref{thm:ft-comm-code-general} and \ref{thm:ft-code-stc} for details]  \label{thm:fault-tol-comm}
   Let $\Lambda$ be a $n$-qubit communication channel and let $\delta$ be a sufficiently small noise rate. Consider a communication code, which can reliably transmit quantum information on channels of type $ \cW' \circ \Lambda \circ \cW$, where $\cW, \cW'$ are arbitrary $n$-qubit adversarial (or local stochastic) channels with parameter $O(\delta)$. Then, the communication code can be used to construct a reliable fault-tolerant communication code for $\Lambda$, considering general circuit (or circuit-level stochastic) noise with the parameter $\delta$ in the encoding and decoding circuits.
\end{theorem}

Using Theorem~\ref{thm:fault-tol-comm}, we further show that one can construct reliable fault-tolerant communication codes for general circuit noise, using communication codes that have linear minimum distance, for example, asymptotically good codes~\cite{ashikhmin2001asymptotically, chen2001asymptotically, li2009family, matsumoto2002improvement, brown2013short, panteleev2022asymptotically, leverrier2022quantum}. This extends a similar result in~\cite{christandl2022fault}, derived for a Pauli circuit noise to general circuit noise. Furthermore, we show that for circuit-level stochastic noise, it is possible to construct fault-tolerant communication codes, using communication codes with sub-linear minimum distance. This improves significantly the result from~\cite{christandl2022fault}, while considering a more general circuit noise than the Pauli noise. As an explicit example, we show that fault-tolerant communication codes can be constructed using toric codes~\cite{kitaev, kitaev2003fault, bravyi1998quantum, dennis2002topological}, whose minimum distance scales as the square root of the code length.  We note that toric codes are only used for the purpose of illustration, a similar argument can be applied for codes with a linear rate of communication, for example, expander codes~\cite{tillich2013quantum, fawzi2018efficient}.

\begin{note}
 In a related work~\cite{christandl2025fault}, we use our state preparation method  from Theorem~\ref{thm:main-ft-informal} within the construction of~\cite{gottesman2013fault} to establish that fault-tolerant quantum computation for general noise can be achieved with constant overhead. We note that the state preparation theorem is only partly responsible for this application; further results on fault-tolerant error correction are needed to establish it.
\end{note}

\subsection{Set up and proof techniques}

\subsubsection{The fault-tolerant framework of Kitaev}

In this section, we describe in an informal way the fault-tolerant framework of Kitaev~\cite{kitaev}, highlighting some of the key features of the framework.  For more details, we refer the reader to Section~\ref{sec:kitaev-frame}.

\paragraph{Quantum codes.} Kitaev's framework considers a notion of \emph{many-to-one} quantum codes to encode logical information in physical qubits, where many-to-one refers to the fact that the representation of a quantum state in the code is not unique. Many-to-one quantum codes can be obtained from a standard (one-to-one) quantum code, encoding $m$ logical qubits in $n$ physical qubits, defined by an isometry $V$ from $m$ qubit space to $n$ qubit space. This is done as follows.    

\smallskip Let $\mathscr{E}$ be any set of error operators that can be corrected by a one-to-one code $C$.  The many-to-one code is defined by applying operators in $\mathscr{E}$ on the code states of $C$. Consider a code state $V \ket{\psi}$ of $C$, where $\ket{\psi}$ is a $m$-qubit state. In the many-to-one code $D$ derived from $C$ using correctable errors $\mathscr{E}$,  the state $E V\ket{\psi}, \forall E \in \mathscr{E}$ is a code state corresponding to $\ket{\psi}$. In fact, there exists a quantum channel $\mu^*$, which acts as an \emph{ideal decoder}, decoding any code state of the many-to-one code to its unencoded version in the sense that 

\begin{equation}
  \mu^*(E V\proj{\psi}V^\dagger E') \propto \proj{\psi},  \forall E, E' \in \mathscr{E}.
\end{equation}
where the proportionality constant solely depends on the operators $E, E'$. Moreover, $\mu^*$ is given by $\tr_F \circ \: \cU^\dagger$, where $\cU = U (\cdot) U^\dagger$ is a unitary channel, and system $F$ can be considered as a syndrome space (see Def.~\ref{def:many-to-one}).

\begin{remark}
It is worth noting that many-to-one codes are not better than the one-to-one code in terms of error correction performance. However, many-to-one codes are natural in the context of fault-tolerant quantum computing, where one can only hope to prepare code states with some correctable errors on it.
\end{remark}

\paragraph{Logic channels.}  
Using a notion of \emph{representation of quantum channels} in the code, Kitaev's framework allows applying quantum channels on the logical subspace of many-to-one codes.
%Roughly speaking, a quantum channel $\cT$ represents a quantum channel $\cP$ in the many-to-one code if the following holds,
% \begin{equation}
%     (\mu^* \circ \cT)(E V\proj{\psi}V^\dagger E') \propto \cP(\proj{\psi}), \forall E, E' \in \mathscr{E}.
% \end{equation}
% where $\circ$ denotes composition of two quantum channels. 
%
The representation of quantum channels is formalized using a commutative diagram (see Fig.~\ref{fig:strong-rep}) involving channels $\cP$ and $\cT$, wherein if the commutative diagram holds with accuracy $\epsilon > 0$, we say $\cT$ represents $\cP$ with accuracy $\epsilon$. The definition is general as $\cT$ may represent $\cP$ in two different many-to-one codes $D_1, D_2$. It also allows \emph{composability} in the sense that if $\cT_1$ represents $\cP_1$ in $D_1, D_2$, with accuracy $\epsilon_1$, and $\cT_2$ represents $\cP_2$ in $D_2, D_3$, with accuracy $\epsilon_2$, then $\cT_2 \circ \cT_1$ represents $\cP_2 \circ \cP_1$ (assuming that $\cP_2 \circ \cP_1$ is a valid composition) in the code $D_1, D_3$, with an accuracy $\epsilon_1 + \epsilon_2$ (see Lemma~\ref{lem:union-bound-rep}).

\paragraph{Circuit noise.} We consider two types of circuit noise, namely general and stochastic circuit noises. For general noise with parameter $\delta > 0$, any gate $g$ (a gate realizes a quantum channel) from a gate set is replaced by a noisy gate $\tilde{g}$, such that $\dnorm{g - \tilde{g}} \leq \delta$. For stochastic noise with parameter $\delta$, the noisy version is given by $\tilde{g} = (1 - \delta) g + \delta \cZ$, for some quantum channel $\cZ$. 

\paragraph{Fault-tolerant scheme.} To realize quantum circuits fault-tolerantly in the presence of the circuit noise, Kitaev's framework supposes existence of a fault-tolerant scheme, which consists of the following three points:

\begin{itemize}
    \item[$(a)$] \textbf{Quantum codes:} A sequence of many-to-one quantum codes $D_k,  k= 1, 2, \dots, $ encoding one logical qubit in $n_k$ physical qubits. Here, $D_1$ is the trivial code, that is, the corresponding $n_k = 1$.

    \item[$(b)$] \textbf{Logic gate:} For any gate $g$ from a gate set, a sequence of quantum circuits $\Psi_{g, k}, k= 1, 2, \dots, $, such that the noisy version $[\Psi_{g, k}]_\delta$, with noise parameter $\delta > 0$, represents the gate $g$ in codes $D_k$, with an accuracy $(c\delta)^k$, for a constant $c > 0$.  

    \item[$(c)$] \textbf{Interface:} A sequence of circuits $\Gamma_{k, l},  k, l= 1, 2, \dots $,  such that the noisy version $[\Gamma_{k, l}]_\delta$, with the parameter $\delta$, represents the identity gate in codes $D_k, D_l$, with accuracy $(c\delta)^{\min\{ k, l \}}$.   
\end{itemize}
The parameters $n_k$, $|\Psi_{g, k}|$ (size of $\Psi_{g, k}$), and $|\Gamma_{k, l}|$ are given by $\mathrm{poly}(k)$. We emphasize that circuits $\Gamma_{k, l}$ in point $(c)$ of the fault-tolerant scheme allows to switch between codes $D_k$'s, in a fault-tolerant way (i.e., preserving the logical information with a reasonable accuracy while switching). This is important for our fault-tolerant realizations of quantum circuits with quantum input/output. We note that the circuits $\Gamma_{k, l}$ are analogous to interfaces considered in~\cite{christandl2022fault}.
%
%
% \paragraph{Kitaev framework (1 page)}
%  simplified version of figure 2 

% highlight generality

% informal version of def 13. 

% generality of codes

% notion of representation (modularity)

% noise model 

% possibility to change the code 

% lemma 15 (in words(

\subsubsection{Our key technical contributions}

\paragraph{Transformation Rules.}
In Lemma~\ref{thm:main-ft}, we show that if a quantum channel $\tilde{\cT}$ approximately represents $\cP$ in the code $D_1, D_2$ with accuracy $\epsilon$, then there exists an exact representation $\cT$ of $\cP$, which is $O(\sqrt{\epsilon})$ close to $\tilde{\cT}$. This strengthens a result from~\cite{kitaev}, where the bound is $O(|M|\sqrt{\epsilon})$, where $|M|$ is the output dimension of $\cP$. Using the exact representation, we obtain a set of transformation rules, when $\cP$ is unitary, state preparation or measurement gate (see Corollaries~\ref{cor:tprime-unitary},~\ref{cor:tprime-stprep}, and~\ref{cor:tprime-measurement}). The obtained transformation rules show that the channel $\cT$ can be written as a product channel over the \emph{logical} and \emph{syndrome} space of the many-to-one code (see also Figures~\ref{fig:Tprime-ind},~\ref{fig:Tprime-prep}, and~\ref{fig:R-measurement}). This property is useful in our proofs of theorems for fault-tolerant quantum input/output.

%In particular, using these transformation rules, we can show that  

\paragraph{Interfaces.} In Appendix~\ref{sec:cons-ft-scheme}, we provide an explicit construction of fault-tolerant scheme. To this end, we translate the fault-tolerant constructions based on concatenated codes~\cite{Aliferis2006quantum, aharonov1997fault} to Kitaev's framework. Our main contribution in this section is showing the existence of quantum circuits $\Gamma_{k, l}$ according to Point $(c)$ of the fault-tolerant scheme.

\paragraph{Fault-tolerant realization of quantum circuits with quantum input/output. }

In Theorem~\ref{thm:main-ft2}, we show that for any quantum circuit $\Phi$ of size $\theta$, with $n$ qubit input and $n'$ qubit output, there exists a quantum circuit $\overline{\Phi}$ of size $\theta \: \mathrm{poly}(k)$, with the same input and output systems, such that the noisy version $[\overline{\Phi}]_\delta$, with parameter $\delta$, realizes the following channel, with an accuracy $O(\theta\sqrt{(c\delta)^k})$ (see also Fig.~\ref{fig:ft-picture-intro}),
\begin{equation} \label{eq:syndrome-noise-intro}
\dnorm{[\overline{\Phi}]_\delta - (\otimes_{i = 1}^{n'} \cW_i) \circ (\Phi \otimes \cF) \circ (\otimes_{i = 1}^n \cN_i)} \leq O( \theta \sqrt{(c\delta)^k}) ,
\end{equation}
where, for some Hilbert space $F_k$, $\cN_i, i \in [n]$ is a channel from the space $\bL(\mathbb{C}^2)$ (that is, the set of linear operators on the qubit Hilbert space $\mathbb{C}^2$) to $\bL(\mathbb{C}^2 \otimes F_k)$ and $\cW_i, i \in [n']$ is a channel from $ \bL(\mathbb{C}^2 \otimes F_k)$ to $\bL(\mathbb{C}^2)$ and $\cF$ is a channel from $\bL(F_k^{\otimes n})$  to $\bL(F_k^{\otimes n'})$. Furthermore, $\cN_i$ and $\cW_i$ satisfy the following inequalities ($\cI$ being the identity channel),
\begin{align}
    \dnorm{\tr_{F_k} \circ\: \cN_i - \cI_{\mathbb{C}^2}} \leq O(\delta). \label{eq:N-i-intro}  \\
    \dnorm{ \cW_i - \cI_{\mathbb{C}^2} \otimes \tr_{F_k}} \leq O(\delta). \label{eq:W-i-intro} 
\end{align}

\begin{figure}[!t]
    \centering
\begin{tikzpicture}
\draw
(0, 1) node[inner sep = 0pt] (a) {}
(0, -1) node[inner sep = 0pt] (b) {}
(0, -2) node[inner sep = 0pt] (e) {}
(0, -4) node[inner sep = 0pt] (f) {}
(1.5,0) node[inner xsep=0.5cm, inner ysep=1.2cm, draw] (c) {$\Phi$}
(1.5,-3) node[inner xsep=0.5cm, inner ysep=1.2cm, draw] (d) {$\cF$}
(-2, 0.8) node[draw, minimum size = 20pt] (g){$\cN_1$}
(-2, -1.2) node[draw, minimum size = 20pt] (h){$\cN_n$}
(5, 0.8) node[draw, minimum size = 20pt] (i){$\cW_1$}
(5, -1.2) node[draw, minimum size = 20pt] (j){$\cW_n$}
;

\draw
 (a -| g.east) to node[above]{$\mathbb{C}^2$} ++(0.5, 0) to (a -| c.west)
(b) to (b -| c.west)
(g.east) to node[below]{$F_k$} ++(0.5, 0) to   (0, -2) to (e -| d.west)
(g.west) to node[above]{$\mathbb{C}^2$} ++ (-0.8, 0) 
(b -| h.east) to ++(0.5, 0) 
(h.east) to ++(0.5, 0) to (0, -4) to (f -| d.west)
(h.west) to ++ (-0.8, 0)
;

\draw
(a -| i.west) to node[above]{$\mathbb{C}^2$} ++(-0.5, 0) to (a -| c.east)
(e -| d.east) ++(0.5, 0) to ($(i.west) +(-0.5, 0)$) to node[below]{$F_k$} (i.west)
(i.east) to node[above]{$\mathbb{C}^2$} ++(0.8, 0)
(e -| d.east) to ++(0.5, 0)
(b -| j.west) to ++(-0.5, 0)
(b -| c.east) to ++(0.5, 0)
(j.west) to ++(-0.5, 0) to ($(f -| d.east) + (0.5, 0)$) to (f -| d.east)
(j.east) to ++(0.8, 0)
;

\draw[dashed, thick]
(b -| h.east) ++(0.5, 0) to (0, -1) 
(b -| j.west)++(-0.5, 0) to ($(b -| c.east) + (0.5, 0)$) 
;  

\draw
($0.5*(g.south) + 0.5*(h.north)$) node[rotate = 90](){$\cdots$}
($0.5*(i.south) + 0.5*(j.north)$) node[rotate = 90](){$\cdots$}
;
\end{tikzpicture}
    \caption{The figure shows the channel $(\otimes_{i = 1}^{n'} \cW_i) \circ (\cT_{\Phi} \otimes \cF) \circ (\otimes_{i = 1}^n \cN_i)$ from Eq.~(\ref{eq:syndrome-noise-intro}).}
 \label{fig:ft-picture-intro}
\end{figure}

In Eq.~(\ref{eq:syndrome-noise-intro}),
systems $F_k$ act as an environment, which interacts with the ``data" qubits (systems corresponding to $\mathbb{C}^2$'s) at the first and last step through channels $\cN_i$'s and $\cW_i$'s, while between the first and the last step, the circuit $\Phi$ is applied perfectly on the data qubits, as depicted in Fig.~\ref{fig:ft-picture-intro}. Therefore, the noisy circuit $[\overline{\Phi}]_{\delta}$ realizes $\Phi$ up to noisy channels, which are $O(\delta)$ close to identity (in the sense of Eqs~(\ref{eq:N-i-intro}) and (\ref{eq:W-i-intro})), applied on each input and output qubit. We note that environment systems $F_k$'s can get correlated with each other through the quantum channel $\cF$, which is not necessarily a product channel. Furthermore, the environment can also have a memory meaning that it preserves correlation with the data qubits after the interaction.

\smallskip When the circuit noise is stochastic, Eq.~(\ref{eq:N-i-intro}) and Eq.~(\ref{eq:W-i-intro}) can be further simplified as follows,
  \begin{align}
    \tr_{F_k} \circ\: \cN_i = (1- \epsilon) \: \cI_{\mathbb{C}^2} + \epsilon \: \cZ_{\Gamma^i_{1, k}}, \label{eq:N_i-stoc-intro}\\
    \cW_i = (1- \epsilon) \: \cI_{\mathbb{C}^2} \otimes \tr_{F_k} + \epsilon \: \cZ_{\Gamma^i_{k, 1}},\label{eq:W_i-stoc-intro},
\end{align}
for $\epsilon = O(\delta)$ and some quantum channels $\cZ_{\Gamma^i_{1, k}}$, and $\cZ_{\Gamma^i_{k, 1}}$.

\paragraph{Proof idea:} To obtain the circuit $\overline{\Phi}$, we first consider a circuit $\Phi_k$, which is realization of $\Phi$ in the code $D_k$ in the sense that every location $g$ in the circuit $\Phi$ is replaced by the corresponding $\Psi_{g, k}$ from the fault-tolerant scheme. We then define $ \overline{\Phi}$, which has the same input and output systems as $\Phi$ 
%(see Fig.~\ref{fig:interface-intro}),  
\begin{equation} \label{eq:circuit-phi-bar-intro}
    \overline{\Phi} := (\otimes_{i=1}^{n'} \Gamma^i_{k, 1}) \circ \Phi_k \circ (\otimes_{i=1}^{n} \Gamma^i_{1, k}),
\end{equation}
where $\Gamma_{k, 1}$ and $\Gamma_{1, k}$ are interfaces according to Point $(c)$ of the fault-tolerant scheme. We insert identity $\cU_k \circ \cU_k^\dagger$ in the noisy circuit $[\overline{\Phi}]_\delta$, where $\cU_k$ is the unitary channel corresponding to the ideal decoder $\mu^*_k$ of $D_k$. The gives us,
\begin{equation} 
    [\overline{\Phi}]_\delta = \left(\otimes_{i=1}^{n'} ([\Gamma^i_{k, 1}]_\delta \circ \cU^i_k)\right) \circ \left( (\otimes_{i=1}^{n'} {\cU^i_k}^\dagger) \circ [\Phi_k]_\delta \circ (\otimes_{i=1}^{n} \cU^i_k)\right) \circ \left(  \otimes_{i=1}^{n} ({\cU^i_k}^\dagger \circ [\Gamma^i_{1, k}]_\delta) \right),
\end{equation}
Using the fact that $[\Gamma^i_{k, 1}]_\delta$ (or $[\Gamma^i_{1, k}]_\delta$) represents the identity channel in $D_k, D_1$ (or $D_1, D_k$) with accuracy $c\delta$, we show that $[\Gamma^i_{k, 1}]_\delta \circ \: \cU^i_k$ and ${\cU^i_k}^\dagger \circ [\Gamma^i_{1, k}]_\delta$ produce the quantum channels $\cW_i$ and $\cN_i$ from Fig.~\ref{fig:ft-picture-intro}, respectively (see also Lemmas~\ref{lem:interface-1} and~\ref{lem:interface-2}). Moreover, using the transformation rules mentioned above, we show that $(\otimes_{i=1}^{n'} {\cU^i_k}^\dagger) \circ [\Phi_k]_\delta \circ (\otimes_{i=1}^{n} \cU^i_k)$ produces the quantum channel $\Phi \otimes \cF$ (see also Lemma~\ref{lem:Phi-code-imp}).

% \begin{figure}[!b]
%  \centering
% \begin{tikzpicture}
% \draw 
% (0,1.6) node[isosceles triangle, isosceles triangle apex angle = 60, draw, shape border rotate = 180, inner sep = 0pt] (a) {$\Gamma^1_{1, k}$}
% (0,-1.6) node[isosceles triangle, isosceles triangle apex angle = 60, draw, shape border rotate = 180, inner sep = 0pt] (b) {$\Gamma^n_{1, k}$}
% (2,0) node[inner xsep=0.6cm, inner ysep=2cm, draw] (c) {$\Phi_k$}
% (4, 1.6) node[isosceles triangle, isosceles triangle apex angle = 60, draw, inner sep = 0pt] (d) {$\Gamma^1_{k, 1}$}
% (4, -1.6) node[isosceles triangle, isosceles triangle apex angle = 60, draw, inner sep = 0pt] (e) {$\Gamma^{n'}_{k, 1}$}
% ;

% \draw [ultra thick]
% (a.east) to node[above]{$N_k$} (a -| c.west)
% (b.east) to (b -| c.west)
% (d.west) to node[above]{$N_k$} (d -| c.east)
% (e.west) to (e -| c.east)
% ;

% \draw
% (a.west) to ++ (-0.8, 0) node[left]{$\mathbb{C}^2$}
% (b.west) to ++ (-0.8, 0)
% (d.east) to ++ (0.8, 0) node[right]{$\mathbb{C}^2$}
% (e.east) to ++ (0.8, 0)
% ;

% \draw[dotted, thick]
% (0.7,1.6) to (0.7,-1.6)
% (3.2,1.6) to (3.2,-1.6)
% ;
% \end{tikzpicture}
% \caption{The circuit $\overline{\Phi}$ according to Eq.~(\ref{eq:circuit-phi-bar-intro}). }
% \label{fig:interface-intro}
% \end{figure}

\begin{figure}[!t]
\centering
\begin{subfigure}[b]{0.45 \textwidth}
\centering
\begin{tikzpicture}
\draw 
(-0.3, 1.6) node[draw, minimum size = 6pt](a){$\overline{\cN}_1$}
(-0.3, -1.6) node[draw, minimum size = 6pt](b){$\overline{\cN}_n$}
(3.0, 1.6) node[above, right](e){}
(4.0, -1.6) node[](f){}
;
\draw 
(1.5,0) node[inner xsep=0.6cm, inner ysep=2cm, draw] (c) {$\cT_{\Phi}$}
(a) -- (a -| c.west)
(b) -- (b -| c.west)
(a) to ++ (-0.8, 0) 
(b) to ++ (-0.8, 0)
;
\draw
($0.5*(a.south) + 0.5*(b.north)$) node[rotate = 90](y){$\cdots$}
(y) ++(0, 0.9) node[rotate = 90] () 
{$\cdots$}
(y) ++(0, -0.9) node[rotate = 90] () 
{$\cdots$} ;
;

\draw[double]
(e) -- (a -| c.east)
;
\end{tikzpicture} 
\caption{Quantum input-classical output}
\end{subfigure}
\begin{subfigure}[b]{0.45 \textwidth}
\centering
\begin{tikzpicture}
\draw 
(0.4, 1.6) node[above, left](a){}
(4.3, 1.6) node[above, right](e){}
(4.3, -1.6) node[above, right](f){}
;
\draw 
(2,0) node[inner xsep=0.6cm, inner ysep=2cm, draw] (c) {$\cT_{\Phi}$}
(3.5,0) node[inner ysep=2cm, draw] (d) {$\overline{\cW}$}
(a) -- (a -| c.west)
(a -| c.east) -- (a -| d.west)
(f -| c.east) -- (f -| d.west)
(e) -- (e -| d.east)
(f) -- (f -| d.east)
;

\draw[double]
(a) -- (a -| c.west)
;

\draw
(d.east) ++ (0.25, 0) node[rotate = 90] (y) 
{$\cdots$}
(y) ++(0, 1.1) node[rotate = 90] () 
{$\cdots$}
(y) ++(0, -1.1) node[rotate = 90] () 
{$\cdots$} ;

\end{tikzpicture}
\caption{Classical input-quantum output}
\end{subfigure}
\caption{Fault-tolerant realizations of quantum circuits with quantum systems only at the input or output. The channels $\overline{\cW}$ and $\overline{\cN}$ are either adversarial or local stochastic (depending on the circuit noise), with parameter $O(\delta)$.}
\label{fig:ft-cl-qu-intro}
\end{figure}

\subsubsection{Techniques behind main results}
In this section, we highlight the techniques used in the proofs of the theorems stated in Section~\ref{sec:results}.

\noindent\paragraph{Fault-tolerant state preparation circuit (Theorem~\ref{thm:main-ft-informal} and~Theorem~\ref{thm:main-ft-informal-stc}):} For a state preparation circuit $\Phi$, we have $n = 0$ (no quantum input) and the channel $\cF$ in Eq.~(\ref{eq:syndrome-noise-intro}) is a state preparation map. Let $\eta$ be the state prepared by $\cF$,  then we get the following from Eq.~(\ref{eq:syndrome-noise-intro}) (see also part $(b)$ of Fig.~\ref{fig:ft-cl-qu-intro}),
\begin{equation} \label{eq:st-prep-noise-intro}
\dnorm{[\overline{\Phi}]_\delta -  \cW \circ \Phi } \leq O( \theta \sqrt{(c\delta)^k}),
\end{equation}
where $\cW := (\otimes_{i = 1}^{n'} \cW_i) (\cdot \otimes \eta)$ is a quantum channel, with the channels $\cW_i, i \in [n']$ satisfying Eq.~(\ref{eq:W-i-intro}). We show that $\cW$ can be approximated by as $ \dnorm{\cW -\cW'} \leq O(\exp(-O(n\delta))$, where $\cW'$ is a superoperator of weight $O(n\delta)$ in the sense that it can be written as linear combination of superoperators $ (\otimes_{i \in [n'] \setminus A} \cI_i) \otimes \cV_A$, where $\cV_A$ is a superoperator acting on qubits in set $A \subseteq [n'] := \{1, 2, \dots, n'\}$, with size $|A| = O(n\delta)$ (see also Fig.~\ref{fig:ft-cl-qu-intro}), that is,
\begin{equation} \label{eq:low-weight}
    \cW' = \sum_{A \subseteq [n'], |A| = O(n\delta)} (\otimes_{i \in [n'] \setminus A} \cI_i) \otimes \cV_A
\end{equation}
When the circuit noise is stochastic with parameter $\delta$, the channels $\cW_i$ in $\cW$ are stochastic with parameter $O(\delta)$, as given by Eq.~(\ref{eq:W_i-stoc-intro}). Using Eq.~(\ref{eq:W_i-stoc-intro}), we show that $\cW$ is a local stochastic channel with parameter $O(\delta)$ for any state $\eta$ as follows,

\smallskip The channel $\cW$ can be written as,
\begin{equation} \label{eq:local-stc-1-intro}
 \cW = \sum_{A \subseteq [n']} \big( \otimes_{i \in [n'] \setminus A} \cI_i \big) \otimes \cV_A,  
\end{equation}
for completely positive maps $\cV_A: \bL(M^{\otimes |A|}) \to \bL(M^{\otimes |A|})$ that only act on the set $A \subseteq [n']$ and that satisfy the following: for every $T \subseteq [n']$,
\begin{equation} \label{eq:local-stc-2-intro}
    \dnorm{\sum_{A : T \subseteq A} \big( \otimes_{i \in [n'] \setminus A} \cI_i \big) \otimes \cV_A } \leq \delta^{|T|}.
\end{equation}
We note that this condition is stronger than than the condition in Eq.~(\ref{eq:low-weight}) for $\cW'$, as roughly speaking, even though errors can be slightly correlated in local stochastic noise, they are still typical (random) errors. This is also a reason for calling $\cW'$ an adversarial noise. %Furthermore, local stochastic channel implies a stronger result for fault-tolerant coding than adversarial noise as discussed in the next paragraph.

Finally, we note that analogous results for a measurement circuit $\Phi$ can be obtained by substituting $n' = 0$ and $\cF = \tr_{F_k}^{\otimes n'}$ in Eq.~(\ref{eq:syndrome-noise-intro}) (see also part $(a)$ of Figure~\ref{fig:ft-cl-qu-intro}).

\noindent\paragraph{Fault-tolerant quantum communication (Theorem~\ref{thm:fault-tol-comm}).}
We apply fault-tolerant realizations of our state preparation and measurements circuits to design fault-tolerant encoders and decoders for communication using noisy channels. We consider the model of fault-tolerant communication from~\cite{christandl2022fault}, where the goal is to establish point-to-point communication between two fault-tolerant quantum computers. More precisely, one would like to transmit logical information encoded in the computational code (for example, $D_k$ from the fault-tolerant scheme) from one fault-tolerant computer to another.

\smallskip To specify the notion of logical information in the computational code, we consider a state preparation circuit $\Phi_1$ at the sender's end, and an arbitrary measurement circuit $\Phi_2$ at the decoder's end~\cite{christandl2022fault}. A qubit (or a set of qubits) can be completely defined by the classical output of the composed circuit $\Phi_2 \circ \Phi_1$. The logical information encoded in the code then can be defined using $\Phi_2^k \circ \Phi_1^k$, where $\Phi_1^k$ and $\Phi_2^k$  are realizations of $\Phi_1$ and $\Phi_2$ in the code $D_k$, respectively.

\smallskip A fault-tolerant coding to transmit $m$ logical qubits ($m \leq n$) on a $n$ qubit communication channel $\Lambda$ consists of finding an encoding circuit $\Phi_A$, that takes $m$ logical qubits as input and outputs $n$ physical qubits and a decoding circuit $\Phi_B$, which takes $n$ physical qubits as input and outputs $m$ logical qubits, respectively, such that there noisy realizations satisfy,
\begin{equation}
[\Phi_2^k]_\delta \circ [\Phi_B]_\delta \circ \Lambda \circ [\Phi_A]_\delta \circ  [\Phi_1^k]_\delta \approx  \Phi_2 \circ \Phi_1  
\end{equation}
for any preparation circuit $\Phi_1$, preparing a $m$ qubit state and any measurement circuit $\Phi_2$, measuring a $n$ qubit state. Note that $[\Phi_B]_\delta \circ \Lambda \circ [\Phi_A]_\delta$ simulates identity on $D_k$. We also note that the effective circuit applied at the encoder's end, that is, $\Phi_A \circ \Phi_1$ is a state preparation circuit and the effective circuit at the decoder end, that is, $\Phi_2 \circ \Phi_B$ is a measurement circuit. Below we give an outline of the proof of Theorem~\ref{thm:fault-tol-comm}, assuming the above model of fault-tolerant communication.

\smallskip   Consider the $n$ qubit communication channel $\Lambda$ in Theorem~\ref{thm:fault-tol-comm}. We run the preparation and measurement circuits at the sender and receiver's side, respectively, using our fault-tolerant realizations. As our fault-tolerant realizations ensure that the circuit noise is pushed to the last or first step (see Fig.~\ref{fig:ft-cl-qu-intro}), we can consider noiseless circuits at both the sender and receiver end, at the cost of an effective communication channel, which is given by $\cW \circ \Lambda \circ \cW'$, where $\cW$ is the channel due to preparation circuit and $\cW'$ is the channel from measurement circuit. The channels $\cW, \cW'$ are either adversarial or local stochastic channels depending on the circuit noise.  Therefore, to communicate reliably, one needs communication codes that are robust against a weak adversarial channel or local stochastic channel applied at the input and output of the communication channel.

\smallskip We also provide examples of robust communication codes. For adversarial noise, we show that a communication code, with a linear minimum distance gives a robust communication code (see Theorem~\ref{thm:ft-comm-code-general}). For local stochastic noise, we show that it suffices to consider the codes that correct random errors with a local stochastic distribution (see Theorem~\ref{thm:ft-code-stc}) and such codes do not need to have linear minimum distance. We demonstrate this with an example, where the toric code is used as a communication code (see Section~\ref{sec:toric-code}).

\subsection{Structure of the manuscript}
The remainder of the paper is organized as follows. In Section~\ref{sec:kitaev-frame}, we review basic definitions related to Kitaev's framework for fault-tolerant quantum computing~\cite{kitaev}, including circuit noise and the corresponding fault-tolerant schemes. In Section~\ref{sec:ft-tol-scheme}, we provide transformation rules, which are crucial for the fault-tolerant realizations of quantum input/output, and also provide a construction of fault-tolerant scheme based on the fault-tolerant constructions in~\cite{aharonov1997fault, Aliferis2006quantum}. In Section~\ref{sec:ft-real}, we prove our main results regarding the fault-tolerant implementation of circuits with quantum input/output. In Section~\ref{sec:ft-comm}, we provide quantum communication with noisy encoder and decoders.

\section{Kitaev's framework for fault-tolerant quantum computation} \label{sec:kitaev-frame}

\subsection{Notation and preliminaries} 

\paragraph{Notation.} We use symbols $M, L, N$ for finite dimensional Hilbert spaces, and $\bL(M, N)$ for linear operators from $M$ to $N$ or simply $\bL(M)$ for linear operators on $M$. We use $\bD(M)$ for quantum states on $M$. A qubit Hilbert space is denoted by $\mathbb{C}^2$. For a Hilbert space $M$, $|M|$ denotes its dimension, for example, $|\mathbb{C}^2| = 2$. Quantum channels are denoted by the calligraphic letters $\cT, \cP, \cR, \cN, \cV, \cW$, etc. The identity operator is denoted by $\ident$ and the identity quantum channel by $\cI$.

\medskip Let $[n] := \{1, \dots, n\}$ be a set of labels for $n$ qubits. For any $A \subseteq [n]$, we define the space of errors with support on $A$ as the space of linear operators acting only on qubits in $A$, i.e.,

\begin{equation} \label{eq:map-A}
  \mathscr{E}(A) := \{ (\otimes_{i \in [n] \setminus A} \ident_i) \otimes E_A : E_A \in \bL(M^{\otimes |A|})\} \subseteq \bL(M^{\otimes n}),
\end{equation}
 Using $\mathscr{E}(A)$, we define the space of operators with \emph{weight} $z \leq n$ qubits as follows.
\begin{definition}[Weight of an operator] \label{def:weight-z}
The set of operators with weight $z \leq n$ qubits is defined as, 
   \begin{equation} \label{eq:weight-z}
    \mathscr{E}(n, z) := \mathrm{span}\{ E : E \in \mathscr{E}(A) \text{ for some } A \subseteq [n] \text{ s.t. } |A| \leq z \}
\end{equation}
\end{definition}
Therefore, an operator in $\mathscr{E}(n, z)$ is a linear combination of operators in $\{\mathscr{E}(A_i) : A_i \subseteq [n], |A_i| < z\}$. Similarly, we can define the space of superoperators with weight $z$.
\begin{definition}[Weight of a superoperator] \label{def:weight}
 The space of superoperators with weight $z \leq n$ is defined as, 
\begin{equation} \label{eq:weight-z-sup}
  \mathscr{E}(n,z) (\cdot) \mathscr{E}(n, z)^*  = \mathrm{span}\{ E_1 (\cdot) E_2 : E_1 \in \mathscr{E}(n,z), E_2 \in \mathscr{E}(n,z)^*\}.
\end{equation}   
\end{definition}
We will use the diamond norm as a distance measure between two superoperators~\cite{kitaev, watrous2009semidefinite, kretschmann2008information}.

\begin{definition}[Diamond norm]
 For any superoperator $\cT: \bL(M) \to \bL(N)$, the diamond norm $\dnorm{T}$ is defined as,
\begin{equation}
    \dnorm{\cT} := \sup_{G} \| \cI_G \otimes \cT  \|_1,
\end{equation}
where the one-norm is given by,
\begin{equation}
 \| \cT  \|_1   := \sup_{\substack{\| \rho\|_1 \leq 1 \\ \rho \in \bL(M)}}  \| \cT(\rho) \|_1. 
\end{equation}    
\end{definition}
We will need the continuity theorem for Stinespring dilations from~\cite{kretschmann2008information, kretschmann2008continuity}, here formulated in the Schrödinger picture rather than the Heisenberg picture.

\begin{theorem} \label{thm:cont-stine}
Let $\cT_1, \cT_2: \bL(M) \to \bL(N)$ be two quantum channels, i.e. completely positive and trace preserving superoperators. Then 
 \begin{equation} \label{eq:cont-stine}
   \frac{\dnorm{\cT_1 - \cT_2}}{2} \leq  \inf_{V_1, V_2} \opn{ V_1 - V_2 } \leq \sqrt{\dnorm{\cT_1 - \cT_2}},
 \end{equation}
where the minimization extends over Stinespring dilations $V_1$ and $V_2 : M \to N \otimes G$, of $\cT_1$ and $\cT_2$, respectively.\end{theorem}

\subsection{Quantum codes}
We consider two definitions of quantum codes from~\cite{kitaev}, namely ``one-to-one" and ``many-to-one" codes. The one-to-one codes are a subfamily of the many-to-one codes. 

\begin{definition}[one-to-one quantum code] \label{def:one-to-one}
Let $V: M \to N$ be an isometry, where $M$ and $N$ are Hilbert spaces. A one-to-one quantum code $C$ of type $(N, M)$ is defined by the code space $\overline{L} := \mathrm{Im}(V) \subseteq N$ in the following sense. For any $\ket{\xi} \in M$, the corresponding code state is unique and is given by $V\ket{\xi} \in \overline{L}.$  If $N$ and $M$ are $n$-qubit and $m$-qubit systems, we say $C$ is a code of type $(n, m)$.
\end{definition}
Let a linear space $\mathscr{E} \subseteq \bL(N)$ be a space of error operators. We define another linear space $L = \mathscr{E} \overline{L} \subseteq N$.  We note that $L$ can be equal to $N$. Let $\mathscr{E} (\cdot) \mathscr{E}^*$ be the linear space of superoperators, consisting of maps of type $X(\cdot)Y^\dagger$ for $X, Y \in \mathscr{E}$ and linear combinations thereof. Then, a one-to-one code is said to correct errors from $\mathscr{E}$ if there exists a ``recovery" quantum channel $\cP: \bL(L) \to \bL(\overline{L})$ such that~\cite[Def.~8.2]{kitaev}
\begin{equation} \label{eq:ec-cond-1}
    (\cP \circ \cT)(\rho) = K(\cT) \rho, \: \forall  \cT \in \mathscr{E}(\cdot)\mathscr{E}^* \text{ and } \rho \in \bL(\overline{L}),
\end{equation}
where the constant $K(\cT)$ only depends on $\cT$.

\smallskip Let $E = \{E_i\}$ be a basis for $\mathscr{E} \subseteq \bL(N)$. We know that such a recovery channel exists if and only if~\cite{knill1997theory, nielsen2001quantum} 
\begin{equation} \label{eq:ec-cond-2}
    P E_i^\dagger E_j P = \alpha_{ij} P,  \: \forall E_i, E_j \in E.
\end{equation}
where $P = VV^\dagger \in \bL(N)$  is the projector on the code space $\overline{L} \subseteq N$ and $\alpha_{ij}$ are elements of a Hermitian matrix $\alpha$. From Eq.~(\ref{eq:ec-cond-2}), we can obtain another basis $F = \{ F_i\}$ for $\mathscr{E}$ such that (see~\cite[Chapter 10]{nielsen2001quantum} for more details)
\begin{equation} 
    P F_i^\dagger F_j P = d_{ij} P,  \: \forall F_i, F_j \in F,
\end{equation}
where $d$ is a diagonal matrix. Consider the polar decomposition of $F_iP$ given by a unitary $U_i$
\begin{equation}
  F_i P = U_i \sqrt{P F_i^\dagger F_i P} = d_{ii} U_i P.   
\end{equation}
It can be seen that $\{P_i \}_{i = 1}^r$ , where $P_i = U_i P U_i^\dagger$, defines a projective measurement on the subspace $L \subseteq N$, that is, $P_i P_j = \delta_{ij} P_i$ and $\sum_i P_i \equiv \ident_L$ (see~\cite[Chapter 10]{ nielsen2001quantum} for more details). Furthermore, the recovery channel $\cP : \bL(L) \to \bL(\overline{L}) $ is given by,
\begin{equation}
    \cP(\sigma) = \sum_i {U_i}^\dagger P_i \sigma P_iU_i, \: \sigma \in \bL(L).
\end{equation}
Note that $\cP$ can be extended to the full space $N$ by adding an extra Kraus operator $X(\ident_N - \sum_i P_i)$, where $X : N \to N$ is any unitary.

\medskip The recovery can also be done coherently by applying the unitary $W^\dagger: L \to \overline{L} \otimes F$, where $F$ is an auxiliary system of dimension  $r$ such that,
\begin{equation} \label{eq:W}
    W^\dagger \ket{\psi} =   \sum_{i= 1}^r (U_i^\dagger P_i \ket{\psi}) \otimes \ket{i}, \: \ket{\psi} \in L.
\end{equation}
Note that $\bra{\phi} W W^\dagger \ket{\psi} = \langle{\phi} |\psi\rangle$ and $|L| = |\overline{L} \otimes F| = r |\overline{L}|$. Hence, $W$ is a unitary.

\medskip\noindent Now consider an $\ket{\xi} \in \overline{L}$ and an $E \in \mathscr{E}$.  Note that 
\begin{align*}
    \tr_F(W^\dagger E|\xi\rangle \langle\xi|E^\dagger W) &= \sum_i U_i^\dagger P_i E|\xi\rangle \langle\xi| E^\dagger P_i U_i \\
    &= \cP( E  |\xi\rangle \langle\xi| E^\dagger) \\
    & =  K(E (\cdot) E^\dagger) |\xi\rangle \langle\xi|, 
\end{align*} 
where the last equality uses Eq.~(\ref{eq:ec-cond-1}). Therefore, we must have $U_i^\dagger P_i E|\xi\rangle \propto \ket{\xi}$, which implies 
\begin{equation} \label{eq:syndrome}
    W^\dagger E\ket{\xi} =     \ket{\xi} \otimes \ket{\varphi},
\end{equation}
where $\ket{\varphi}$ is an unnormalized state. 
\paragraph{Many-to-one codes.} 
 Let $\mathscr{E}$ be a space of errors that can be corrected by a one-to-one code $C$ of type $(N, M)$ defined by the subspace  $\overline{L} \subseteq N$, and consider the linear space $L = \mathscr{E} \overline{L} \subseteq N$. From Eq.~(\ref{eq:syndrome}), there exists a unitary $W: \overline{L} \otimes F \to L$ such that, for all $\ket{\xi} \in \overline{L}$ and $E \in \mathscr{E}$, we have~\cite{knill1997theory}
\begin{equation}
    W^\dagger  (E\ket{\xi}) =  \ket{\xi} \otimes \ket{\varphi}_F, 
\end{equation}
where $F$ is an auxiliary Hilbert space and $\ket{\varphi}$ is an unnormalized state depending on $E$.  The unitary $U: M \otimes F \to L , \: U = W(V \otimes \ident_F)$ defines a many-to-one quantum code in the following sense:  A quantum state $\gamma \in \bD(M)$ is encoded by any $\rho \in \bD(L)$ if $$\tr_F(U^\dagger \rho U)= \gamma.$$ 
The many-to-one code defined by $U: M \otimes F \to L$ is also called ``derived code"  and is denoted by $\mathrm{Der}(\overline{L}, \mathscr{E})$ in~\cite{kitaev} (see Fig.~\ref{fig:der-space}). Below, we give  a formal definition of many-to-one codes.
\begin{figure}[!b]
\centering
\begin{tikzpicture}
\draw (0, 0) -- (0,1.5) -- (1.5, 1.5) -- (1.5,0) -- (0,0);
\draw (-0.5, -0.5) -- (-0.5,2) -- (2.0,2.0) -- (2.0, -0.5) -- (-0.5, -0.5);
\draw (-0.8, -0.8) -- (-0.8, 2.8) -- (2.8,2.8) -- (2.8, -0.8) -- (-0.8, -0.8);
\draw (-1.5, -1.5) -- (0,-1.5) -- (0, -3.0) -- (-1.5, -3.0) -- (-1.5, -1.5);
\draw (0.5, -1.5) -- (2,-1.5) -- (2, -3.0) -- (0.5, -3.0) -- (0.5, -1.5);
\draw[decorate, decoration = {brace}]  (-1.5, -1.45) -- (2,-1.45);
\draw[->, thick] (0,1.5) to node[right] () {$\mathscr{E}$} (-0.5,2);

\draw[->, thick] plot [smooth] coordinates { (0.25, -1.4) (2.3, -0.8) (2, 1.25) };

\draw[->, thick] plot [smooth] coordinates { (-1.5, -2.0) (-2.0, 1.0) (0, 0.75) };

\draw
(0.75, 0.75) node[] () {$\overline{L}$}
(1.65, -0.3) node[] () {$L$}
(2.4, 2.4) node[] () {$N$}
(-0.75, -2.25) node[] () {$M$}
(1.25, -2.25) node[] () {$F$}
(-2.1, 1.0) node[] () {$V$}
(2.45, -0.8)  node[] () {$U$}
;
\end{tikzpicture}
\caption{Schematic representation of the spaces in the derived code $\mathrm{Der}(\overline{L}, \mathscr{\mathscr{E}})$. The isometry $V: M \to N, \mathrm{Im}(V) = \overline{L}$ is the original one-to-one quantum code. The unitary $U: M \otimes F \to L$, where $L = \mathscr{E}{\overline{L}} \subseteq N$, defines the many-to-one code.}
\label{fig:der-space}
\end{figure}
 
\begin{definition}[many-to-one quantum code] \label{def:many-to-one}
A many-to-one quantum code of type $(N,M)$ is defined by a unitary $U : M \otimes F \to L, L \subseteq N$ for some Hilbert space $F$. Let $\mu: \bL(M) \to \bL(L)$ be the superoperator defined by $\mu : X \mapsto U(X \otimes \ident_F) U^{\dagger}$. A quantum state $\gamma \in \bD(M)$ is encoded by any $\rho \in \bD(L)$ satisfying $$ \mu^*(\rho) = \tr_F(U^\dagger \rho U) = \gamma.$$
 A many-to-one code is fully specified by the pair $(L, \mu)$. As for one-to-one codes, if $N$ and $M$ are $n$-qubit and $m$-qubit systems, we say the code is of type $(n, m)$.
 \end{definition}

\medskip A quantum channel $\cE:\bL(M) \to \bL(L)$ is said to be an encoding transformation of the many-to-one code $(L, \mu)$ if $\cE(\gamma) \in \bL(L)$ is an encoding of $\gamma$ for all $\gamma \in \bL(M)$, that is $\mu^* \circ \cE = \cI_M$.
Encoding transformations of many-to-one codes have a particular form, as given in the following lemma.
\begin{lemma}[\cite{kitaev}]
   For any encoding transformation $\cE:\bL(M) \to \bL(L)$ of a code $(L, \mu)$, there exists an $\eta \in \bD(F)$ such that
    \begin{equation*}
        \cE(\cdot) = U (\cdot \otimes \eta)U^\dagger,
    \end{equation*}
    where $\mu( \cdot ) = U  (\cdot \otimes \ident_F) U^{\dagger}$.
\end{lemma}

\paragraph{Encoding for composite quantum systems.} For a composite quantum system consisting of $b$ subsystems, i.e. $M^{\otimes b}$, we consider a subsytemwise  enoding using the code $(L, \mu)$. This gives us the tensor product code  $(L^{\otimes b}, \mu^{ \otimes b})$. A quantum state $\rho \in \bD(L^{\otimes b})$ is an encoding of $\gamma \in \bD(M^{\otimes b})$ if the qubit-wise decoding yields
$$(\mu^{*\otimes b})(\rho) = \gamma.$$

\begin{remark}[Subsystem codes]
  The quantum code $(L, \mu)$  partitions the Hilbert space (up to a unitary equivalence) as $N = M \otimes F \bigoplus K  $, where $K = (M \otimes F)^\perp$ and the information is stored in the subsystem $M$ of the subspace $(M \otimes F)$. Therefore, $(L, \mu)$ is  a subsystem code~\cite{kribs2005operator}. $C^*$-algebra encodings for subsystem codes have been given in~\cite{bodmann2007decoherence}.
\end{remark}

\subsubsection{Many-to-one codes from stabilizer codes} \label{sec:stab-CSS-codes}

In this section, we explain how to obtain many-to-one codes from the Stabilizer codes. A stabilizer code $C$ of type $(n, 1)$ is  defined by a commuting set of Pauli operators $\{ S_1, \dots, S_{n-1}\}$, generating a subgroup of the $n$-qubit Pauli group not containing $-\ident$. Consider a bit string $\mathbf{s} = (s_1, \dots s_{n-1})$, where $s_i \in \{0, 1\}$ and consider the subspace $H_\mathbf{s} \subseteq (\mathbb{C}^2)^{\otimes n}$ such that any element in $H_\mathbf{s}$ is an eigenvector of $S_i$ with the eigenvalue $(-1)^{s_i}$ for all $i \in \{ 1, \dots, n-1\}$.

\medskip  We have that $H_\mathbf{s} \cong {\mathbb{C}^2}$. The spaces $H_\mathbf{s}$ and $H_{\mathbf{s}'}$ are orthogonal to each other for $s \neq s'$ as they are common eigenspaces  of the generating set $\{ S_1, \dots, S_{n-1}\}$ with respect to different eigenvalues. Therefore,
\begin{equation}
    (\mathbb{C}^2)^{\otimes n} = \bigoplus_{\mathbf{s} \in \{0, 1\}^{n-1}}  H_\mathbf{s}.
\end{equation}
The stabilizer code $C$ is a one-to-one code given by $\overline{L} := H_{(0,0, \dots, 0)}$. We associate the standard basis $\{ \ket{0}, \ket{1} \}$ of ${\mathbb{C}^2}$ to a basis $ \{\ket{\overline{0}}, \ket{\overline{1}} \}$ in  $\overline{L}$. For any $\ket{\psi} = \alpha \ket{0} + \beta \ket{1} \in {\mathbb{C}^2}$, the state $\ket{\overline{\psi}} = \alpha \ket{\overline{0}} + \beta \ket{\overline{1}} \in \overline{L}$ is the corresponding code state.

\smallskip Any Pauli error $E$ maps $\ket{\overline{\psi}} \in \overline{L}$ to a vector in $H_\mathbf{s}$, where $\mathbf{s} = (s_1,  \dots, s_{n-1}) \break \in \{0, 1\}^{n-1}$ is as follows:
$s_i = 0$ if $[S_i, E] = 0$ and $s_i = 1$ if $\{S_i, E\} = 0$. The vector $\mathbf{s}$ corresponding to the Pauli error $E$ is referred to as its syndrome.

\smallskip Consider now a set of Pauli errors $\{E_1, \dots, E_l \}$ such that syndromes of any $E_i$ and $E_j$ are different. Consider the error space $\mathscr{E} = \mathrm{span} \{E_1, \dots, E_l \}.$ We now define a derived (many-to-one) code $\mathrm{Der}(C, \mathscr{E})$ with the code space, 
\begin{equation} \label{eq:codespace-L}
    L := \mathscr{E}\overline{L} = \bigoplus_{\mathbf{s} \text{ of errors }  E_1, \dots, E_l } \: H_\mathbf{s} \subseteq (\mathbb{C}^2)^{\otimes n}.
\end{equation}
Note that $L \cong U({\mathbb{C}^2} \otimes F)$, where $F = \mathrm{span} \{ \ket{\mathbf{s}} : \mathbf{s} \text{ of errors in } \{ E_1, \dots, E_l \} \}$ and $U: {\mathbb{C}^2} \otimes F \to L$ is the unitary defined as follows, 
\begin{equation} \label{eq:U-unitary}
    U^\dagger E\ket{\overline{\psi}} = \ket{\psi} \otimes \ket{\mathbf{s}},  \forall \ket{\psi} \in {\mathbb{C}^2}, E \in \{ E_1, \dots, E_l \},
\end{equation}
where $\mathbf{s}$ is the syndrome of $E$.  The $\mu : \bL(\mathbb{C}^2) \to \bL(L)$ corresponding to $\mathrm{Der}(C, \mathscr{E})$ is given by $\mu(\cdot) = U (\cdot \otimes \ident_F)U^\dagger$. 

\subsection{Represention of quantum channels in  quantum codes}  \label{sec:rep}
To do fault-tolerant quantum computation using a many-to-one code $(L, \mu)$, we need to process the encoded quantum information, without leaving the code space. To this end, we need a notion of representation of  quantum channels in codes.

\medskip Let $\cJ : \bL(L) \to \bL(N), L \subseteq N$ be the natural embedding such that $$\cJ = J (\cdot) J^\dagger,$$ where $J: L \to N$ is the natural embedding of the Hilbert space $L$ into $N$. $J$ is an isometry, that is, $J^\dagger J = \ident_{L}$, and therefore $J J^\dagger:N \to N$ is the projector on the subspace $L \subseteq N$. Note that for a channel $\cT : \bL(N) \to \bL(N')$, the channel $\cT \circ \cJ:\bL(L) \to \bL(N')$ is the same as $\cT$, with the input being restricted to the code space $\bL(L)$.

\smallskip We now give the following definition, which is referred as the representation of a quantum channel in the strong sense in~\cite{kitaev}.

\begin{definition}[Representation of a quantum channel in quantum codes~\cite{kitaev}] \label{def:rep_codes}
    Consider two quantum codes $(L, \mu)$ and $(L', \mu')$ of types $(N, M)$ and $(N', M')$, respectively. Consider quantum channels $\cP : \bL(M) \to \bL(M ')$ and $\cT : \bL(N) \to \bL(N')$. We say that $\cT$ represents $\cP$ in codes $(L, \mu)$ and $(L', \mu')$ if the commutative diagram in Fig.~\ref{fig:strong-rep} holds for some quantum channels $\cR \in \bL(L) \to \bL(L')$ and $\cS \in \bL(N) \to \bL(L')$, that is, the following hold,
    \begin{align}
        (\mu^\prime)^* \circ \cR &= \cP \circ \mu^{*}. \label{eq:comm-1} \\
       \cJ'\circ \cR &= \cT \circ \cJ. \label{eq:comm-2} \\
       \cJ' \circ \cS  &= \cT. \label{eq:comm-3}
    \end{align}
Furthermore, we say $\cT$ represents $\cP$ in codes $C, C'$ with accuracy $(\epsilon_1, \epsilon_2, \epsilon_3)$ if Eqs.~(\ref{eq:comm-1}), ~(\ref{eq:comm-2}) and ~(\ref{eq:comm-3}) with accuracy  $\epsilon_1, \epsilon_2$ and $\epsilon_3$ in diamond norm, that is, 
       \begin{align}
        \dnorm{(\mu^\prime)^* \circ \cR  - \cP \circ \mu^{*}} & \leq \epsilon_1. \label{eq:comm-1-del} \\
       \dnorm{\cJ'\circ \cR - \cT \circ \cJ} & \leq \epsilon_2. \label{eq:comm-2-del} \\
       \dnorm{ \cJ' \circ \cS  - \cT} & \leq \epsilon_3. \label{eq:comm-3-del}
    \end{align}
\end{definition}

\begin{figure}[!t]
    \centering
   % \resizebox{0.5\linewidth}{!}{
\begin{tikzpicture}[node distance=2cm, auto]
  \node (A) {$\bL(M')$};
  \node(B) [right of =A]{$\bL(L')$};
  \node (C) [right of =B] {$\bL(N')$};
  \node (D) [below of =A] {$\bL(M)$};
  \node(E) [below of =B]{$\bL(L)$};
  \node(F) [below of =C]{$\bL(N)$};  
  \draw[<-] (A) to node {$(\mu^{\prime})^*$} (B);
  \draw[->] (B) to node {$\cJ'$} (C);
  \draw[<-] (D) to node {$\mu^{*}$} (E);
  \draw[->] (E) to node {$\cJ$} (F);
  \draw[<-] (A) to node {$\cP$} (D);
  \draw[<-] (B) to node {$\cR$} (E);
  \draw[<-] (C) to node {$\cT$} (F);
  \draw[<-] (B) to node {$\cS$} (F);
\end{tikzpicture}
%}
\caption{Commutative diagram. $\cT$ represents $\cP$ in codes $(L, \mu)$ and $(L', \mu')$. Here, $\cJ : \bL(L) \to \bL(N), L \subseteq N$ is the natural embedding.}
\label{fig:strong-rep}
\end{figure}
Note that Eq.~(\ref{eq:comm-3}) means that $\cT$ maps any element in $\bL(N)$ to an element in the code space $\bL(L')$, that is, $\mathrm{Im}(\cT) \subseteq \bL(L')$. In other words, $\cJ^* \circ \cT$ is a quantum channel.
This condition on $\cT$ is not always needed, leading to a weak representation as given in the remark below. 

\begin{remark}[Weak representation] \label{rem:weak-rep}
   In Def.~\ref{def:rep_codes}, we will say $\cT$ weakly represents $\cP$ if there exists $\cR$ satisfying Eqs.~(\ref{eq:comm-1}) and~(\ref{eq:comm-2}). Similarly, $\cT$ weakly represents $\cP$ with accuracy $(\epsilon_1, \epsilon_2)$ if Eqs.~(\ref{eq:comm-1-del}) and~(\ref{eq:comm-2-del}) hold.
\end{remark}
The following lemma ensures the composability of representations of quantum channels, which can themselves be composed.  

\begin{lemma} \label{lem:union-bound-rep}
    Let $\cP_1: \bL(M) \to \bL(M')$ and $\cP_2:\bL(M') \to \bL(M'')$ be quantum channels. Consider quantum codes $(L, \mu), (L', \mu')$ and $(L'', \mu'')$ of types $(N, M)$, $(N', M')$, and $(N'', M'')$, respectively. Suppose  $\cT_1: \bL(N) \to \bL(N')$ represents $\cP_1$ in codes $(L, \mu), (L', \mu')$, with accuracy $(\epsilon_1, \epsilon_2, \epsilon_3)$ and  $\cT_2: \bL(N') \to \bL(N'')$ represents $\cP_2$ in codes $(L', \mu'),(L'', \mu'')$, with accuracy $(\kappa_1, \kappa_2, \kappa_3)$. Then, $\cT_2 \circ \cT_1$ represents $\cP_2 \circ \cP_1$ in $(L, \mu), (L'', \mu'')$, with accuracy $(\epsilon_1 + \kappa_1, \epsilon_2 + \kappa_2, \kappa_3)$.
\end{lemma}
\begin{proof}
    Consider $\cR_1$ for $\cT_1$ and $\cP_1$, and $\cR_2, \cS_2$ for $\cT_2$ and $\cP_2$, according to Def.~\ref{def:rep_codes}. By taking $\cR = \cR_2 \circ \cR_1$, and $\cS = \cS_2 \circ \cT_1$, it can be easily seen that the commutative diagram in Fig.~\ref{fig:strong-rep} holds with accuracy $(\epsilon_1 + \kappa_1, \epsilon_2 + \kappa_2, \kappa_3)$ for $\cT = \cT_2 \circ \cT_1$ and $\cP = \cP_2 \circ \cP_1$, and for codes $D, D''$. 
\end{proof}
In order to make the reader more familiar with the introduced definition, we give an example of the representation of quantum channels in Appendix~\ref{app:ex-rep}.

\subsection{Quantum circuits and noise model} \label{sec:q_circuit-noise}
In this section, we provide the quantum circuit model, consisting of qubits and classical bits and the noise model for quantum circuits from~\cite{kitaev}.

\paragraph{Classical states:} 
We denote the space of a classical bit by $$\mathbb{B} \cong \mathrm{span} \{\proj{0}, \proj{1} \} \subset \bL(\mathbb{C}^2) .$$ The space of $n$ bits is given by $$\mathbb{B}^{\otimes n} \cong \mathrm{span}\{ \proj{\mathbf{u}} :  \mathbf{u} \in  \{0, 1 \}^{n}\} \subset \bL((\mathbb{C}^2)^{\otimes n}).$$

\paragraph{Combined classical-quantum states:}
Consider a hybrid classical-quantum system  with $q$-qubits and $b$-bits. A state in the combined space is given by,
\begin{equation} \label{eq:hybrid_classical_quantum}
    \rho = \sum_{\mathbf{u} \in \{0, 1\}^{b}} p(\mathbf{u}) ( \proj{\mathbf{u}} \otimes \gamma_{\mathbf{u}}) \in\mathbb{B}^{\otimes b} \otimes \bL((\mathbb{C}^2)^{ \otimes q}),
\end{equation}
where $p(\mathbf{u}) \geq 0$, s.t. $\sum_{\mathbf{u} \in \{0, 1\}^{b}} p(\mathbf{u}) = 1$ and $\gamma_{\mathbf{u}} \in \bD((\mathbb{C}^2)^{ \otimes q})$.

\subsubsection{Quantum circuits}
To construct quantum circuits, we need to use a gate set $\mathbf{A}$, which generates universal quantum computation. There are several finite universal quantum gate sets, however here we shall not fix a particular gate set but rather work in a general setting. We shall consider a gate set $\mathbf{A}$, containing quantum measurement, state preparation, and unitary gates, and classical controlled unitary gates. We also suppose that $\mathbf{A}$ contains an identity (idle) gate that realizes identity channel on a qubit.  

\begin{figure}[!b]
\centering
%\resizebox{0.6 \textwidth}{!}{
\begin{tikzpicture}
\draw
(0, 0.4) node[] (a){}
(0, -0.4) node[] (c) {}
(0, 1.3) node[] (d) {}
(a)++ (0.5, 0) node[above] (){{\footnotesize $OC$}}
(c)++ (0.5, 0) node[above] (){{ \footnotesize $Q$}}
(d)++ (0.5, 0) node[above] (){{ \footnotesize $CC$}}
(1.8, 0) node[draw, minimum height = 1.8cm, minimum width = 1.0 cm] (g) {$\cT_\mathbf{u}$}
;
\draw
(d) to ++(3.5, 0) node[right]  {}
(g.north) to ++(0, 0.4) 
(a -| g.west)[double, double distance = 0.45mm] to (a)
(a -| g.east)[double] to +(1.2, 0) node[right]  {}
;
\draw
(c -| g.west) to (c)
(c -| g.east) to +(1.2, 0)
node[right]  {}
;
\end{tikzpicture}
%}
\caption{The quantum channel $\cT_g$ realized by the quantum gate $g \in \mathbf{A}$. Each double wire ($CC$ and $OC$) represents a set of classical bits and the single wire ($Q$) represents a set of qubits. The state of classical bits are diagonal in the computational basis (see Eq.~\ref{eq:hybrid_classical_quantum}). Classical wires are divided into control and operational bits, labelled by $CC$ and $OC$, respectively. Depending on the state $\ket{\mathbf{u}}\bra{\mathbf{u}}_{CC}$ of the control wire, the quantum channel $\cT_{\mathbf{u}}$ is applied on $OC$ and $Q$.  } 
\label{fig:model-qgate}
\end{figure}

\medskip In general, we model a gate as follows (see also Fig.~\ref{fig:model-qgate}):
each gate $g \in \mathbf{A}$ has hybrid classical-quantum input and output. We divide classical bits as control bits and operational bits. Let $l$ be the number of control bits and $o, o'$ be the number of operational bits at the input and output, respectively. Let $q, q'$ be the number of qubits at the input and output, respectively.  Then, the gate $g$ realizes the quantum channel $\cT_g$,
\begin{equation} \label{eq:realization-qgate}
\cT_g = \sum_{\mathbf{u} \in \{0, 1\}^l} \proj{\mathbf{u}}(\cdot)\proj{\mathbf{u}} \otimes \cT^\mathbf{u},
\end{equation}
where $\cT^\mathbf{u}$ is a quantum channel from $ \mathbb{B}^{\otimes o} \otimes  \bL((\mathbb{C}^2)^{\otimes q})$ to $  \mathbb{B}^{\otimes o'} \otimes \bL((\mathbb{C}^2)^{\otimes q'})$.

\smallskip We now give the definition of quantum circuits, which is a slightly modified version of the one from~\cite{kitaev}. 

\begin{definition}[Quantum Circuit] \label{def:quantum_circuit}
    A quantum circuit $\Phi$ of depth $d$ is a collection of the following objects.
\begin{itemize}
    \item[$(1)$] A sequence of finite sets $\Delta_0, \dots, \Delta_i, \dots, \Delta_d$ called layers; the zeroth layer is called input of the quantum circuit and the last layer $\Delta_d$ is called the output. Each $\Delta_i = \Delta_i^c \cup \Delta_i^q$, where $\Delta_i^c$ contains classical wires and $\Delta^q$ contains quantum wires.     For a quantum circuit with no input, $\Delta_0$ is an empty set.
    
    \item[$(2)$] Partitions of each set $\Delta_{i-1}$, $i = 1, \dots, d$ into ordered subsets (registers) $A_{i1}, \dots, A_{is_i}$, as well as partition of some sets $\Delta'_i \supseteq \Delta_i$ into registers $A'_{i1}, \dots, A'_{is_i}$. Each register $A_{ij}, A'_{ij}$ contains $l_{ij}$ control bits, $q_{ij}, q'_{ij}$ qubits,  and $o_{ij}, o'_{ij}$ operational bits, respectively.

    \item[$(3)$] A family of gates $g_{ij}$ from the gate set $\mathbf{A}$, realizing  channels $\cT_{g_{i,j}}: \mathbb{B}^{
    \otimes {l_{ij}}} \otimes \mathbb{B}^{\otimes {o_{ij}}} \otimes \bL((\mathbb{C}^2)^{\otimes q_{ij}}) \to \mathbb{B}^{\otimes {l_{ij}}} \otimes \mathbb{B}^{\otimes {o'_{ij}}} \otimes \bL((\mathbb{C}^2)^{\otimes q'_{ij}})$.   The gate $g_{i,j}$ or their numbers $(i,j)$ are called the locations of the quantum circuit and the total number of locations is called the size of the quantum circuit.  
\end{itemize}
\end{definition}
\begin{remark} \label{rem:idle}
 \smallskip Note that in Def.~\ref{def:quantum_circuit}, all the qubits in $\Delta_{i-1}, i =1, \dots, d$, are acted upon by a gate at the $i^{th}$ layer. We suppose that the idle gate acts on each waiting qubit. This is natural for capturing the effect of noise on the waiting qubits (see Section~\ref{sec:noise-qcirc} below).    
\end{remark}
 
\smallskip A quantum circuit $\Phi$ according to Definition~\ref{def:quantum_circuit} realizes a quantum channel $\cT = \cT_d \circ \cdots \circ \cT_1$, where $\cT_i$ acts as follows on any  $\rho \in \bL({\mathbb{C}^2}^{\otimes |\Delta_{i-1}|})$,
\begin{equation}
    \cT_i (\rho) =  \tr_{\Delta'_i \setminus \Delta_i}\Big[\cT_{g_{i,1}}[A'_{i1}; A_{i1}] \otimes \dots \otimes \cT_{g_{i, s_i}}[A'_{is_i}; A_{is_i}] \Big] \left(\rho \right),
\end{equation}
where $\cT_{g_{i, j}}[A'_{ij}; A_{ij}]$ denotes the quantum channel $\cT_{g_{i,j}}$, applied to the register $A_{ij}$, whose output is stored in the register $A'_{ij}$.

\paragraph{Notation.} In the following, we shall denote the set of locations of $\Phi$ by $\mathrm{Loc}(\Phi)$, and its size by $|\Phi| := |\mathrm{Loc}(\Phi)|$.

\subsection{Noise in quantum circuits} \label{sec:noise-qcirc}
We shall work with two kinds of circuit noise, namely, general circuit-noise and circuit-level stochastic noise, which are described below.

\paragraph{General circuit noise.}
Let $\delta \in [0, 1]$ be an error rate. When general circuit noise with an error rate $\delta$ acts on a quantum circuit $\Phi$, we have that
\begin{enumerate}
    \item Purely classical elements of $\Phi$, i.e. the gates with only classical input-output, are implemented perfectly without any error.

    \item  A gate $g \in \mathbf{A}$ with the hybrid classical-quantum input as in Eq.~(\ref{eq:realization-qgate}) is replaced  by an arbitrary $\tilde{\cT}$ such that,
    \begin{equation} \label{eq:realization-qgate-noisy}
      \tilde{\cT}_g = \sum_{\mathbf{u} \in \{0, 1\}^l } \ket{\mathbf{u}}\bra{\mathbf{u}}(\cdot) \ket{\mathbf{u}}\bra{\mathbf{u}} \otimes \tilde{\cT}^\mathbf{u},
\end{equation}
where for all $\mathbf{u} \in \{0, 1\}^{l} $, $\tilde{\cT}^\mathbf{u}$ is an arbitrary quantum channel with the same input and output space as $\cT^\mathbf{u}$ (in particular, classical registers remain classical), and  $\dnorm{\Tilde{\cT}^\mathbf{u} - \cT^\mathbf{u} } \leq \delta$.  Note that a noisy idle gate realizes a channel $\cT_g$, such that $\dnorm{\cT_g - \cI} \leq \delta$.     
\end{enumerate}

We note that the above noise model is general as it includes, as special case, incoherent noise models such as depolarizing, amplitude and phase damping~\cite{tomita2014low} as well the coherent noise, that is, a small unknown unitary rotation~\cite{bravyi2018correcting,greenbaum2017modeling}.

\smallskip The quantum circuit $\Phi$, with the noise rate $\delta$, realizes a quantum channel in the following set,
\begin{multline}
    \trans(\Phi, \delta) := \Big\{ \tilde{\cT}_d \circ \cdots \circ \tilde{\cT}_1 :  \tilde{\cT}_i =  \tr_{\Delta'_i \setminus \Delta_i} \circ \Big[\tilde{\cT}_{g_{i, 1}}[A'_{i1}; A_{i1}] \otimes \cdots \\
  \otimes \tilde{\cT}_{g_{i, s_i}}[A'_{is_i}; A_{is_i}] \Big],
     \text{ with $\tilde{\cT}_{g_{i,j}}$ being a $\delta$-noisy version of $\cT_{g_{i,j}}$ as in Eq.~(\ref{eq:realization-qgate-noisy}). } \Big\}.
\end{multline}

\paragraph{Circuit-level stochastic noise.} We will also consider a stochastic noise model with parameter $\delta$ as follows: in Eq.~(\ref{eq:realization-qgate-noisy}), we take $\tilde{\cT}_u$ as follows,  
\begin{equation} \label{eq:simple-noise}
    \tilde{\cT}^{\mathbf{u}} = (1 - \delta) \cT^{\mathbf{u}} + \delta \: \cZ^{\mathbf{u}},
\end{equation}
for some arbitrary channel $\cZ$, with the same input and output systems as $\cT^{\mathbf{u}}$. Therefore, we  get
\begin{equation} \label{eq:noisy-g}
    \tilde{\cT}_g =  (1 - \delta) \cT_g + \delta \: \cZ_g,
\end{equation}
for some quantum channel $\cZ_g$ with the same input and output systems as $\cT_g$.
\begin{remark} \label{rem:noisy-g}
From Eq.~(\ref{eq:noisy-g}), note that $\dnorm{\tilde{\cT}_g - \cT_g} \leq 2\delta$. Therefore, the channel $\tilde{\cT}_{\Phi}$ (under circuit-level stochastic noise with parameter $\delta$) belongs to $\trans(\Phi, 2\delta)$ for any circuit $\Phi$. 
\end{remark}

\subsection{Fault-tolerant scheme} \label{sec:ft-implementation}

To implement quantum circuits reliably in the presence of noise, we consider a fault-tolerant scheme as given in Def.~\ref{def:ft_scheme}. We provide an explicit construction of the fault-tolerant scheme according to Def.~\ref{def:ft_scheme} in Appendix~\ref{sec:cons-ft-scheme}.

\begin{definition}[Fault-tolerant scheme~\cite{kitaev}\footnote{It is same as the \textit{polynomial error correction} in~\cite{kitaev}.}] \label{def:ft_scheme}
For some constant $c > 1$, let $\delta < \frac{1}{c}$ be a fixed noise rate, considering general circuit noise on the elements of a gate set $\mathbf{A}$. A fault-tolerant scheme for the gate set $\mathbf{A}$ is a collection of the following objects, 
\begin{enumerate}
\item[$(a)$] A sequence of codes $D_k = (L_k, \mu_k), \: k = 1, 2, \dots$ of type $(n_k, 1)$, that is, encoding $1$ logical qubit in $n_k$ physical qubits, where $n_k$ scales polynomially with $k$. Here, $L_k = U_k(\mathbb{C}^2 \otimes F_k) \subseteq N_k$, where $U_k$ is a unitary, $F_k \subseteq (\mathbb{C}^2)^{n_k - 1}, N_k := (\mathbb{C}^2)^{ \otimes n_k}$ and $\mu_k(\cdot) := U_k ( \cdot \otimes \ident_{F}) U_k^{\dagger}$.  The code $D_1$ is the trivial code, i.e. a code of type (1, 1), encoding one logical qubit in one physical qubit,  with $F_1$ being a one-dimensional Hilbert space and $\mu_1$ being the identity channel on $\bL(\mathbb{C}^2)$.

\item[$(b)$] For every $g \in \mathbf{A}$, a   family of circuits $(\Psi_{g,k}, k= 1,2, \dots)$, where the size of $\Psi_{g, k}$ scales  polynomially in $k$, such that $\tilde{\cT}_{\Psi_{g, k}} \in \trans(\Psi_{g, k}, \delta)$ represents $\cT_g$ in the code $D_k$, with accuracy $\epsilon_1 = \epsilon_2 = \epsilon_3 = (c\delta)^k$.

\item[$(c)$] A family of quantum circuits $(\Gamma_{k,l}, k, l = 1, 2, \dots, \text{ and } k \neq l )$ in  $\mathbf{A}$, referred to as interfaces, where the size of $\Gamma_{k,l}$ scales polynomially with $k, l$, such that any  $\tilde{\cT}_{\Gamma_{k,l}} \in \trans(\Gamma_{k,l}, \delta)$ represents the identity channel in the codes $D_k, D_l$ with accuracy $ \epsilon_1 = \epsilon_2 = (c \delta)^{\min\{k, l\}}$ and $\epsilon_3 = (c \delta)^{l}$.
\end{enumerate}
\end{definition}
For the sake of clarity, we specify systems $M, L, N$ and $M', L', N'$ in the commutative diagram, given in Fig.~\ref{fig:strong-rep}), for parts $(b)$ and $(c)$ of Def.~\ref{def:ft_scheme}.

\begin{enumerate}
    \item[$(1)$] In part $(b)$, if $g$ is a unitary gate, acting on $q$ qubits, then we have  $M = M' = (\mathbb{C}^2)^{\otimes q}$,  $L = L' = L_k^{\otimes q}$, and $N = N' = N_k^{\otimes q}$. 
    
    \item[$(2)$] In part (b), if $g$ is a state preparation gate, with $b$ bit input and $q$ qubit output, $\bL(M) = \bL(L) = \bL(N) =  \mathbb{B}^{\otimes b}$, and $M' = (\mathbb{C}^2)^{\otimes q}, L' = L_k^{\otimes q}, N' = N_k^{\otimes q}$. By a slight abuse of notation, let $D_1$ also be the trivial code for a classical bit $\mathbb{B}$, encoding a classical bit to a classical bit. Then, in fact, $\tilde{\cT}_{\Psi_{g, k}}$ represents $\cT_g$ in the code $D_1, D_k$, with accuracy $\epsilon_1 = \epsilon_2 = \epsilon_3 = (c\delta)^k$.
    
    \item[$(3)$]  In part (b), if $g$ is a measurement gate, with $q$ qubit input and $b$ bit classical output, $M = (\mathbb{C}^2)^{\otimes q}, L = L_k^{\otimes q}, N = N_k^{\otimes q}$, and $\bL(M') = \bL(L') = \bL(N') =  \mathbb{B}^{\otimes b}$. Similarly to point $(2)$ above, $\tilde{\cT}_{\Psi_{g, k}}$ represents $\cT_g$ in the code $D_k, D_1$, with accuracy $\epsilon_1 = \epsilon_2 = \epsilon_3 = (c\delta)^k$.

    \item[$(4)$] In part (c), for the interface $\Gamma_{k, l}$, we have $M = \mathbb{C}^2, L = L_k, N = N_k$, and $M' = \mathbb{C}^2, L' = L_l, N' = N_l$.
\end{enumerate}

\begin{remark} \label{rem:controlled-unit}
   A classical controlled unitary gate $g$, acting on $b$-bits and $q$-qubits, corresponds to a collection of unitary gates $\{g_\mathbf{u}\}_{\mathbf{u} \in \{ 0, 1\}^{b}}$. The corresponding quantum channel $\trans(g): \mathbb{B}^{ \otimes b} \otimes \bL((\mathbb{C}^2)^{\otimes q}) \to \mathbb{B}^{ \otimes b} \otimes \bL((\mathbb{C}^2)^{\otimes q})$ is as follows,
\begin{equation}
    \trans(g) = \sum_{\mathbf{u} \in \{0, 1\}^b}  \proj{\mathbf{u}}(\cdot) \proj{\mathbf{u}}\otimes \trans(g_\mathbf{u}).
\end{equation}

The corresponding $\Psi_{g, k}$ is given by the collection $\{\Psi_{g_{\mathbf{u}}, k} \}_\mathbf{u}$. Furthermore, the quantum channel $\trans(\Psi_{g, k}) : \mathbb{B}^{ \otimes b} \otimes \bL(N_k^{\otimes q}) \to \mathbb{B}^{ \otimes b} \otimes \bL(N_k^{\otimes q})$ is as follows,
\begin{equation}
    \trans(\Psi_{g, k}) = \sum_{\mathbf{u} \in \{0, 1\}^b} \proj{\mathbf{u}} (\cdot) \proj{\mathbf{u}}\otimes \trans(\Psi_{g_{\mathbf{u}}, k}).
\end{equation}
\end{remark}

\paragraph{Fault-tolerant scheme for circuit-level stochastic noise.} For circuit-level stochastic noise model given in Section~\ref{sec:noise-qcirc}, we consider the following definition of fault-tolerant scheme, where exact representation is used in contrast to approximate representation used in Def.~\ref{def:ft_scheme}. We provide an explicit construction of the fault-tolerant scheme according to Def.~\ref{def:ft_scheme} in Appendix~\ref{sec:cons-ft-scheme}.

\begin{definition}[Fault-tolerant scheme for circuit-level stochastic noise] \label{def:ft_scheme_stoc}
 For some constant $c > 1$, let $\delta < \frac{1}{c}$ be a fixed noise rate, considering circuit-level stochastic noise on the elements of a gate set $\mathbf{A}$. A fault-tolerant scheme for $\mathbf{A}$ is a collection of the following objects, 

\begin{enumerate}
    \item[$(a)$] A sequence of codes $D_k = (L_k, \mu_k), \: k = 1, 2, \dots$ of type $(n_k, 1)$ as in Part $(a)$ of Def.~\ref{def:ft_scheme}.

    \item[$(b)$]  For every $g \in \mathbf{A}$, a   family of circuits $(\Psi_{g,k}, k= 1,2, \dots)$, such that the noisy realization $\tilde{\cT}_{\Psi_{g, k}}$  is given by,
    \begin{equation}
      \tilde{\cT}_{\Psi_{g, k}} = (1 - \epsilon) \: \overline{\cT}_{\Psi_{g, k}}   + \epsilon \: \overline{\cZ}_{\Psi_{g, k}},
    \end{equation}
    where $\epsilon \leq (c\delta)^k$,  $\overline{\cT}_{\Psi_{g, k}}$ represents $\cT_g$ in the code $D_k$ and $\overline{\cZ}_{\Psi_{g, k}}$ is some quantum channel.

    \item[$(c)$] A family of quantum circuits $(\Gamma_{k,l}, k, l = 1, 2, \dots, \text{ and } k \neq l )$ in  $\mathbf{A}$, such that the noisy realization $\tilde{\cT}_{\Gamma_{k,l}}$ is given by,
    \begin{equation}
      \tilde{\cT}_{\Gamma_{k,l}} = (1 - \epsilon) \: \overline{\cT}_{\Gamma_{k,l}}   + \epsilon \: \overline{\cZ}_{\Gamma_{k,l}},   
    \end{equation}
    where $\epsilon \leq (c\delta)^l$,  and $\overline{\cT}_{\Gamma_{k,l}}$ is such that $\cJ_l^* \circ \overline{\cT}_{\Gamma_{k,l}}$ is a quantum channel, and  
     \begin{equation}
      \overline{\cT}_{\Gamma_{k,l}} = (1 - \epsilon') \: \overline{\cT}'_{\Gamma_{k,l}}   + \epsilon' \: \overline{\cZ}'_{\Gamma_{k,l}},   
    \end{equation}
  where $\epsilon' \leq (c \delta)^{\min\{k, l\}}$ and $\overline{\cT}'_{\Gamma_{k,l}}$ represents identity in the code  $D_k,D_l$.  Here, $\overline{\cZ}_{\Gamma_{k,l}}, \overline{\cZ}'_{\Gamma_{k,l}}$ are some quantum channels. 
\end{enumerate}
\end{definition}

\section{Transformation rules for logic gates} \label{sec:ft-tol-scheme}

This section sets the stage for the next section, where the fault-tolerant realizations of quantum circuits with quantum input/output is given. In this section, we obtain transformation rules, using the fault-tolerant scheme in Def.~\ref{def:ft_scheme}, analogous to the ones given in~\cite{Aliferis2006quantum, christandl2022fault}. These transformation rules will be later useful in Section~\ref{sec:ft-real} for fault-tolerant realizations of quantum circuits.

%In the second subsection (Sections~\ref{sec:cons-ft-scheme}), we construct fault-tolerant schemes according to Defs.~\ref{def:ft_scheme} based on the results from~\cite{aharonov1997fault, Aliferis2006quantum}; therefore, connecting Kitaev's framework with the other well-known constructions in the literature. It is worth emphasizing that Section~\ref{sec:cons-ft-scheme} is included for the sake of completeness.  While it can be of independent interest, the next sections will be referring to Defs.~\ref{def:ft_scheme} and \ref{def:ft_scheme_stoc}, instead of the constructions presented; therefore, the reader may skip Sections~\ref{sec:cons-ft-scheme}.

\subsection{Transformation rules} \label{sec:approx-to-exact}

In the fault-tolerant construction of~\cite{Aliferis2006quantum}, the transformation rules involving data and syndrome spaces are crucial for the realization of quantum circuits fault-tolerantly (see also~\cite[Lemma II.6]{christandl2022fault}). In this section, we obtain analogous relations in Corollaries~\ref{cor:tprime-unitary},~\ref{cor:tprime-stprep}, and~\ref{cor:tprime-measurement} for Kitaev's framework. To do so, we need an exact representation in codes rather than an approximate representation, as given in the fault-tolerant scheme of Def.~\ref{def:ft_scheme}. Below, we state Lemma 10.3 from~\cite{kitaev}, which gives two ways of obtaining exact commutativity in Fig.~\ref{fig:strong-rep} from  an approximate representation.

\begin{lemma} \normalfont{(\cite[Lemma 10.3]{kitaev})} \label{lem:rep-approx}
Let $\cP: \bL(M) \to \bL(M')$ and $\tilde{\cT}: \bL(N) \to \bL(N')$ be quantum channels and let $D = (L, \mu)$ and $D' = (L', \mu')$ be codes of type $(N, M)$ and $(N', M')$, respectively.

\begin{enumerate}
    \item [(a)] $\tilde{\cT}$ represents $\cP$ in codes $D, D'$ with accuracy $\epsilon_1 = \epsilon_2 = \epsilon_3 = O(\epsilon)$ if and only if there exist superoperators $\cT$, $\cR$, and $\cS$ (not necessarily channels) such that commutative diagram in Fig.~\ref{fig:strong-rep} holds exactly and the superoperator $\cT$ satisfies $\dnorm{\cT - \tilde{\cT}} \leq O(\epsilon).$ 

    \item[(b)] If $\tilde{\cT}$ represents $\cP$ in codes $D, D'$ with accuracy $\epsilon_1 = \epsilon_2 = \epsilon_3 = O(\epsilon)$, then there exists a channel $\cT$ representing $\cP$ such that $\dnorm{\cT - \tilde{\cT}} \leq |M'| O(\sqrt{\epsilon})$.
\end{enumerate}
\end{lemma}

Here we consider part (b) of Lemma~\ref{lem:rep-approx} (however, we note that part (a) is later used in Section~\ref{sec:cons-ft-scheme}). For the transformation rules, we only care about an exact commutativity in terms of the channel $\cR$ in Fig.~\ref{fig:strong-rep} (see Eq.~(\ref{eq:code-perf-mapp}) below). Note that part (b) of Lemma~\ref{lem:rep-approx} implies the following:

\smallskip For $\tilde{\cT}:\bL(N) \to \bL(N')$ and $\cP:\bL(M) \to \bL(M')$ according to part (b) of Lemma~\ref{lem:rep-approx}, there exists a channel $\cR:\bL(L) \to \bL(L')$, such that the left part of the commutative diagram of Figure~\ref{fig:strong-rep} commutes exactly for $\cR$ and $\cP$, that is,
\begin{equation} \label{eq:code-perf-mapp}
      {\mu'}^* \circ \cR = \cP \circ \mu^*, 
  \end{equation}
 and the right part of the diagram with $\cR$ and $\tilde{\cT}$ commutes approximately with accuracy $|M'| O(\sqrt{\epsilon})$, that is, 
\begin{equation} \label{eq:code-imperf-mapp}
    \dnorm{ \tilde{\cT} \circ \cJ -\cJ' \circ \cR} \leq |M'| O(\sqrt{\epsilon}).
\end{equation} 
In the following lemma, we strengthen Eq.~(\ref{eq:code-imperf-mapp}) by removing the dependence on $|M'|$. 
\begin{lemma} \label{thm:main-ft}
    Let $\cP: \bL(M) \to \bL(M')$ and $\tilde{\cT}: \bL(N) \to \bL(N')$ be quantum channels and let $D = (L, \mu)$ and $D' = (L', \mu')$ be codes of type $(N, M)$ and $(N', M')$, respectively. If  $\tilde{\cT}$ weakly represents (see Remark~\ref{rem:weak-rep}) $\cP$ with accuracy parameters $\epsilon_1 = \epsilon_2 = \epsilon$, then there exists a quantum channel $\cR:\bL(L) \to \bL(L')$ satisfying,
    \begin{equation} \label{eq:code-perf-map}
      {\mu'}^* \circ \cR = \cP \circ \mu^*, 
  \end{equation}
 and, 
\begin{equation} \label{eq:code-imperf-map}
    \dnorm{ \tilde{\cT} \circ \cJ -\cJ' \circ \cR} \leq 3 \sqrt{\epsilon}.
\end{equation} 
\end{lemma}

\begin{proof}
We will use $\mu^* = \tr_F \circ \: \cU^\dagger$, where $\cU^\dagger := U^\dagger(\cdot)U$, for the unitary $U: M \otimes F \to L$. Similarly, $(\mu')^* = \tr_{F'} \circ \: (\cU')^\dagger$ for the unitary $U': M' \otimes F' \to L$.
    
\smallskip As $\tilde{\cT}$ represents $\cP$ in  $D, D'$ with $\epsilon_1 = \epsilon_2 = \epsilon$, there exists a channel $\tilde{\cR}:\bL(L) \to \bL(L^\prime)$, satisfying 
\begin{align} 
    \dnorm{(\mu')^* \circ \tilde{\cR}   -  \cP \circ \mu^*} \leq \epsilon  \label{eq:strong-rep-qinout-1} \\
    \dnorm{ \tilde{\cT} \circ \cJ  - \cJ' \circ \tilde{\cR}} \leq \epsilon. \label{eq:strong-rep-qinout-2}
\end{align}
Using Theorem~\ref{thm:cont-stine}, there exist Stinespring dilations $V, W:  L \to M' \otimes G $ of $(\mu')^*\circ \tilde{\cR}$ and $\cP \circ \mu^*$, respectively such that,
\begin{equation} \label{eq:V-W-op}
    \opn{V - W} \leq \sqrt{\epsilon}.
\end{equation}
Let $X: L \to L' \otimes K$ be a Stinespring dilation of $\tilde{\cR}$.  Then, the following isometry is also a Stinespring dilation of $(\mu')^* \circ \tilde{\cR} = \tr_{F'} \circ \: (\cU')^\dagger \circ \tilde{\cR}$,
\begin{equation}
  ((U')^\dagger \otimes \ident_K) X  : L \to (M' \otimes F') \otimes K.  
\end{equation}
Therefore, there exists an isometry $Y:G \to F' \otimes K$ connecting the two purifications of  $(\mu')^* \circ \tilde{\cR}$, that is,
\begin{equation} \label{eq:X-V}
  ((U')^\dagger \otimes \ident_K) X = (\ident \otimes Y) V.   
\end{equation}
We now consider the isometry 
 \begin{equation} \label{eq:Z-def}
  Z : L \to  L' \otimes K, Z = (U' \otimes \ident_K)(\ident_{M'} \otimes Y)W.    
 \end{equation}
 We take $$\cR(\cdot) := \tr_K(Z (\cdot) Z^\dagger).$$ 
 As $(\ident_{M'} \otimes Y) W:  L \to (M' \otimes F') \otimes K$ is also a Stinespring dilation of $\cP \circ \mu^*$, we have 
\begin{align}
(\mu')^* \circ \tr_K(Z (\cdot)Z^\dagger) &=  \tr_{F'} \circ \tr_K \circ \big( (\ident \otimes Y) W (\cdot) W^\dagger (\ident \otimes Y^\dagger)\big) \nonumber\\
&= \cP \circ \mu^*.
\end{align}
Moreover, we have
 \begin{align}
 \opn{X  -  Z }  &=  \opn{(U' \otimes \ident_K)(\ident \otimes Y) V -  (U'  \otimes \ident_K)(\ident \otimes Y)W} \nonumber \\
  & = \opn{ V  -  W }  \leq  \sqrt{\epsilon}, \label{eq:X-Z}
 \end{align}
where the first equality uses Eqs.~(\ref{eq:X-V}), and (\ref{eq:Z-def}), and the second equality uses Eq.~(\ref{eq:V-W-op}).

\smallskip Now, using $ \tilde{\cR} =\tr_K(X(\cdot)X^\dagger)$, $\cR =\tr_K(Z(\cdot)Z^\dagger)$ and Theorem~\ref{thm:cont-stine}, we get 
\begin{equation} \label{eq:R-T'-e}
    \dnorm{ \tilde{\cR} - \cR } \leq 2 \sqrt{\epsilon}. 
\end{equation}
From Eq.~(\ref{eq:strong-rep-qinout-2}) and~(\ref{eq:R-T'-e}), we get 
\begin{align} \label{eq:T-R}
    \dnorm{\tilde{\cT} \circ \cJ -  \cJ' \circ \cR} \leq 3\sqrt{\epsilon}.
\end{align}
\end{proof}
We now consider $\cR: \bL(L) \to \bL(L')$, and $\cP: \bL(M) \to \bL(M')$ satisfying Eq.~(\ref{eq:code-perf-map}). We suppose that $\cP$ corresponds to a unitary, a pure state preparation, or a measurement gate (respectively, in Corollary~\ref{cor:tprime-unitary}, \ref{cor:tprime-stprep} and \ref{cor:tprime-measurement}), and we provide transformation rules for the quantum channel $\cR$, showing how it acts on systems $M$ and $F$ of $L = U (M \otimes F)$.
\begin{corollary}[Unitary gate] \label{cor:tprime-unitary}
  Let $\cP: \bL(M) \to \bL(M)$ be a unitary. Consider $\cR: \bL(L) \to \bL(L)$, satisfying Eq.~(\ref{eq:code-perf-map}). Then,  $\cR$ is as follows (see also Fig.~\ref{fig:Tprime-ind}):
\begin{equation}
    \cU^\dagger \circ \cR  =  (\cP \otimes \cF) \circ \cU^\dagger, 
\end{equation} 
where $\cU := U(\cdot)U^\dagger$, and $\cF: \bL(F) \to \bL(F)$ is some quantum channel. 
\end{corollary}

\begin{proof}
Using Eq.~(\ref{eq:code-perf-map}), we have that 
  \begin{equation}
      \tr_F \circ \: {\cU^\dagger} \circ \cR  = \cP \circ \tr_F \circ \: {\cU^\dagger} 
  \end{equation}
 Therefore, we have 
  \begin{align} 
      \tr_F \circ \: {\cU^\dagger} \circ \cR \circ \cU  & = \cP \circ \tr_F   \\
      & =  \tr_F  \circ (\cP \otimes \cI_{F}) \label{eq:P-unitary}
  \end{align}
As $\cP$ is a unitary channel, we have $\cP = P (\cdot) P^\dagger$ for a unitary $P$. From Eq.~\ref{eq:P-unitary}, it follows that $P \otimes \ident_F $ is a Stinespring dilation of the quantum channel $\tr_F \circ \: {\cU^\dagger} \circ \cR \circ \cU $.

\smallskip Let $X: L \to L \otimes G$ be a Stinespring dilation of the quantum channel $\cR$. It follows that $(U^\dagger \otimes \ident_G) X (U \otimes \ident_G)$ is also a Stinespring dilation of $ \tr_F \circ \:{\cU^\dagger} \circ \cR \circ \cU $. Therefore, there exists an isometry $Y: F \to F \otimes G$ such that 
\begin{equation} \label{eq:P-Y}
  P \otimes Y =  (U^\dagger \otimes \ident_G) X (U \otimes \ident_G).
\end{equation}
Taking the quantum channels corresponding to the isometries on both sides of Eq.~(\ref{eq:P-Y}) and then tracing out system $G$, gives
\begin{equation}
     {\cU^\dagger} \circ \cR \circ \cU  =   \cP \otimes \cF,
\end{equation}
where $\cF = \tr_G( Y (\cdot) Y^\dagger)$.
\end{proof}

\begin{figure}[!b]
\centering
\begin{tikzpicture}
\draw
(0, 0) node (a){}
(2, 0) node[draw, minimum width =   15mm, minimum height = 15mm] (b) {$\cR$}
(4, 0) node[draw, minimum width =   15mm, minimum height = 15mm] (c) {$\cU^\dagger$}
(6.0, 0) node[] (d){$=$}
;
\draw
(a) to (b.west)
(b.east) to  (c.west)
(c.east) to ++ (0.5, 0)
;
\draw
(6.4, 0) node (e){}
(6.4, 1) node (g){}
 (8.2, 0) node[draw, minimum width =   15mm, minimum height = 15mm] (f) {$\cU^\dagger$}
(10.2, 0.6) node[draw, minimum width =   8mm, minimum height = 8mm] (g) {$\cP$}
(10.2, -0.6) node[draw, minimum width =   8mm, minimum height = 8mm] (h) {$\cF$}
;
\draw
(e) to (f.west)
(g.west -| f.east) to (g.west)
(h.west -| f.east) to (h.west)
(g.east) to ++(0.5, 0)
(h.east) to ++(0.5, 0)
;
\end{tikzpicture}
\caption{ $\cR$ for a unitary $\cP$.}
\label{fig:Tprime-ind}
\end{figure}

\begin{corollary}[Preparation of a pure state] \label{cor:tprime-stprep}
Let $\cP: \emptyset \to \bL(L)$ be a state preparation channel, with no input and output in $\bL(L)$. Consider $\cR: \emptyset \to \bL(L)$, satisfying Eq.~(\ref{eq:code-perf-map}). Let $\cP(\cdot)$ denote the state prepared by $\cP$. If $\cP(\cdot) =  \proj{\psi}$, then the corresponding $\cR$ is as follows (see also Fig.~\ref{fig:Tprime-prep})
\begin{equation} \label{eq:pure-state-prep}
   {\cU^\dagger} \circ \cR(\cdot) =  \proj{\psi} \otimes \gamma,
\end{equation}
for some $\gamma \in \bD(F)$ and $\cU := U(\cdot)U^\dagger$.
\end{corollary}
\begin{proof}
From Theorem~\ref{thm:main-ft}, we have that
\begin{equation}
 {\mu^*} \circ \cR = \cP.
\end{equation}
Since $\cP$ prepares $\proj{\psi}$, this implies $\tr_{F}({\cU^\dagger} \circ \cR(\cdot)) =  \proj{\psi}$. Therefore, ${\cU^\dagger} \circ \cR(\cdot) = \proj{\psi} \otimes \gamma$.
\end{proof}
 
\begin{figure}[!t]
\centering
\begin{tikzpicture}
\draw
(0, 0) node (a){}
(2, 0) node[draw, minimum width =   15mm, minimum height = 15mm] (b) {$\cR$}
(4, 0) node[draw, minimum width =   15mm, minimum height = 15mm] (c) {$\cU^\dagger$}
(6.0, 0) node[] (d){$=$}
;
\draw
(b.east) to  (c.west)
(c.east) to ++ (0.5, 0)
;
\draw
(7.5, 0.6) node[draw, minimum width =   8mm, minimum height = 8mm] (g) {$\cP$}
(7.5, -0.8) node[above] (h) {$\gamma$}
;
\draw
(g.east) to ++(0.5, 0)
(h.east) to ++(0.5, 0)
(h.west) to ++(-0.5, 0)
;
\draw[double, double distance = 0.45mm]
(a) to (b.west)
(g.west) to ++(-0.5, 0)
;

\end{tikzpicture}
\caption{ $\cR$ for the state preparation channel $\cP$.}
\label{fig:Tprime-prep}
\end{figure}
Finally, we consider $\cP:\bL(M) \to \mathbb{B}_d$ is a measurement channel that outputs a $d$-dimensional classical system $\mathbb{B}_d$. As classical systems are encoded into the trivial code, Eq.~(\ref{eq:code-perf-map}) gives $\cR = \cP \circ \mu^*$, which implies the following.

\begin{corollary}\label{cor:tprime-measurement}
Let $M' = \mathbb{B}_d$ be a $d$-dimensional classical system.  Let $\cP:\bL(M) \to \mathbb{B}_d$ and $\cR:\bL(L) \to \mathbb{B}_d$ be measurement channels such that $\cR = \cP \circ \mu^*$. Then, $\cR$ is given by (see also Fig.~\ref{fig:R-measurement})
\begin{equation}
    \cR = (\cP  \otimes \tr_F) \circ \cU^\dagger,
\end{equation}
where $\cU := U(\cdot)U^\dagger$.
\end{corollary}
\begin{proof}
We have 
\begin{align*}
    \cR &= \cP \circ \mu^*  \\
        &= \cP \circ \tr_F \circ \: \cU^\dagger \\
        &= (\cP \otimes \tr_F) \circ \cU^\dagger.
\end{align*}
\end{proof}
    
\begin{figure}[!htbp]
\centering
\begin{tikzpicture}
\draw
(0, 0) node (a){}
(1, 0) node[draw, minimum width =   15mm, minimum height = 15mm] (c) {$\cR$}
(3.0, 0) node[] (d){$=$}
;
\draw
(c.east) to ++ (0.5, 0)
(c.west) to ++ (-0.5, 0)
;
\draw
(3.4, 0) node (e){}
(3.4, 1) node (g){}
(5.2, 0) node[draw, minimum width =   15mm, minimum height = 15mm] (f) {$\cU^\dagger$}
(7.2, 0.6) node[draw, minimum width =   8mm, minimum height = 8mm] (g) {$\cP$}
(7.2, -0.6) node[draw, minimum width =   8mm, minimum height = 8mm] (h) {$\tr_F$}
;
\draw
(e) to (f.west)
(g.west -| f.east) to (g.west)
(h.west -| f.east) to (h.west)
(g.east) to ++(0.5, 0)
;
\end{tikzpicture}
\caption{ $\cR$ for a measurement $\cP$.}
\label{fig:R-measurement}
\end{figure}

\section{ A Fault-tolerant realization of quantum input/output} \label{sec:ft-real}

\textbf{Notation.} 
For the quantum circuit $\Phi$, we use $|\Phi| := |\mathrm{Loc}(\Phi)|$, and  $\Phi_k$ denotes the circuit obtained by replacing each location $g_{i, j}$ in $\Phi$ by the corresponding $\Psi_{g_{i, j}, k}$ from Def.~\ref{def:ft_scheme} or Def.~\ref{def:ft_scheme_stoc}. We use $\cT_\Phi$ to denote the quantum channel realized by $\Phi$, and $\tilde{\cT}_\Phi$ to denote an arbitrary element in $\trans(\Phi, \delta)$.

\medskip In the remark below, we note that a quantum circuit containing classical control gates can be seen as collection of quantum circuits that do not contain any classical control gates. Therefore, for the sake of clarity, from now on, we choose to work with quantum circuits, without any classical-controlled gate. Whatever we do in the following sections can be easily generalized to include classical control unitary gates, using Eq.~(\ref{eq:con-uni}).

\begin{remark} \label{rem:no-bit-input}
Consider a circuit $\Phi$, containing  classical control gates. A classical control gate $g$, with $b$ control bits, corresponds to a collection of gates $\{g_\mathbf{u}\}_{\mathbf{u} \in \{ 0, 1\}^{b}}$ (see also Remark~\ref{rem:controlled-unit}), without classical control. Therefore, $\Phi$ corresponds to a collection of quantum circuits $\{\Phi^{\mathbf{v}}\}_\mathbf{v}$, without any classical control. Here, $\mathbf{v} \in \{ 0, 1\}^t$, where $t$ is the number of all control bits in $\Phi$.    The quantum channel corresponding to $\Phi$ is then given by,
\begin{equation} \label{eq:con-uni}
    \cT_\Phi = \sum_{\mathbf{v}}  \proj{\mathbf{v}} (\cdot) \proj{\mathbf{v}}\otimes  \cT_{\Phi^\mathbf{v}}.
\end{equation}
\end{remark}

\subsection{Fault-tolerant realization for general circuit noise} \label{sec:qin-out-gen}
In this section, we consider general circuit noise, and provide fault-tolerant realizations of quantum circuits with quantum input and output, quantum input and classical output, and classical input and quantum output.  Throughout this section, we shall consider constants $c , \delta$ from Def.~\ref{def:ft_scheme}.
\subsubsection{Quantum circuits with quantum input and output} \label{sec:q-in-out}
Below, we will first prove Lemma~\ref{lem:Phi-code-imp} for the quantum circuit $\Phi_k$, Lemma~\ref{lem:interface-1} for the interface $\Gamma_{1, k}$, and Lemma~\ref{lem:interface-2} for the interface $\Gamma_{k, 1}$. We will then combine them together to prove Theorem~\ref{thm:main-ft2}, regarding fault-tolerant realization of quantum input/output.

\begin{lemma} \label{lem:Phi-code-imp}
Consider a quantum circuit $\Phi$,  with $n$ qubit input, $n'$ qubit output, and composed of  unitary, state preparation, and quantum measurement gates. For an integer $k$, consider the code $D_k = (L_k, \mu_k)$, $L_k = U_k(\mathbb{C}^2 \otimes F_k) \subseteq N_k$, and the circuits $\Psi_{g, k}$ for each $g \in \Phi$ according to Def.~\ref{def:ft_scheme}. Then, for the encoded circuit $\Phi_k$, any $\tilde{\cT}_{\Phi_k} \in \trans(\Phi_k, \delta)$ satisfies
\begin{equation} \label{eq:tilde-R-1}
    \dnorm{ \tilde{\cT}_{\Phi_k} \circ \cJ_k^{\otimes n}  - \cJ_k^{\otimes n'} \circ \cU_k^{\otimes n'} \circ (\cT_{\Phi} \otimes \cF) \circ {\cU_k^\dagger}^{\otimes n} } \leq 3 |\Phi| \sqrt{ (c \delta)^k}, 
\end{equation}
where $\cU_k := U_k (\cdot) U_k^\dagger$,  $\cF: \bL(F_k^{\otimes n}) \to \bL(F_k^{\otimes n'})$ is a quantum channel, $\cJ_k: \bL(L_k) \to \bL(N_k)$ is the natural embedding, and $|\Phi|$ is the size of $\Phi$.
\end{lemma}

\begin{proof}
Consider the set of locations of $\Phi$, that is, $\mathrm{Loc}(\Phi)$. Recall that $\Phi_k$ is obtained by replacing $g_{i, j}$ by $\Psi_{g_{i, j}, k}$ for all $g_{i, j} \in \mathrm{Loc}(\Phi)$. By definition, any $\tilde{\cT}_{\Psi_{g_{i, j}, k}} \in \trans(\Psi_{g_{i, j},k}, \delta)$ represents  $\cT_{g_{i, j}}$ with accuracy $\epsilon_1 = \epsilon_2 = \epsilon_3 = (c \delta)^k$ in $D_k = (L_k, \mu_k)$. We take $\cR_{\Psi_{g_{i, j}, k}}$ according to Lemma~\ref{thm:main-ft}, satisfying,
   \begin{equation} \label{eq:code-perf-map-1}
      {\mu^*}^{\otimes q_{\mathrm{out}}} \circ \cR_{\Psi_{g_{i, j}, k}} = \cT_{g_{i,j}} \circ {\mu^*}^{\otimes q_{\mathrm{in}}}, 
  \end{equation}
 and furthermore, 
\begin{equation} \label{eq:code-imperf-map-1}
    \dnorm{ \tilde{\cT}_{\Psi_{g_{i, j}}} \circ \cJ^{\otimes q_{\mathrm{in}}} - \cJ^{\otimes q_{\mathrm{out}}} \circ \cR_{\Psi_{g_{i, j}, k}}} \leq 3 \sqrt{(c\delta)^k},
\end{equation} 
where $q_{\mathrm{in}}$ and $q_{\mathrm{out}}$ are the number of input and output qubits of the gate $g_{i,j}$. 

We now define a quantum channel $\cR$ by composing $\cR_{\Psi_{g_{i, j}, k}}, \forall g_{i, j} \in \mathrm{Loc}(\Phi)$ in the structure of the circuit $\Phi$. Using Eq.~(\ref{eq:code-imperf-map-1}) and repeatedly applying the triangle inequality, we get,
\begin{align} \label{eq:tilde-R-3}
    \dnorm{ \tilde{\cT}_{\Phi_k} \circ \cJ_k^{\otimes n} -  \cJ_k^{\otimes n'} \circ \cR} \leq 3  |\Phi| \sqrt{ (c \delta)^k}. 
\end{align}
We will now show that $\cR   = \cU_k^{\otimes n'} \circ  (\cT_\Phi \otimes \cF) \circ {\cU_k^\dagger}^{\otimes n}$, for some quantum channel $\cF: \bL(F_k^{\otimes n}) \to \bL(F_k^{\otimes n'})$. Using Corollaries~\ref{cor:tprime-unitary},~\ref{cor:tprime-stprep} and~\ref{cor:tprime-measurement}, we can express $\cR_{\Psi_{g_{i, j}, k}}$ satisfying Eq.~(\ref{eq:code-perf-map-1}), as below
\begin{equation} \label{eq:r-tilde-prod}
    {\cU_k^\dagger}^{\otimes q_{\mathrm{out}}} \circ \cR_{\Psi_{g_{i, j}, k}}  =   (\cT_{g_{i, j}} \otimes \cF_{g_{i, j},k}) \circ {\cU_k^\dagger}^{\otimes q_{\mathrm{in}}}
\end{equation}
where $\cF_{g_{i, j}, k}:  \bL(F_k^{\otimes q_{\mathrm{in}}}) \to \bL(F_k^{\otimes q_{\mathrm{out}}})$ is a quantum channel. The following cases are in order.
\begin{enumerate}
    \item  When $g_{i, j}$ is a unitary, we have $q_{\mathrm{in}} = q_{\mathrm{out}}$ and Eq.~(\ref{eq:r-tilde-prod}) follows directly from Corollary~\ref{cor:tprime-unitary}.

    \item If $g_{i, j}$ is a state preparation gate with $b$ bits input ($q_{\mathrm{in}} = 0$), preparing a pure state for each pure classical input. From Corollary~\ref{cor:tprime-stprep}, $\cF_{g_{i, j}, k}:  \mathbb{B}^{\otimes b} \to \bL(F_k^{\otimes  q_{\mathrm{out}}})$ is a channel that prepares a state in $\bL(F_k^{\otimes q_{\mathrm{out}}})$. 
    
    \item If $g_{i, j}$ is a measurement gate, we have $q_{\mathrm{out}} = 0$. From Corollary~\ref{cor:tprime-measurement}, we have 
\begin{equation}
    \cR_{\Psi_{g_{i, j}, k}} 
     = (\cT_{g_{i, j}} \otimes (\tr_{F_k})^{\otimes q_{\mathrm{in}}})  \circ \:{\cU_k^\dagger}^{\otimes q_{\mathrm{in}}}.
\end{equation}
Therefore, we have 
\begin{equation} \label{eq:measurement-F}
    \cF_{g_{i, j},k}: \bL(F_k^{\otimes q_{\mathrm{in}}}) \to \mathbb{C}, \cF_{g_{i, j},k} = (\tr_{F_k})^{\otimes q_{\mathrm{in}}}
\end{equation}    
\end{enumerate}
Using Eq.~(\ref{eq:r-tilde-prod}), sequentially moving $\cU_k^\dagger$ through each $\cR_{\Psi_{g_{i, j}, k}}$, we get
\begin{align*}
    {\cU_k^\dagger}^{\otimes n'} \circ \cR 
    & = (\cT_\Phi \otimes  \cF) \circ  {\cU_k^\dagger}^{\otimes n},
\end{align*}
where $\cF :\bL(F_k^{\otimes n}) \to \bL(F_k^{\otimes n'})$ is a quantum channel obtained by composing channels $\cF_{g_{i, j}, k}$ in the structure of the quantum circuit $\Phi$.  Therefore, we have that $\cR   = \cU_k^{\otimes n'} \circ  (\cT_\Phi \otimes \cF) \circ {\cU_k^\dagger}^{\otimes n}$. Finally, substituting $\cR$ in Eq.~(\ref{eq:tilde-R-3}), we get  Eq.~(\ref{eq:tilde-R-1}). 
\end{proof}

\medskip
\begin{remark} \label{rem:stat-prep-measurement}
 Note that when $\Phi$ is a state preparation circuit, that is, a circuit with no quantum input $(n = 0)$, the corresponding $\cF: \bL(F_k^{\otimes n}) \to \bL(F_k^{\otimes n'})$ from Lemma~\ref{lem:Phi-code-imp} is also a state preparation map. When $\Phi$ is a measurement circuit, that is, a circuit with no quantum output $(n' = 0)$, from Eq.~(\ref{eq:measurement-F}), it follows that  
$\cF = (\tr_{F_k})^{\otimes n}$.
\end{remark}

\medskip We now give Lemmas~\ref{lem:interface-1} and~\ref{lem:interface-2} for interfaces $\Gamma_{1, k}$ and $\Gamma_{k, 1}$ from Def.~\ref{def:ft_scheme}.
\begin{lemma} \label{lem:interface-1}
 For any $\tilde{\cT}_{\Gamma_{1, k}} \in \trans(\Gamma_{1, k}, \delta)$, there exists a quantum channel $\cS_{\Gamma_{1, k}} : \bL(\mathbb{C}^{2}) \to \bL(L_k)$, such that
    \begin{align} \label{eq:comm-1k-3}
        \dnorm{ \cJ_k \circ \cS_{\Gamma_{1, k}}  - \tilde{\cT}_{\Gamma_{1, k}}} & \leq (c \delta)^k,
    \end{align}
and, 
    \begin{equation} \label{eq:u_k-s}
    \dnorm{\tr_{F_k} \circ \: \cU_k^\dagger \circ \cS_{\Gamma_{1, k}}  - \cI}  \leq 3c \delta.
    \end{equation}
\end{lemma}
\begin{proof}
 We know that any $\tilde{\cT}_{\Gamma_{1, k}} \in \trans(\Gamma_{1, k}, \delta)$ represents the identity channel in codes $D_1, D_k$ (Recall that $D_1$ is the trivial code), with accuracy $\epsilon_1 = \epsilon_2 = c\delta$, and $\epsilon_3 = (c\delta)^k$.  Here we have $M = L = N = \mathbb{C}^2$ and $M' = \mathbb{C}^2, L' = L_k, N' = N_k$ in Fig.~\ref{fig:strong-rep}. Therefore, there exists $\cS_{\Gamma_{1, k}}: \bL(\mathbb{C}^2) \to \bL(L_k)$ such that (see Eq.~(\ref{eq:comm-3-del}))
\begin{align}  \label{eq:comm-1k-3-inproof}
\dnorm{ \cJ_k \circ \cS_{\Gamma_{1, k}}  - \tilde{\cT}_{\Gamma_{1, k}}} & \leq (c \delta)^k. 
\end{align}
Furthermore, there exists a quantum channel  $\cR_{\Gamma_{1, k}}: \bL(\mathbb{C}^2) \to \bL(L_k)$, such that (see Eqs.~(\ref{eq:comm-2-del}) and (\ref{eq:comm-3-del}))
\begin{align}
 \dnorm{\mu^*_k \circ \cR_{\Gamma_{1, k}} - \cI} & \leq c \delta. \label{eq:comm-1k-1} \\
 \dnorm{\cJ_k \circ \cR_{\Gamma_{1, k}} - \tilde{\cT}_{\Gamma_{1, k}}} & \leq c \delta. \label{eq:comm-1k-2}   
\end{align}
Eqs.~(\ref{eq:comm-1k-3-inproof}) and~(\ref{eq:comm-1k-2}) imply,
\begin{equation} \label{eq:R-S}
   \dnorm{\cR_{\Gamma_{1, k}} - \cS_{\Gamma_{1, k}}} \leq 2c\delta. 
\end{equation}
From Eqs.~(\ref{eq:comm-1k-1}) and~(\ref{eq:R-S}), we get
\begin{equation} \label{eq:R'_1-id}
    \dnorm{\mu^*_k \circ \cS_{\Gamma_{1, k}}  - \cI}  \leq 3c \delta,
\end{equation}
which is the same as Eq.~(\ref{eq:u_k-s}), by using $\mu^*_k = \tr_{F_k} \circ \: \cU_k^\dagger$.
\end{proof}

\medskip
\begin{lemma} \label{lem:interface-2}
Any $\tilde{\cT}_{\Gamma_{k, 1}} \in \trans(\Gamma_{k, 1}, \delta)$ satisfies,
\begin{equation} \label{eq:Q-J-mu}
   \dnorm{\tilde{\cT}_{\Gamma_{k, 1}} \circ \cJ_k \circ \cU_k - \cI \otimes \tr_{F_k}} \leq 2c\delta. 
\end{equation}
\end{lemma}
\begin{proof}
    We know that any $\tilde{\cT}_{\Gamma_{k, 1}} \in \trans(\Gamma_{k, 1}, \delta)$ represents the identity channel in codes $D_k, D_1$ with accuracy $\epsilon_1 = \epsilon_2 = c\delta$. Here we have $M = \mathbb{C}^2, L = L_k, N = N_k$, and $M' = L' = N' = \mathbb{C}^2$ in Fig.~\ref{fig:strong-rep}. Therefore, for any $\tilde{\cT}_{\Gamma_{k, 1}} \in \trans(\Gamma_{k, 1}, \delta)$, there exists a channel $\cR_{\Gamma_{k, 1}}: \bL(L_k) \to \bL(\mathbb{C}^2)$, such that
\begin{align}
    \dnorm{\cR_{\Gamma_{k, 1}} -  \mu^*_k} \leq c\delta, \label{eq:R-2-mu} \\
    \dnorm{\cR_{\Gamma_{k, 1}} -  \tilde{\cT}_{\Gamma_{k, 1}} \circ \cJ_k} \leq c\delta,
\end{align}
which implies the following,
\begin{equation} 
   \dnorm{\tilde{\cT}_{\Gamma_{k, 1}} \circ \cJ_k - \mu^*_k} \leq 2c\delta, 
\end{equation}
from which we get Eq.~(\ref{eq:Q-J-mu}) by using $\mu^*_k = \tr_{F_k} \circ \: \cU_k^\dagger$. 
\end{proof}

\medskip Using Lemma~\ref{lem:Phi-code-imp}, \ref{lem:interface-1}, and \ref{lem:interface-2}, we now prove the following theorem, concerning the fault-tolerant realization of quantum circuits with quantum input and output (see also Fig.~\ref{fig:ft-picture-intro}).

\begin{theorem} \label{thm:main-ft2}
Consider a quantum circuit $\Phi$,  with $n$ qubit input, $n'$ qubit output, and composed of unitary, state preparation, and measurement gates. For any integer $k$, there exists another circuit $\overline{\Phi}$ of size $|\Phi| \cdot \mathrm{poly}(k)$, with the same input-output systems as $\Phi$, such that for any $\tilde{\cT}_{\overline{\Phi}} \in \trans(\overline{\Phi}, \delta)$, we have  
  \begin{equation} \label{eq:ft-io}
      \dnorm{\tilde{\cT}_{\overline{\Phi}}  - (\otimes_{i = 1}^{n'} \cW_i) \circ (\cT_{\Phi} \otimes \cF) \circ (\otimes_{i = 1}^n \cN_i) } \leq 4 |\Phi| \sqrt{(c \delta)^k},
  \end{equation}
  for some channels $\cN_i : \bL(\mathbb{C}^2) \to \bL(\mathbb{C}^2 \otimes F_k), i \in [n]$, $\cF: \bL(F_k^{\otimes n}) \to \bL(F_k^{\otimes n'})$, and $\cW_i: \bL(\mathbb{C}^2 \otimes F_k) \to \bL(\mathbb{C}^2), i \in [n']$, with $F_k$ being an auxiliary Hilbert space. The channels $\cN_i, \forall i \in [n]$ and $\cW_i, \forall i \in [n']$ satisfy,
  \begin{align}
    \dnorm{\tr_{F_k} \circ\: \cN_i - \cI_{\mathbb{C}^2}} \leq 3 c\delta. \label{eq:N_i}\\
    \dnorm{ \cW_i - \cI_{\mathbb{C}^2} \otimes \tr_{F_k}} \leq 2 c\delta. \label{eq:W_i}
\end{align}
\end{theorem}

\begin{proof}
Consider the circuit $ \Phi_k$ as in Lemma~\ref{lem:Phi-code-imp}. Consider interfaces $\Gamma_{1,k}$ and $\Gamma_{k,1}$ from Def.~\ref{def:ft_scheme}. We take the quantum circuit $\overline{\Phi}$ 
as follows,
\begin{equation} \label{eq:circuit-Phi-prime}
    \overline{\Phi} := (\otimes_{i=1}^{n'} \Gamma^i_{k, 1}) \circ \Phi_k \circ (\otimes_{i=1}^{n} \Gamma^i_{1, k}).
\end{equation}
By construction, $\overline{\Phi}$ has the same number of input and output qubits as $\Phi$. Any $\tilde{\cT}_{\overline{\Phi}} \in \trans(\overline{\Phi}, \delta)$ is given by,
\begin{equation}
    \tilde{\cT}_{\overline{\Phi}}  = (\otimes_{i = 1}^{n'} \tilde{\cT}_{\Gamma^i_{k, 1}})  \circ \tilde{\cT}_{\Phi_k} \circ   (\otimes_{i = 1}^{n} \tilde{\cT}_{\Gamma^i_{1, k}}),
\end{equation}
where $\tilde{\cT}_{\Gamma^i_{1, k}} \in \trans(\Gamma^i_{1, k}, \delta)$, $\tilde{\cT}_{ \Phi_k} \in \trans(\Phi_k, \delta)$ and $\tilde{\cT}_{\Gamma^i_{k, 1}} \in \trans(\Gamma^i_{k, 1}, \delta)$. 

\medskip Consider now the quantum channel $\cF: \bL(F_k^{\otimes n}) \to \bL(F_k^{\otimes n'})$ corresponding to the circuit $\Phi_k$ from Lemma~\ref{lem:Phi-code-imp} and let $ \cR := \cU_k^{\otimes n'} \circ  (\cT_\Phi \otimes \cF) \circ {\cU_k^\dagger}^{\otimes n}$. Then, from Eq.~(\ref{eq:tilde-R-1}), we have
\begin{equation} \label{eq:int-t-r}
\dnorm{\tilde{\cT}_{\Phi_k} \circ \cJ_k^{\otimes n}  - \cJ_k^{\otimes n'} \circ \cR } \leq 3 |\Phi| \sqrt{ (c \delta)^k}.
\end{equation}
Consider $\cS_{\Gamma^i_{1, k}}$ corresponding to $\tilde{\cT}_{\Gamma^i_{1, k}}$ from Lemma~\ref{lem:interface-1}, and the following composed channel, which has the same input and output systems as $\tilde{\cT}_{\overline{\Phi}}$, 
\begin{equation}
  (\otimes_{i = 1}^{n'} (\tilde{\cT}_{\Gamma^i_{k, 1}}  \circ \cJ_k))  \circ \cR\circ (\otimes_{i = 1}^{n} \cS_{\Gamma^i_{1, k}}).  
\end{equation}

The following is a crucial step in the proof: by inserting identity $\cU_k^{\otimes n} \circ {\cU_k^\dagger}^{\otimes n}$ before and after $\cR$, we get 
\begin{align}
  &  (\otimes_{i = 1}^{n'} (\tilde{\cT}_{\Gamma^i_{k, 1}}  \circ \cJ_k))  \circ \cR\circ (\otimes_{i = 1}^{n} \cS_{\Gamma^i_{1, k}}) \nonumber \\
& =  (\otimes_{i = 1}^{n'} (\tilde{\cT}_{\Gamma^i_{k, 1}}  \circ \cJ_k \circ \cU_k)) \circ  ({\cU_k^\dagger}^{\otimes n'} \circ {\cR} \circ  \cU_k^{\otimes n}) \circ (\otimes_{i = 1}^{n} \cU_k^\dagger \circ \cS_{\Gamma^i_{1, k}}) \nonumber\\
& =  (\otimes_{i = 1}^{n'} \cW_i) \circ  (\cT_{\Phi} \otimes \cF)  \circ (\otimes_{i = 1}^{n} \cN_i), 
\end{align}
where in the last equality 
\begin{align}
 \cN_i &:= \cU_k^\dagger \circ \cS_{\Gamma^i_{1, k}}, \label{eq:n-i-def} \\
  \cW_i &:= \tilde{\cT}_{\Gamma^i_{k, 1}}  \circ \cJ_k \circ \cU_k. \label{eq:w-i-def}
\end{align}
From Lemmas~\ref{lem:interface-1} and~\ref{lem:interface-2}, we know that  $\cN_i$ and $\cW_i$ satisfy Eq.~(\ref{eq:N_i}) and Eq.~(\ref{eq:W_i}), respectively.  Finally, we have that
\begin{align}
 &   \dnorm { \tilde{\cT}_{\overline{\Phi}} -  (\otimes_{i = 1}^{n'} (\tilde{\cT}_{\Gamma^i_{k, 1}}  \circ \cJ_k))  \circ \cR\circ (\otimes_{i = 1}^{n} \cS_{\Gamma^i_{1, k}})} \nonumber\\
& = \dnorm { \tilde{\cT}_{\overline{\Phi}} -  (\otimes_{i = 1}^{n'} \tilde{\cT}_{\Gamma^i_{k, 1}})   \circ ( \cJ_k^{\otimes n'} \circ \cR )\circ (\otimes_{i = 1}^{n} \cS_{\Gamma^i_{1, k}})} \nonumber\\
& \leq \dnorm { \tilde{\cT}_{\overline{\Phi}} -  (\otimes_{i = 1}^{n'} \tilde{\cT}_{\Gamma^i_{k, 1}})   \circ (  \tilde{\cT}_{\Phi_k} \circ \cJ_k^{\otimes n}) \circ (\otimes_{i = 1}^{n} \cS_{\Gamma^i_{1, k}})} + 3 |\Phi| \sqrt{ (c \delta)^k} \nonumber\\
& \leq \dnorm { \otimes_{i = 1}^{n} \tilde{\cT}_{\Gamma^i_{1, k}} -  \otimes_{i = 1}^{n} \cJ_k \circ \cS_{\Gamma^i_{1, k}}} + 3 |\Phi| \sqrt{ (c \delta)^k} \nonumber\\
& \leq 4 |\Phi| \sqrt{ (c \delta)^k},
\end{align}
where the first inequality uses Eq.~(\ref{eq:int-t-r}), the second inequality uses monotonicity of the diamond norm under quantum channels, and the last inequality uses Eq.~(\ref{eq:comm-1k-3}), and $ |\Phi| \geq n$.
\end{proof}

\subsubsection{Quantum circuits with only quantum input or quantum output} \label{sec:q-c-in-out}

\medskip In this section, we consider circuits with quantum input and classical output, and classical input and quantum output. We provide their fault-tolerant realization, which is independent of auxiliary systems $F_k$'s. To characterize the noise channel channels appearing in fault-tolerant realization of these circuits, we consider the following definition of adversarial channels.

\begin{definition}[Adversarial channels] \label{def:adver-can}
   A quantum channel   $\cV: \bL(M^{\otimes n})\to \bL(M^{\otimes n})$ is said to be an adversarial channel with parameter $\delta$ if for any $t > 0$, there exists a superoperator $\cV_t$ of weight $t$, that is, $\cV_t \in \mathscr{E}(n, t) (\cdot) \mathscr{E}(n, t)^*$ (see Def.~\ref{def:weight}), such that 
   \begin{equation} \label{eq:d-v-vt}
    \dnorm{\cV - \cV_t} \leq  \sum_{j > t}^n \binom{n}{j} \delta^j.
\end{equation}
\end{definition}

\smallskip An adversarial channel can be approximated by a superoperator of weight $O(n \delta)$, as shown later in the proof of Theorem~\ref{thm:ft-comm-code-general}. In the following Remark, we give an example of adversarial channel according to the above definition. 
\begin{remark} \label{rem:ex-advers}
 Consider quantum channel $\cV:= (\otimes_{i = 1}^n \cV_i)$, where $\cV_i: \bL(M) \to \bL(M), \dnorm{\cV_i - \cI_i} \leq \delta, \forall i \in [n]$. It can be seen that $\cV$ is an adversarial channel according to Def.~\ref{def:adver-can} by considering the following expansion,
\begin{align}
(\otimes_{i = 1}^n \cV_i) &= \otimes_{i = 1}^n (\cI_i + (\cV_i - \cI_i))  \\
&= \sum_{A \subseteq [n] } (\otimes_{i \in [n]\setminus A} \cI_i) \otimes (\otimes_{i \in A} (\cV_i - \cI_i)),
\end{align}
and taking $\cV_t$ as follows,
\begin{equation}
    \cV_t := \sum_{A \subseteq [n]: |A| \leq t} (\otimes_{i \in [n]\setminus A} \cI_i) \otimes (\otimes_{i \in A} (\cV_i - \cI_i)).
\end{equation}
\end{remark}

\medskip The following corollary gives a fault-tolerant realization of quantum circuits with quantum input/classical output, up to a product channel as in Remark~\ref{rem:ex-advers}, acting on input qubits (see also Part (a) of Fig.~\ref{fig:ft-cl-qu-intro}). 
\begin{corollary} \label{cor:enc}
Consider the quantum circuit $\Phi$ from Theorem~\ref{thm:main-ft2}, with $n' = 0$. Then, there exists a quantum circuit $\overline{\Phi}$, with the same input-output systems as $\Phi$ and of size $|\Phi| \cdot \mathrm{poly}(k)$, such that for any $\tilde{\cT}_{\overline{\Phi}} \in \trans(\overline{\Phi}, \delta)$, we have 
\begin{equation} \label{eq:phi-bar-N}
    \dnorm{\tilde{\cT}_{\overline{\Phi}} -  \cT_{\Phi} \circ \overline{\cN}} \leq 4 |\Phi| \sqrt{(c\delta)^k},
\end{equation}
where $\overline{\cN}: \bL((\mathbb{C}^2)^{\otimes n}) \to \bL((\mathbb{C}^2)^{\otimes n})$ is  given by, $\overline{\cN} = (\otimes_{i = 1}^n \overline{\cN}_i)$, such that the $\overline{\cN}_i: \bL(\mathbb{C}^2) \to \bL(\mathbb{C}^2), i \in [n]$ satisfy
 \begin{equation} \label{eq:N_i-1}
   \dnorm{\overline{\cN}_i - \cI_{\mathbb{C}^2}} \leq 3c\delta.  
 \end{equation}
\end{corollary}

\begin{proof}
We consider the circuit $\overline{\Phi}$ as in Eq.~(\ref{eq:circuit-Phi-prime}). Since there is no quantum output, we have
\begin{equation} \label{eq:circuit-Phi-prime-1}
    \overline{\Phi} = \Phi_k \circ (\otimes_{i=1}^{n} \Gamma^i_{1, k}).
\end{equation}
As quantum circuit $\Phi$ has only classical output, the quantum channel $\cF$ in Lemma~\ref{lem:Phi-code-imp} for $\Phi_k$, corresponds to tracing out the system $F_k^{\otimes n}$ (see Eq.~(\ref{eq:measurement-F}) and also Remark~\ref{rem:stat-prep-measurement}). Substituting $\cF = (\tr_{F_k})^{\otimes n}$ and $n' = 0$ in Eq.~(\ref{eq:ft-io}), we get 
\begin{equation}
 \dnorm{ \cT_{\overline{\Phi}} - \cT_{\Phi}  \circ (\otimes_{i = 1}^n (\tr_{F_k} \circ \cN_i)) } \leq 4 |\Phi| \sqrt{(c \delta)^k}.
\end{equation}
We take 
\begin{equation} \label{eq:bar-N-i}
   \overline{\cN}_i := \tr_{F_k} \circ \cN_i. 
\end{equation}
Using Eq.~(\ref{eq:N_i}), we get Eq.~(\ref{eq:N_i-1}). Therefore, Eq.~(\ref{eq:phi-bar-N}) holds for $\overline{\cN} = \otimes_{i = 1}^n \overline{\cN}_i$, where $\overline{\cN}_i, \forall i \in [n]$ satisfies Eq.~(\ref{eq:N_i-1}).  
\end{proof}

\medskip The following corollary gives a fault-tolerant realization of quantum circuits with classical input/quantum output, up to an adversarial channel according to Def.~\ref{def:adver-can}, acting on output qubits (see also Part (b) of Fig.~\ref{fig:ft-cl-qu-intro}).

\begin{corollary} \label{cor:dec}
Consider the quantum circuit $\Phi$ from Theorem~\ref{thm:main-ft2}, with $n= 0$. Then, there exists a quantum circuit $\overline{\Phi}$, with the same input-output systems as $\Phi$ and of size $|\Phi| \cdot \mathrm{poly}(k)$, such that for any $\tilde{\cT}_{\overline{\Phi}} \in \trans(\overline{\Phi}, \delta)$, we have
\begin{equation} \label{eq:CQ}
    \dnorm{\tilde{\cT}_{\overline{\Phi}} -  \overline{\cW} \circ \cT_{\Phi}} \leq 4 |\Phi| \sqrt{(c\delta)^k},
\end{equation}
where $\overline{\cW}: \bL((\mathbb{C}^2)^{\otimes n'}) \to \bL((\mathbb{C}^2)^{\otimes n'})$ is an adversarial channel with parameter $2c\delta$ according to Def.~\ref{def:adver-can}.
\end{corollary}
\begin{proof}
We consider the circuit $\overline{\Phi}$ as in Eq.~(\ref{eq:circuit-Phi-prime}). Since there is no quantum input, we have
\begin{equation}
     \overline{\Phi} = (\otimes_{i=1}^{n'} \Gamma^i_{k, 1}) \circ \Phi_k.
  \end{equation}
For $\Phi_k$, the quantum channel $\cF$ in Lemma~\ref{lem:Phi-code-imp} is a state preparation channel, with output in $\bL(F_k^{\otimes n'})$. Let $\eta$ be the quantum state prepared by $\cF$ ($\eta$ can depend on the classical input). Then, substituting $\cF$ by $\eta$, and $n = 0$ in Eq.~(\ref{eq:ft-io}), we get
\begin{equation} \label{eq:bar-w-1}
  \dnorm{\tilde{\cT}_{\overline{\Phi}}  - (\otimes_{i = 1}^{n'} \cW_i) \circ (\cT_{\Phi} \otimes \eta) } \leq 4 |\Phi| \sqrt{(c \delta)^k}.  
\end{equation}
We take  
\begin{equation} \label{eq:bar-w}
    \overline{\cW} := (\otimes_{i = 1}^{n'} \cW_i) \big((\cdot) \otimes \eta \big),
\end{equation}
We now show that $\overline{\cW}$ is an adversarial channel, by showing that there exists a superoperator $\overline{\cW}_t$ of weight $t$ such that Eq.~(\ref{eq:d-v-vt}) is satisfied for $\overline{\cW}$ and $\overline{\cW}_t$. 

\medskip For a subset $A \subseteq [n'] = \{ 1, \dots, n'\}$, we define a superoperator $\cW_A : \bL((\mathbb{C}^2)^{\otimes |A|}) \to \bL((\mathbb{C}^2)^{\otimes |A|})$ as follows:
\begin{equation}
    \cW_A := (\otimes_{i = 1}^{|A|} (\cW_i - \cI_i \otimes \tr_{F_k^i})) \big(\cdot \otimes \: \eta_A \big),
\end{equation}
where $\eta_A := (\otimes_{i \in [n'] \setminus A } \tr_{F_k^i}) (\eta)$. Using Eq.~(\ref{eq:W_i}), we get
\begin{equation} \label{eq:W-A-norm}
\dnorm{\cW_A} \leq (2 c\delta)^{|A|}.
\end{equation}
We now consider the following expansion,
\begin{align} 
  \otimes_{i = 1}^{n'} \cW_i &=   \otimes_{i = 1}^{n'} (\cI_i \otimes \tr_{F_k^i} +  (\cW_i - \cI_i \otimes \tr_{F_k^i})) \nonumber\\
  &= \sum_{A \subseteq [n']} (\otimes_{i \in [n'] \setminus A} \: \cI_i \otimes \tr_{F_k^i} ) \otimes (\otimes_{i \in A} (\cW_i- \cI_i \otimes \tr_{F_k^i})) \label{eq:W-bar-exp-1}.
\end{align}
Note that
\begin{equation} \label{eq:W-bar-exp-2}
    \left((\otimes_{i \in [n'] \setminus A} \: \cI_i \otimes \tr_{F_k^i} ) \otimes (\otimes_{i \in A} (\cW_i- \cI_i \otimes \tr_{F_k^i}))\right) \left( \cdot \otimes \eta \right) =     \left((\otimes_{i \in [n'] \setminus A} \: \cI_i ) \otimes \cW_A \right).
\end{equation}
From Eqs.~(\ref{eq:W-bar-exp-1}) and~(\ref{eq:W-bar-exp-2}), we have
\begin{align}
   \overline{\cW} = \sum_{A \subseteq [n']} \left((\otimes_{i \in [n'] \setminus A} \: \cI_i ) \otimes \cW_A \right).
\end{align}
We take 
\begin{equation}
    \overline{\cW}_t := \sum_{A \subseteq [n']: |A| \leq t} \left((\otimes_{i \in [n'] \setminus A} \: \cI_i ) \otimes \cW_A \right).
\end{equation}
It's clear that $\overline{\cW}_t$ has weight $t$. Finally, we have
\begin{align}
    \dnorm{\overline{\cW} - \overline{\cW}_t} &= \dnorm{\sum_{A \subseteq [n']: |A| > t} \left((\otimes_{i \in [n'] \setminus A} \: \cI_i ) \otimes \cW_A \right)} \nonumber \\
    &\leq \sum_{A \subseteq [n']: |A| > t}  \dnorm{\cW_A} \nonumber \\
    &= \sum_{j > t} \binom{n}{j} (2c \delta)^j \label{eq:W-bar-advers},
\end{align}
where the second line uses the triangle inequality and the last line uses Eq.~(\ref{eq:W-A-norm}). Therefore, from Eq.~(\ref{eq:W-bar-advers}), $\overline{\cW}$ is an adversarial channel with parameter $2c\delta$.
\end{proof}

\subsection{Fault-tolerant realization for circuit-level stochastic noise}
In this section, we give fault-tolerant realizations of quantum circuits, considering circuit-level stochastic noise from Section~\ref{sec:noise-qcirc}. Throughout this section, we shall consider constants $c , \delta$ from Def.~\ref{def:ft_scheme_stoc}.

\smallskip From Remark~\ref{rem:noisy-g}, the channel $\tilde{\cT}_\Phi$ realized by a quantum circuit $\Phi$, with circuit-level stochastic noise of parameter $\delta/2$, is also in $\trans(\Phi, \delta)$. Therefore,  the results from Section~\ref{sec:qin-out-gen}, based on Def.~\ref{def:ft_scheme}, are also valid for circuit-level stochastic noise. However, Parts $(b)$ and $(c)$ in Def.~\ref{def:ft_scheme_stoc} are more explicit than Parts $(b)$ and $(c)$ in Def.~\ref{def:ft_scheme}. This allows us to make stronger statements than Section~\ref{sec:q-in-out}. In particular, there are two major improvements; firstly the error scaling in Eq.~(\ref{eq:tilde-R-1}) and Eq.~(\ref{eq:ft-io}) improves to $O((c\delta)^k)$ from $O(\sqrt{(c\delta)^k})$ (see Lemma~\ref{lem:Phi-code-imp-stoc} and Theorem~\ref{thm:main-ft2-stoc} below), and secondly, the channels $\overline{\cN}$ and $\overline{\cW}$ from Corollaries~\ref{cor:enc} and \ref{cor:dec}, respectively, are stochastic channels instead of adversarial ones (see Corollaries~\ref{cor:enc-stc} and \ref{cor:dec-stc} below). More precisely, the channel $\overline{\cN}$ is an independent stochastic channel (see Def.~\ref{def:stoc} below) and $\overline{\cW}$ is a local stochastic channel (see Def.~\ref{def:local-stoc} below).

\begin{definition}[Independent stochastic channels] \label{def:stoc}
    A quantum channel $\cV: \bL(M^{\otimes n}) \to \bL(M^{\otimes n})$ is said to be an independent stochastic channel with parameter $\delta$ if 
 \begin{equation}
 \cV = (\otimes_{i = 1}^n \cV_i),
\end{equation}
 for quantum channels $\cV_i: \bL(M) \to \bL(M), i \in [n]$ such that
    \begin{equation}
        \cV_i = (1- \epsilon) \: \cI_M + \epsilon \: \cZ_i,
    \end{equation}
    where $\epsilon \leq \delta$ and $\cZ_i:\bL(M) \to \bL(M)$ is some quantum channel.
\end{definition}

\medskip
\begin{definition}[Local stochastic channels~\cite{gottesman2013fault, fawzi2018efficient}] \label{def:local-stoc}
 A quantum channel $\cV: \bL(M^{\otimes n}) \to \bL(M^{\otimes n})$ is said to be  a local stochastic channel with parameter $\delta$ if 
\begin{equation} \label{eq:local-stc-1}
 \cV = \sum_{A \subseteq [n]} \big( \otimes_{i \in [n] \setminus A} \cI_i \big) \otimes \cV_A,  
\end{equation}
for completely positive maps $\cV_A: \bL(M^{\otimes |A|}) \to \bL(M^{\otimes |A|})$ that only act on the set $A \subseteq [n]$ and that satisfy the following: for every $T \subseteq [n]$,
\begin{equation} \label{eq:local-stc-2}
    \dnorm{\sum_{A : T \subseteq A} \big( \otimes_{i \in [n] \setminus A} \cI_i \big) \otimes \cV_A } \leq \delta^{|T|}.
\end{equation}
\end{definition}
We note that independent stochastic channels with parameter $\delta$ are also local stochastic channels with parameter $\delta$. In general local stochastic channels can generate correlation among the qubits. 

\subsubsection{Quantum circuits with quantum input and output}
In this section, we provide  a fault-tolerant realization for quantum circuits with quantum input and output against circuit-level stochastic noise. To do so, we will first provide Lemmas~\ref{lem:Phi-code-imp-stoc}, \ref{lem:interface-1-stc}, and \ref{lem:interface-2-stc}, which are analogous to Lemmas~\ref{lem:Phi-code-imp}, \ref{lem:interface-1}, and \ref{lem:interface-2}, respectively. These Lemmas then can be combined to prove the Theorem~\ref{thm:main-ft2-stoc}, which is analogous to Theorem~\ref{thm:main-ft2}. Here, we will skip the proofs due to the similarity with Lemmas from Section~\ref{sec:q-in-out}. 

%(see~\cite[Section~4.2.1]{christandl2024fault} for more details.)
%
\begin{lemma} \label{lem:Phi-code-imp-stoc}
Consider the quantum circuit $\Phi$ as in Lemma~\ref{lem:Phi-code-imp}. For an integer $k$, let $\Phi_k$ be the realization of $\Phi$ in the code $D_k$ and let $\tilde{\cT}_{\Phi_k}$ be the channel corresponding to the noisy realization of $\Phi_k$, under circuit-level stochastic noise with parameter $\delta$. Then, we have  
\begin{equation} \label{eq:tilde-R-1-stoc}
    \dnorm{  \tilde{\cT}_{\Phi_k} \circ \cJ_k^{\otimes n}  - \cJ_k^{\otimes n'} \circ \cU_k^{\otimes n'} \circ (\cT_{\Phi} \otimes \cF) \circ {\cU_k^\dagger}^{\otimes n}} \leq 2 |\Phi|  (c \delta)^k, 
\end{equation}
where $\cU_k := U_k (\cdot) U_k^\dagger$,  $\cF: \bL(F_k^{\otimes n}) \to \bL(F_k^{\otimes n'})$ is a quantum channel, $\cJ_k: \bL(L_k) \to \bL(N_k)$ is the natural embedding.
\end{lemma}
\medskip
\begin{lemma} \label{lem:interface-1-stc}
Let $\tilde{\cT}_{\Gamma_{1, k}}$ be the realization of $\Gamma_{1, k}$ under circuit-level stochastic noise with parameter $\delta$. Then, there exists a quantum channel $\cS_{\Gamma_{1, k}}: \bL(\mathbb{C}^2) \to \bL(L_k)$ such that 
\begin{equation} \label{eq:S-T-stoc}
    \dnorm{ \cJ_k \circ \cS_{\Gamma_{1, k}} - \tilde{\cT}_{\Gamma_{1, k}}} \leq 2 (c\delta)^k,
\end{equation}
and furthermore,
\begin{equation} \label{eq:s-gamma-1-k}
  \cS_{\Gamma_{1, k}} = (1 - \epsilon)  \: \cS^1_{\Gamma_{1, k}} + \epsilon \: \cS^2_{\Gamma_{1, k}},
\end{equation}
where $\epsilon \leq c\delta$, and $\cS^1_{\Gamma_{1, k}}, \cS^2_{\Gamma_{1, k}} : \bL(\mathbb{C}^2) \to \bL(L_k)$ are quantum channels. The channel $\cS^1_{\Gamma_{1, k}}$ satisfies  (using $\cU_k := U_k (\cdot) U_k^\dagger$)
\begin{equation} \label{eq:s-gamma-1-k-I}
    \tr_{F_k} \circ \: \cU_k^\dagger \circ \cS^1_{\Gamma_{1, k}} = \cI_{\mathbb{C}^2}.
\end{equation}
\end{lemma}
\begin{lemma} \label{lem:interface-2-stc}
Let $\tilde{\cT}_{\Gamma_{k, 1}}$ be the realization of $\Gamma_{k, 1}$ under circuit-level stochastic noise with parameter $\delta$.  Then, we have
\begin{equation} \label{eq:tilde-T-I}
    \tilde{\cT}_{\Gamma_{k, 1}} \circ \cJ_k \circ \cU_k =  (1- \epsilon) \: \cI_{\mathbb{C}^2} \otimes \tr_{F_k} + \epsilon \: \cZ_{\Gamma_{k, 1}},
\end{equation}
for $\epsilon \leq 2c\delta$, and some channel $\cZ_{\Gamma_{k, 1}} : \bL(\mathbb{C}^2 \otimes F_k) \to \bL(\mathbb{C}^2)$.
\end{lemma}
%%%%%%%%%%%%%%%%%%%%%%%%%%%%%%%%%%%%%%%%%%%%%%%
%Proof
%%%%%%%%%%%%%%%%%%%%%%%%%%%%%%%%%%%%%%%%%%%%%%%
% \begin{proof}
% From Part (c) of Def.~\ref{def:ft_scheme_stoc}, for the interface $\Gamma_{k, 1}$, we have 
% \begin{equation} \label{eq:Ttilde-Tbar-1}
%       \tilde{\cT}_{\Gamma_{k,1}} = (1 - \epsilon) \: \overline{\cT}_{\Gamma_{k, 1}}   + \epsilon \: \cZ_{\Gamma_{k,1}}, 
%     \end{equation}
%     where $\epsilon \leq 2c\delta$, $\overline{\cT}_{\Gamma_{k, 1}}$ represents identity channel in $D_k, D_1$, and $\cZ_{\Gamma_{k,1}}$ is some quantum channel. Using the commutation diagram in Fig.~\ref{fig:strong-rep} for $\overline{\cT}_{\Gamma_{k, 1}}$, we  get 
%     \begin{equation}
%         \overline{\cT}_{\Gamma_{k, 1}} \circ \cJ_k = \mu_k^*,
%     \end{equation}
%     which implies 
%     \begin{equation} \label{eq:bar-T-I}
%       \overline{\cT}_{\Gamma_{k, 1}} \circ \cJ_k \circ \cU_k = \cI_{\mathbb{C}^2} \otimes \tr_{F_k}. 
%     \end{equation}
%   Therefore, from Eq.~(\ref{eq:Ttilde-Tbar-1}) and Eq.~(\ref{eq:bar-T-I}), we get Eq.~(\ref{eq:tilde-T-I}).
% \end{proof}
% 
%%%%%%%%%%%%%%%%%%%%%%%%%%%%%%%%%%%%%%%%%%%%%%
%Proof
%%%%%%%%%%%%%%%%%%%%%%%%%%%%%%%%%%%%%%%%%%%%%%%
\begin{theorem} \label{thm:main-ft2-stoc}
Consider the quantum circuit $\Phi$ as in Theorem~\ref{thm:main-ft2}. For any integer $k$, there exists another circuit $\overline{\Phi}$ of size $|\Phi| \cdot \mathrm{poly}(k)$, with the same input-output systems as $\Phi$, such the channel  $\tilde{\cT}_{\overline{\Phi}}$, realized by $\overline{\Phi}$ with circuit-level stochastic channel with parameter $\delta$, satisfies
  \begin{equation} \label{eq:ft-io-stoc}
      \dnorm{\tilde{\cT}_{\overline{\Phi}}  - (\otimes_{i = 1}^{n'} \cW_i) \circ (\cT_{\Phi} \otimes \cF) \circ (\otimes_{i = 1}^n \cN_i) } \leq 4 |\Phi| (c \delta)^k,
  \end{equation}
  where $\cN_i : \bL(\mathbb{C}^2) \to \bL(\mathbb{C}^2 \otimes F_k), i \in [n]$, $\cF: \bL(F_k^{\otimes n}) \to \bL(F_k^{\otimes n'})$, and $\cW_i: \bL(\mathbb{C}^2 \otimes F_k) \to \bL(\mathbb{C}^2), i \in [n']$ are quantum channels, for an auxiliary Hilbert space $F_k$. The channels $\cN_i$ and $\cW_i$ satisfy,
  \begin{align}
    \tr_{F_k} \circ\: \cN_i = (1- \epsilon) \: \cI_{\mathbb{C}^2} + \epsilon \: \cZ_{\Gamma^i_{1, k}}, \label{eq:N_i-stoc}\\
    \cW_i = (1- \epsilon) \: \cI_{\mathbb{C}^2} \otimes \tr_{F_k} + \epsilon \: \cZ_{\Gamma^i_{k, 1}},\label{eq:W_i-stoc}
\end{align}
where $\epsilon \leq 2c\delta$ and  $\cZ_{\Gamma^i_{1, k}}: \bL(\mathbb{C}^2) \to \bL(\mathbb{C}^2)$ and $\cZ_{\Gamma^i_{k, 1}}: \bL(\mathbb{C}^2 \otimes F_k) \to\bL(\mathbb{C}^2)$ are some channels.
\end{theorem}

By taking the quantum circuit $
 \overline{\Phi} := (\otimes_{i=1}^{n'} \Gamma^i_{k, 1}) \circ \Phi_k \circ (\otimes_{i=1}^{n} \Gamma^i_{1, k})$, and using Lemma~\ref{lem:Phi-code-imp-stoc} for $\Phi_k$ , Lemma~\ref{lem:interface-1-stc} for  $\Gamma_{1, k}$, and Lemma~\ref{lem:interface-2-stc} for $\Gamma_{k, 1}$, the proof of Theorem~\ref{thm:main-ft2-stoc} can be done similarly to Theorem~\ref{thm:main-ft2}. We note that using Eq.~(\ref{eq:n-i-def}) and Lemma~\ref{lem:interface-1-stc}, one can get Eq.~(\ref{eq:N_i-stoc}). Furthermore, using Eq.~(\ref{eq:w-i-def}) and Lemma~\ref{lem:interface-2-stc}, one can get Eq.~(\ref{eq:W_i-stoc}). 
\subsubsection{Quantum circuits with only quantum input or quantum output}

Corollary~\ref{cor:enc-stc} below considers  fault-tolerant realization of quantum circuits with classical input and quantum output against circuit-level stochastic noise and is analogous to Corollary~\ref{cor:enc}.

\begin{corollary}\label{cor:enc-stc}
Consider the quantum circuit $\Phi$ from Theorem~\ref{thm:main-ft2-stoc}, with $n' = 0$. Then, there exists a quantum circuit $\overline{\Phi}$, with the same input-output systems as $\Phi$ and of size $|\Phi| \cdot \mathrm{poly}(k)$, such that corresponding channel $\tilde{\cT}_{\overline{\Phi}}$, under circuit-level stochastic channel with parameter $\delta$ satisfies,
\begin{equation} \label{eq:phi-bar-N-stoc}
    \dnorm{\tilde{\cT}_{\overline{\Phi}} -  \cT_{\Phi} \circ \overline{\cN}} \leq 4 |\Phi| (c\delta)^k,
\end{equation}
where $\overline{\cN}: \bL((\mathbb{C}^2)^{\otimes n}) \to \bL((\mathbb{C}^2)^{\otimes n})$ is an independent stochastic channel with parameter $2c\delta$.
\end{corollary} 

\begin{proof}
Similarly to Corollary~\ref{cor:enc}, substituting $n' = 0$ and $\cF = (\tr_{F_k})^{\otimes n}$ in Eq.~(\ref{eq:ft-io-stoc}) , we get that Eq.~(\ref{eq:phi-bar-N-stoc}) is satisfied for $\overline{\cN} = \otimes_{i = 1}^n \overline{\cN}_i$, where $\overline{\cN}_i = \tr_{F_k} \circ \cN_i$. From Eq.~(\ref{eq:N_i-stoc}), we get
\begin{equation}
\overline{\cN}_i =  (1- \epsilon) \: \cI_{\mathbb{C}^2} + \epsilon \: \cZ_{\Gamma^i_{1, k}},
\end{equation}
for $\epsilon \leq 2c\delta$. Therefore,  $\overline{\cN}$ is an independent stochastic channel with parameter $ 2c\delta$.
\end{proof}
Corollary~\ref{cor:dec-stc} below considers  fault-tolerant realization of quantum circuits with quantum input and classical output against circuit-level stochastic noise and is analogous to Corollary~\ref{cor:dec}.

\begin{corollary} \label{cor:dec-stc}
Consider the quantum circuit $\Phi$ from Theorem~\ref{thm:main-ft2}, with $n= 0$. Then, there exists a quantum circuit $\overline{\Phi}$, with the same input-output systems as $\Phi$ and of size $|\Phi| \cdot \mathrm{poly}(k)$, such that the corresponding channel $\tilde{\cT}_{\overline{\Phi}}$, under circuit-level stochastic channel with parameter $\delta$ satisfies,
\begin{equation} \label{eq:CQ-stc}
    \dnorm{\tilde{\cT}_{\overline{\Phi}} -  \overline{\cW} \circ \cT_{\Phi}} \leq 4 |\Phi| (c\delta)^k,
\end{equation}
where $\overline{\cW}: \bL((\mathbb{C}^2)^{\otimes n'}) \to \bL((\mathbb{C}^2)^{\otimes n'})$ is a local stochastic channel with parameter $2c\delta$.
\end{corollary}

\begin{proof}
Similarly to Corollary~\ref{cor:dec}, substituting $n = 0$ and a state preparation map for $\cF$ in Eq.~(\ref{eq:ft-io-stoc}), we get that Eq.~(\ref{eq:CQ-stc}) is satisfied for $\overline{\cW}$ such that 
\begin{equation} \label{eq:w-bar-stc}
    \overline{\cW} = (\otimes_{i = 1}^{n'} \cW_i) \big((\cdot) \otimes \eta \big),
\end{equation}
 where $\eta \in \bL(F_k^{\otimes n'})$ is a quantum state, and $\cW_i, i \in [n']$ is given by, 
 \begin{equation} \label{eq:w-i-stoc}
     \cW_i = (1- \epsilon) \: \cI_{\mathbb{C}^2} \otimes \tr_{F_k} + \epsilon \: \cZ_{\Gamma^i_{k, 1}},
 \end{equation}
for $\epsilon \leq 2c \delta$ and some  channel $\cZ_{\Gamma^i_{k, 1}} : \bL(\mathbb{C}^2 \otimes F_k) \to \bL(\mathbb{C}^2)$. 

\smallskip We now show that the channel $\overline{\cW}$ from Eq.~(\ref{eq:w-bar-stc}) is a local stochastic channel, by verifying that $\overline{\cW}$ satisfies Eqs.~(\ref{eq:local-stc-1}) and~(\ref{eq:local-stc-2}). We will use the following notation: for a subset $A \subseteq [n']$, we define
\begin{equation}
    \cZ_A := (\otimes_{i = 1}^{|A|} \cZ_{\Gamma^i_{k, 1}}) \big(\cdot \otimes \: \eta_A \big),
\end{equation}
where $\eta_A := (\otimes_{i \in [n'] \setminus A } \tr_{F_k^i}) (\eta)$. Note that $\cZ_A : \bL((\mathbb{C}^2)^{\otimes |A|}) \to \bL((\mathbb{C}^2)^{\otimes |A|})$ is a quantum channel. Now using Eq.~(\ref{eq:w-i-stoc}), we can expand $\otimes_{i = 1}^{n'} \cW_i$ as follows,
\begin{equation} \label{eq:W-bar-exp-1-stc}
  \otimes_{i = 1}^{n'} \cW_i = \sum_{A \subseteq [n']} (1-\epsilon)^{n' - |A|} \epsilon^{|A|} (\otimes_{i \in [n'] \setminus A} \: \cI_i \otimes \tr_{F_k} ) \otimes (\otimes_{i \in A} \cZ_{\Gamma^i_{k, 1}}).
\end{equation}
Note that 
\begin{equation} \label{eq:W-bar-exp-2-stc}
    (\otimes_{i \in [n'] \setminus A} \: \cI_i \otimes \tr_{F_k} ) \otimes (\otimes_{i \in A} \cZ_{\Gamma^i_{k, 1}}) (\cdot \otimes \eta) =  \left((\otimes_{i \in [n'] \setminus A} \: \cI_i ) \otimes \cZ_A \right).
\end{equation}
We take $\cV_A := (1-\epsilon)^{n' - |A|} \epsilon^{|A|} \cZ_A$, which is clearly a completely positive map. From Eqs.~(\ref{eq:W-bar-exp-1-stc}) and~(\ref{eq:W-bar-exp-2-stc}), we have
\begin{align}
    \overline{\cW} &= \sum_{A \subseteq [n']} (1-\epsilon)^{n' - |A|} \epsilon^{|A|} \left((\otimes_{i \in [n'] \setminus A} \: \cI_i ) \otimes \cZ_A \right) \nonumber\\
    &= \sum_{A \subseteq [n']}  (\otimes_{i \in [n'] \setminus A} \: \cI_i ) \otimes \cV_A. \label{eq:bar-w-decomp}
\end{align}
Furthermore,
\begin{align}
  \dnorm{\sum_{A : T \subseteq A} \big( \otimes_{i \in [n'] \setminus A} \cI_i \big) \otimes \cV_A } &\leq \sum_{A : T \subseteq A} (1-\epsilon)^{n' - |A|} \epsilon^{|A|} \nonumber\\
  &=  \sum_{j = 0}^{n' - |T|} \binom{n'- |T|}{j} (1-\epsilon)^{n' - |T| - j} \epsilon^{|T| + j} \nonumber\\
  & = \epsilon^{|T|}. \label{eq:bar-w-loc-stoc}
\end{align} 
Therefore, from Eq.~(\ref{eq:bar-w-decomp}) and Eq.~(\ref{eq:bar-w-loc-stoc}), $\overline{\cW}$ is a local stochastic channel.
\end{proof}

\section{Fault-tolerant communication} \label{sec:ft-comm}
One of the goals of quantum Shannon theory is to design codes for reliable quantum communication between the sender and the receiver connected by a noisy channel~\cite{wilde2013quantum, holevo1998capacity, devetak2005private}. In~\cite{christandl2022fault,belzig2023fault}, it is noted that the standard quantum Shannon theory is not fit for practical applications as it only considers the noise in the communication channel and disregards the noise in encoding and decoding circuits of the communication code. This necessitates the construction of fault-tolerant communication codes, consisting of fault-tolerant encoders and decoders. In~\cite{christandl2022fault,belzig2023fault}, assuming a circuit-level Pauli noise, it has been shown that fault-tolerant encoders and decoders can be constructed from standard non fault-tolerant encoders and decoders. Albeit the main focus  of~\cite{christandl2022fault} is proving achievability of channel capacity in the fault-tolerant setting, a specific fault-tolerant communication code has also been provided using communication codes that have linear minimum distance~\cite[Section~V.C]{christandl2022fault}. In this section, we shall focus on the latter part and
%that is, the use of codes for fault-tolerant quantum communication going beyond linear mi. 
go beyond the results in~\cite[Section~V.C]{christandl2022fault} in the following ways:
\begin{enumerate}
    \item We consider more realistic circuit-level noise such as general circuit noise and circuit-level stochastic noise from Sec.~\ref{sec:noise-qcirc}.

    \item Using the fault-tolerant realizations given in Section~\ref{sec:ft-real}, we provide tools to construct fault-tolerant communication codes using a notion of \emph{robust} communication codes, where robustness is defined in terms of adversarial channel as in Def.~\ref{def:adver-can} and local stochastic channel as in Def.~\ref{def:local-stoc}. Therefore, a general method is provided for importing existing communication codes to the fault-tolerant scenario. 

    \item For general circuit noise, we show that fault-tolerant communication codes can be obtained from communication codes with a linear minimum distance; therefore, extending the results from~\cite[Section~V.C]{christandl2022fault} to general circuit noise. For circuit-level stochastic noise, we show that that fault-tolerant communication codes can be obtained from communication codes with a sub-linear minimum distance; therefore, improving significantly the specific coding construction from~\cite[Section~V.C]{christandl2022fault}.  
\end{enumerate}

\medskip To avoid confusion with the communication codes, we shall refer to the codes $D_k, k= 1, 2, \dots$ used in the fault-tolerant scheme from Def.~\ref{def:ft_scheme} and~\ref{def:ft_scheme_stoc}, as the computational codes.

\subsection{Fault-tolerant codes for quantum communication}
We shall work with the same definition of fault-tolerant communication codes, as given in~\cite[Def.~V.6]{christandl2022fault}. Before giving the definition of  fault-tolerant communication codes, we firstly recall standard one shot communication codes in the following definition.

\begin{definition}[Communication codes] \label{def:comm-code}
 Let $\Lambda: \bL(H_A) \to \bL(H_B)$ be a communication channel. For an integer $m$ and $\zeta \in [0, 1]$, an $(m, \zeta)$ communication code for $\Lambda$ is defined by encoding and decoding channels $\cE: \bL((\mathbb{C}^2)^{\otimes m}) \to \bL(H_A)$ and $\cD: \bL(H_B)  \to \bL((\mathbb{C}^2)^{\otimes m})$, respectively, such that 
 \begin{equation}
    \dnorm{ \mathcal{D} \circ \Lambda \circ \mathcal{E} - \cI_{(\mathbb{C}^2)^{\otimes m}}} \leq \zeta.
\end{equation}
\end{definition}

\medskip In the following, we shall consider spaces $H_{A}$, and $H_B$ in $\Lambda: \bL(H_A) \to \bL(H_B)$ to be multi-qubit systems, that is, $H_{A} = (\mathbb{C}^2)^{\otimes n_A}$, and $H_{B} = (\mathbb{C}^2)^{\otimes n_B}$.

\medskip For fault-tolerant communication, the goal is to transfer logical information encoded in a computational code through the communication channel. Naturally, to define fault-tolerant communication codes, one needs to consider encoding and decoding circuits in contrast to encoding and decoding channels considered in Def.~\ref{def:comm-code}.

\begin{definition}[Fault-tolerant communication code] \label{def:fault-tol-comm}
A fault-tolerant code on a communication channel $\Lambda: \bL(H_A) \to \bL(H_B)$, where  $H_{A} = (\mathbb{C}^2)^{\otimes n_A}, \: H_{B} = (\mathbb{C}^2)^{\otimes n_B}$, using the computational code $D_k$ consists of the following steps:
\begin{enumerate}
    \item The sender holds an $m$ qubit logical state encoded in the computational code $D_k$. 

    \item The sender applies an encoding circuit $\Phi_A$ on the logical qubits. $\Phi_A$ takes $m$ logical qubits as input and outputs a quantum system $H_A$. Note that $\Phi_A$ realizes a quantum channel $\cT_{\Phi_A}: \bL(N_k^{\otimes m}) \to \bL(H_A)$. 
    
    \item  The channel $\Lambda: \bL(H_A) \to \bL(H_B)$ is applied on the output of $\Phi_A$, and the receiver gets the channel output.

    \item  Finally, the receiver decodes the channel output using a decoding circuit $\Phi_B$, which takes a quantum system $H_B$ as input and outputs $m$ logical qubits encoded in the computational code $D_k$. Note that $\Phi_B$ realizes a quantum channel $\cT_{\Phi_B}: \bL(H_B) \to \bL(N_k^{\otimes m})$.  
\end{enumerate} 
\end{definition}
To quantify the communication error of a fault-tolerant communication code, we will use the same measure as in~\cite[Def.~V.6]{christandl2022fault}, which is defined as follows.

\medskip Let $\Phi^1$ be a state preparation circuit (classical input and quantum output), and $\Phi^2$ be a measurement circuit (quantum input and classical output), such that their sizes are upper bounded by $l > 0$, that is, $|\Phi^1|, |\Phi^2| \leq l$.  Let $\Phi^1_k$ and $\Phi^2_k$ be the  realizations of $\Phi^1$ and $\Phi^2$, respectively, in the code $D_k$. For a quantum circuit $\Phi$, consider the corresponding quantum channel $\cT_\Phi$, and also $\tilde{\cT}_\Phi$ corresponding to the noisy realization of $\Phi$ under a circuit noise with error rate $\delta$ (e.g. general circuit noise or circuit-level stochastic noise from Section~\ref{sec:noise-qcirc}). Then the communication error of a fault-tolerant communication code according to Def.~\ref{def:fault-tol-comm} is given by, 
\begin{equation} \label{eq:comm-err}
  \zeta^{l, \delta} := \max_{\substack{\Phi^1, \Phi^2 \\ |\Phi^1|, |\Phi^2|  \leq l}} \: \max_{\tilde{\cT}_{\Phi^1_k}, \tilde{\cT}_{\Phi_A}, \tilde{\cT}_{\Phi_B}, \tilde{\cT}_{\Phi^2_k}} \dnorm{\cT_{\Phi^2} \circ \cT_{\Phi^1} -  \tilde{\cT}_{\Phi^2_k}  \circ \tilde{\cT}_{\Phi_B} \circ \Lambda \circ \tilde{\cT}_{\Phi_A} \circ \tilde{\cT}_{\Phi^1_k}},
\end{equation}
where the first maximization is over all noisy realizations of $\tilde{\cT}_{\Phi^1_k}, \tilde{\cT}_{\Phi_A}, \tilde{\cT}_{\Phi_B}, \tilde{\cT}_{\Phi^2_k}$, with error rate $\delta$, and the second maximization is over all the preparation and measurement circuits of size $l$. We know from Solovay-Kitaev theorem~\cite{kitaev, nielsen2001quantum, dawson2005solovay} that for any state preparation channel $\cP: \mathbb{B}^p \to \bL((\mathbb{C}^2)^{\otimes m})$ and any measurement channel $\cM: \bL((\mathbb{C}^2)^{\otimes m}) \to \mathbb{B}^{p'}$, there exist circuits of size no more than an $l(m, \zeta) > 0$, realizing $\cP$ and $\cM$, up to an accuracy $\zeta > 0$. Therefore, it suffices to consider $l := l(m, \zeta)$ in Eq.~(\ref{eq:comm-err}), for sufficiently small $\zeta$.

\medskip We now give a definition of $(m, \zeta, \delta)$ fault-tolerant communication codes for $\Lambda: \bL(H_A) \to \bL(H_B)$.

\begin{definition}[$(m, \zeta, \delta)$ fault-tolerant communication codes] \label{def:ft-comm-code}
    Consider a fault-tolerant communication code according to Def.~\ref{def:fault-tol-comm} on $\Lambda: \bL(H_A) \to \bL(H_B)$. It is said to be an $(m, \zeta, \delta)$ code if the communication error according to Eq.~(\ref{eq:comm-err}) under a circuit noise of error rate $\delta$ and $l = l(m, \zeta)$ is upper bounded by $\zeta$, that is, $\zeta^{l, \delta} \leq \zeta$.
\end{definition}

\subsection{Construction of fault-tolerant communication codes for general circuit noise}

In this section, we consider general circuit noise from Section~\ref{sec:noise-qcirc}, and show that fault-tolerant communication codes can be constructed using a notion of adversarial robust communication codes. Furthermore, we show that adversarial robust communication codes can be obtained from communication codes correcting  arbitrary errors of a linear weight; therefore, providing a construction of fault-tolerant codes. We note that there are many classes of known codes that correct arbitrary errors of a linear weight, notably, \emph{asymptotically good} quantum codes, which have linear minimum distance and communication rate~\cite{ashikhmin2001asymptotically, chen2001asymptotically, li2009family, matsumoto2002improvement, brown2013short, panteleev2022asymptotically, leverrier2022quantum}.

\smallskip Below, we define adversarial robust communication codes, using the adversarial channel from Def.~\ref{def:adver-can}.

\begin{definition}[$(m, \zeta, \delta)$ adversarial robust communication codes] \label{def:robust-can-code}
Consider a communication code for $\Lambda: \bL(H_A) \to \bL(H_B), \: H_{A} = (\mathbb{C}^2)^{\otimes n_A}, \: H_{B} = (\mathbb{C}^2)^{\otimes n_B}$ according to Def.~\ref{def:comm-code}, defined by $\cE: \bL((\mathbb{C}^2)^{\otimes m}) \to \bL(H_A)$ and $\cD: \bL(H_B)  \to \bL((\mathbb{C}^2)^{\otimes m})$. We say it is an $(m, \zeta, \delta)$ adversarial robust code if for any adversarial channels $\cV_A:\bL(H_A) \to \bL(H_A)$ and $\cV_B:\bL(H_B) \to \bL(H_B)$, with parameter $\delta$ (see Def.~\ref{def:adver-can}), the channels $\cE$ and $\cD$ give an $(m, \zeta)$ communication code for the combined channel $\cV_B \circ \Lambda \circ \cV_A$.
\end{definition}

\smallskip  
%We need communication codes correcting arbitrary errors of weight at most $t$. 
We say that $\cE:\bL((\mathbb{C}^2)^{\otimes m}) \to \bL(H)$ and $\cD: \bL(H) \to \bL((\mathbb{C}^2)^{\otimes m})$ corrects arbitrary errors of weight $t$ if 
for any superoperator $\cV_{t}: \bL(H) \to \bL(H)$ of weight $t$, we have
 \begin{equation} \label{eq:corr-cond}
     \cD \circ \cV_{t} \circ \cE = K(\cV_{t}) \: \cI_{(\mathbb{C}^2)^{\otimes m}}, 
 \end{equation}
where $K(\cV_{t})$ is a constant depending only on $\cV_t$. We note that $K(\cV_t) = 1$, when $\cV_t$ is quantum channel. In general, we have   $K(\cV_t) \: \ident_{(\mathbb{C}^2)^{\otimes m}} = (\cE^* \circ \cV_t^*)(\ident_{(\mathbb{C}^2)^{\otimes n}})$.

\medskip In the following theorem, we show that fault-tolerant communication codes can be obtained from adversarial robust communication codes.

\begin{theorem} \label{thm:robust-ft-code}
Consider $c, \delta \geq 0$ from Def.~\ref{def:ft_scheme} and let $\delta' := 3c\delta$. Let $(m, \zeta, \delta')$ be an adversarial robust communication code for $\Lambda: \bL(H_A) \to \bL(H_B), \: H_{A} = (\mathbb{C}^2)^{\otimes n_A}, \: H_{B} = (\mathbb{C}^2)^{\otimes n_B}$, according to Def.~\ref{def:robust-can-code}. Then there exists $k_0 > 0$, such that for any computational code $D_k, k \geq k_0$, one can obtain an $(m, 2\zeta, \delta)$ fault-tolerant communication code for $\Lambda$.
\end{theorem}
\begin{proof}
Consider an $(m, \zeta, \delta')$ adversarial robust code on $\Lambda$ given by the channels $\cE: \bL((\mathbb{C}^2)^{\otimes m})  \to \bL(H_A)$ and $\cD: \bL(H_B)  \to \bL((\mathbb{C}^2)^{\otimes m})$. Therefore, we have
\begin{equation} \label{eq:can-cod-stan}
    \dnorm{ \mathcal{D} \circ \cV_B \circ \Lambda \circ \cV_A \circ \mathcal{E} - \cI_{(\mathbb{C}^2)^{\otimes m}}} \leq \zeta,
\end{equation}
for any adversarial channels $\cV_A:\bL(H_A) \to \bL(H_A)$, and $\cV_B:\bL(H_B) \to \bL(H_B)$, with parameter $\delta'$.

\medskip\noindent Consider quantum circuits $\Phi_E$ and $\Phi_D$, such that 
  \begin{align}
     \dnorm{ \cT_{\Phi_E} - \cE} &\leq \frac{\zeta}{4}, \label{eq:phi-E} \\
     \dnorm{ \cT_{\Phi_D} - \cD} &\leq \frac{\zeta}{4}, \label{eq:phi-D}
  \end{align} 
 and let $\Phi_{E, k}$ and $\Phi_{D, k}$ be their realizations in the code $D_k$. We define quantum circuits $\Phi_A$ and $\Phi_B$ in Def.~\ref{def:fault-tol-comm} as follows, 
 \begin{align}
     \Phi_A := (\otimes_{i=1}^{n_A} \Gamma^i_{k, 1}) \circ \Phi_{E,k}. \\
     \Phi_B :=  \Phi_{D,k} \circ (\otimes_{i=1}^{n_B} \Gamma^i_{1, k}).
 \end{align}

Consider now a state preparation circuit $\Phi^1$, with $m$ qubit output and a measurement circuit $\Phi^2$, with $m$ qubit input, such that $\max(|\Phi^1|, |\Phi^2|) \leq l(m, \zeta)$. Let $\Phi^1_k$ and $\Phi^2_k$ be their realizations in the code $D_k$.

\medskip Applying Corollary~\ref{cor:dec} on the state preparation circuit $\Phi_A \circ \Phi^1_k$, we get for any $\tilde{\cT}_{\Phi_A} \in \trans(\Phi_A, \delta)$ and $\tilde{\cT}_{\Phi_k^1} \in \trans(\Phi_k^1, \delta)$
\begin{equation} \label{eq:phi^1}
    \dnorm{\tilde{\cT}_{\Phi_A} \circ \tilde{\cT}_{\Phi^1_k} -  \overline{\cW} \circ  \cT_{\Phi_E} \circ \cT_{\Phi^1}} \leq 4 ( |\Phi_E| + l(m, \zeta)) \sqrt{(c\delta)^k},
\end{equation}
where $\overline{\cW}$ is an adversarial channel with parameter $\delta'$.

\medskip Applying Corollary~\ref{cor:enc} on the measurement circuit $\Phi^2_k \circ \Phi_B $, we get for any $\tilde{\cT}_{\Phi_B} \in \trans(\Phi_B, \delta)$ and $\tilde{\cT}_{\Phi_k^2} \in \trans(\Phi_k^2, \delta)$, 
\begin{equation} \label{eq:phi^2}
    \dnorm{ \tilde{\cT}_{\Phi^2_k} \circ \tilde{\cT}_{\Phi_B}  - \cT_{\Phi^2} \circ \cT_{\Phi_D} \circ \overline{\cN}} \leq 4 ( |\Phi_D| + l(m, \zeta)) \sqrt{(c\delta)^k},
\end{equation}
where $\overline{\cN}$ is an adversarial channel with parameter $\delta'$.

\medskip \noindent We take $k_0$ sufficiently big such that, 
\begin{equation} \label{eq:k-o}
    4( |\Phi_E| + |\Phi_D| +2l(m, \zeta))  \sqrt{(c\delta)^{k_0}}\leq \frac{\zeta}{2}
\end{equation}
For all $k \geq k_0$, the communication error (see Eq.~(\ref{eq:comm-err})) is given by,
\begin{align}
\zeta^{l, \delta} &=   \dnorm{\cT_{\Phi^2} \circ \cT_{\Phi^1} - \tilde{\cT}_{\Phi^2_k} \circ \tilde{\cT}_{\Phi_B} \circ \Lambda \circ \tilde{\cT}_{\Phi_A} \circ \tilde{\cT}_{\Phi^1_k} } \nonumber \\
 &  \leq  \frac{\zeta}{2} +  \dnorm{ \cT_{\Phi^2} \circ \cT_{\Phi^1} - \cT_{\Phi^2} \circ \cT_{\Phi_D} \circ \overline{\cN} \circ \Lambda \circ  \overline{\cW} \circ  \cT_{\Phi_E} \circ \cT_{\Phi^1}}  \nonumber \\
& \leq  \zeta +  \dnorm{ \cT_{\Phi^2} \circ \cT_{\Phi^1} 
- \cT_{\Phi^2} \circ \cD \circ \overline{\cN} \circ \Lambda \circ  \overline{\cW}  \circ  \cE \circ \cT_{\Phi^1}} \nonumber \\
& \leq 2 \zeta,
\end{align}
where the first inequality follows from Eq.~(\ref{eq:phi^1}),~(\ref{eq:phi^2}), and~(\ref{eq:k-o}), the second inequality follows from Eq.~(\ref{eq:phi-E}) and~(\ref{eq:phi-D}), and the last inequality follows from Eq.~(\ref{eq:can-cod-stan}).  Therefore, from Def.~\ref{def:ft-comm-code}, $\Phi_A$ and $\Phi_B$ give an $(m, 2\zeta, \delta)$ fault-tolerant communication code. 
\end{proof}

\medskip In the following theorem, we consider $\cE$ and $\cD$ correcting arbitrary errors of weight $t = \alpha \: n$ according to  Eq.~(\ref{eq:corr-cond}), and the communication channel $\Lambda_\nu: \bL(H) \to \bL(H), \: H = (\mathbb{C}^2)^{\otimes n}$ to be an adversarial channel with parameter $\nu \in [0, 1]$. Using $\cE$ and $\cD$, we provide reliable fault-tolerant communication code on $\Lambda_\nu$, with a communication error decreasing exponentially with the code length.
\begin{theorem} \label{thm:ft-comm-code-general}
   Consider $c, \delta > 0$ from Def.~\ref{def:ft_scheme} and let $\delta' = 3c\delta$. Let $\Lambda_\nu: \bL(H) \to \bL(H), \: H = (\mathbb{C}^2)^{\otimes n}$ be an adversarial channel with parameter $\nu \in [0, 1]$. Suppose channels $\cE:\bL((\mathbb{C}^2)^{\otimes m}) \to \bL(H)$ and $\cD: \bL(H) \to \bL((\mathbb{C}^2)^{\otimes m})$ can correct arbitrary errors of weight  $t = \alpha n$ for an $\alpha \in [0, 1]$, such that $\alpha \geq 5(2\delta' + \nu)$. Then there exists $k_0 > 0$, such that for any computational code $D_k, k \geq k_0$, one can obtain an $(m, 2\zeta, \delta)$ fault-tolerant communication code for $\Lambda_\nu$, where $\zeta :=  14 \exp(-\frac{n \min\{\delta', \nu\}}{3})$.
\end{theorem}

\begin{proof}
We will show that $\cE$ and $\cD$ define an $(m, \zeta)$ communication code for $ \overline{\Lambda}_\nu  := \cV_B \circ \Lambda_\nu \circ \cV_A$ for any adversarial channels $\cV_A: \bL(H) \to \bL(H)$ and $\cV_B: \bL(H) \to \bL(H)$ with parameter $\delta'$, i.e., they define an $(m, \zeta, \delta')$ adversarial robust communication code on $\Lambda_\nu$. Then, from Theorem~\ref{thm:robust-ft-code}, it follows that $\cE$ and $\cD$ define an $(m, 2\zeta, \delta)$ fault-tolerant code. 

\medskip Firstly, we show that any adversarial channel $\cV: \bL(H) \to \bL(H)$ with parameter $\delta$ can be approximated by a superoperator of weight $5 n \delta$.     From Def.~\ref{def:adver-can}, we know that there exists a superoperator $\cV_t$ of weight $t$ such that 
\begin{align}
    \dnorm{\cV - \cV_t} &\leq  \sum_{j > t}^n \binom{n}{j} \delta^j \nonumber \\
    & = (1 + \delta)^n \sum_{j > t}^n \binom{n}{j} \left(\frac{\delta}{1 + \delta}\right)^j \left(\frac{1}{1 + \delta}\right)^{n-j}. \label{eq:exp-v-vt}
\end{align}
Consider $n$ $i.i.d.$ random variables $X_1, \dots, X_n$ such that $X_i \in \{0, 1\}$ and $\mathrm{Pr}(X_i = 1) = \frac{\delta}{1 + \delta}, \forall i \in [n]$, and let $X = \sum_{i \in [n]} X_i$. Then, we have that 
\begin{equation} \label{eq:p-xt}
 \mathrm{Pr}(X > t) =  \sum_{j > t}^n \binom{n}{j} \left(\frac{\delta}{1 + \delta}\right)^j \left(\frac{1}{1 + \delta}\right)^{n-j}.  
\end{equation}
From Eqs.~(\ref{eq:exp-v-vt}) and (\ref{eq:p-xt}) and using $(1 + \delta)^n \leq e^{n\delta}$, we get 
\begin{equation} \label{eq:sim-v-vt}
 \dnorm{\cV - \cV_t} \leq e^{n\delta}  \mathrm{Pr}(X > t).
\end{equation}
For a $\beta > 0$, we take $t = (1 + \beta) \frac{n \delta}{1 + \delta}$. Then, using the multiplicative Chernoff bound, we get
\begin{equation} \label{eq:p-xt-1}
   \mathrm{Pr}(X \geq (1 + \beta) \frac{n \delta}{1 + \delta}) \leq \exp(-\frac{\beta^2 n \delta }{(1 + \delta) (2 + \beta) })  \leq \exp(-\frac{\beta^2 n \delta }{2(2 + \beta) }).
\end{equation}
We choose $\beta = 4$. Then from Eqs.~(\ref{eq:sim-v-vt}) and (\ref{eq:p-xt-1}), we get that 
\begin{equation} \label{eq:sim-v-vt-1}
 \dnorm{\cV - \cV_t} \leq \exp(-\frac{n \delta}{3}).
\end{equation}
Therefore, there exists a superoperator $\cV' = \cV_t$ of weight $t = \frac{5n\delta}{1 + \delta} \leq 5n\delta$ satisfying Eq.~(\ref{eq:sim-v-vt-1}).  We now consider superoperators $\cV'_A, \Lambda'_\nu, \cV'_B$ corresponding to the adversarial channels $\cV_A, \Lambda_\nu, \break \cV_B$, respectively, where $\cV'_A$, and  $\cV'_B$ are of weight $5n\delta'$ and $\Lambda'_\nu$ is of weight $5n\nu$. Therefore, the combined channel $\widetilde{\Lambda}_\nu := \cV'_B \circ \Lambda'_\nu \circ \cV'_A$ is of weight $5 n(2\delta' + \nu)$. We have that
\begin{align}
    &\dnorm{\overline{\Lambda}_\nu - \widetilde{\Lambda}_\nu} = \dnorm{\cV_B \circ \Lambda_\nu \circ \cV_A - \cV'_B \circ \Lambda'_\nu \circ \cV'_A} \nonumber\\
    & \leq \dnorm{\cV_B - \cV'_B} + \dnorm{\cV'_B} \: \dnorm{\Lambda_\nu \circ \cV_A - \Lambda'_\nu \circ \cV'_A} \nonumber\\
    & \leq \exp(-\frac{n \delta'}{3}) + (1 + \exp(-\frac{n \delta'}{3}))  \: \left( \dnorm{\cV_A - \cV'_A} +  \dnorm{\cV_A} \dnorm{\Lambda_\nu - \Lambda'_\nu} \right) \nonumber \\
    & \leq \exp(-\frac{n \delta'}{3}) + (1 + \exp(-\frac{n \delta'}{3})) \exp(-\frac{n \delta'}{3}) + (1 + \exp(-\frac{n \delta'}{3}))^2 \exp(-\frac{n \nu}{3}) \nonumber \\
    & \leq 7 \exp(-\frac{n \min\{\delta', \nu\}}{3}), \label{eq:Lambda-b-t}
\end{align}
%\left( \dnorm{\Lambda_\nu - \Lambda'_\nu} +  \dnorm{\Lambda_\nu} \dnorm{\cV_A - \Lambda'_\nu} \right) 
where the first inequality follows by inserting $-\cV'_B \circ \Lambda_\nu \circ \cV_A + \cV'_B \circ \Lambda_\nu \circ \cV_A = 0$, using triangle inequality and $\dnorm{\cW \circ \cW'} \leq \dnorm{\cW} \dnorm{\cW'}$ for any superoperators $\cW, \cW'$~\cite{kitaev}. The second inequality uses Eq.~(\ref{eq:sim-v-vt-1}), and $\dnorm{\cV'_B} \leq \dnorm{\cV_B} + \dnorm{\cV_B - \cV'_B} \leq 1 + \exp(-\frac{n \delta'}{3})$.

\smallskip As $\alpha \geq 5(2\delta' + \nu)$, we have that $\cD \circ \widetilde{\Lambda}_\nu  \circ \cE = K(\tilde{\Lambda}_\nu) \: \cI$, for a constant $K(\tilde{\Lambda}_\nu)$. We note that 
\begin{align}
    |1 - K(\tilde{\Lambda}_\nu)| &= |\tr\left( \cD \circ \overline{\Lambda}_\nu  \circ \cE (\rho) -  (\cD \circ \widetilde{\Lambda}_\nu  \circ \cE)(\rho)\right) |  \nonumber \\
    & \leq 7 \exp(-\frac{n \min\{\delta', \nu\}}{3}), \label{eq:K-V}
\end{align}
where $\rho$ is any quantum state, and the second line uses $  \dnorm{ \cD \circ \overline{\Lambda}_\nu \circ \cE - \cD \circ \widetilde{\Lambda}_\nu \circ \cE}  \break \leq 7 \exp(-\frac{n \min\{\delta', \nu\}}{3})$, which follows from Eq.~(\ref{eq:Lambda-b-t}).

\smallskip\noindent We now have 
\begin{align}
    \dnorm{\cD \circ \overline{\Lambda}_\nu  \circ \cE - \cI} &\leq \dnorm{\cD \circ \widetilde{\Lambda}_\nu  \circ \cE - \cI} + 7 \exp(-\frac{n \min\{\delta', \nu\}}{3}) \nonumber \\
    &\leq |1 - K(\tilde{\Lambda}_\nu)| + 7 \exp(-\frac{n \min\{\delta', \nu\}}{3}) \nonumber \\
    &\leq 14 \exp(-\frac{n \min\{\delta', \nu\}}{3}), \label{eq:robust-code-cons}
\end{align}
where the first inequality uses Eq.~(\ref{eq:Lambda-b-t}), the second inequality uses $\cD \circ \widetilde{\Lambda}_\nu  \circ \cE = K(\tilde{\Lambda}_\nu) \: \cI$, and the last inequality uses Eq.~(\ref{eq:K-V}). From Eq.~(\ref{eq:robust-code-cons}), $\cE$ and $\cD$ define an $(m, \zeta)$ code on $\overline{\Lambda}_\nu$; therefore they define an $(m, \zeta, \delta')$ adversarial robust code on $\Lambda_\nu$.
\end{proof}

\subsection{Construction of fault-tolerant communication codes for circuit-level stochastic noise}

In this section, we consider circuit-level stochastic noise from Section~\ref{sec:noise-qcirc}, and show that fault-tolerant communication codes can be constructed using a notion of local stochastic robust communication codes. Furthermore, we show that fault-tolerant communication codes can be obtained from communication codes on local stochastic channels. This permits the use of sub-linear minimum distance codes such as toric and expander codes~\cite{kitaev, kitaev2003fault,  bravyi1998quantum, dennis2002topological, tillich2013quantum, fawzi2018efficient} for fault-tolerant communication. We illustrate this in detail in Section~\ref{sec:toric-code} for the toric code by showing that it is a  reliable communication code on local stochastic channels with sufficiently small parameter.

\smallskip Below, we define local stochastic robust communication codes (in an analogous way to Def.~\ref{def:robust-can-code}), using local stochastic channels from Def.~\ref{def:local-stoc}.

\begin{definition}[$(m, \zeta, \delta)$ local stochastic robust communication codes]  \label{def:robust-can-code-stc}
Consider a communication code for $\Lambda: \bL(H_A) \to \bL(H_B), \: H_{A} = (\mathbb{C}^2)^{\otimes n_A}, \: H_{B} = (\mathbb{C}^2)^{\otimes n_B}$ according to Def.~\ref{def:comm-code}, defined by $\cE: \bL((\mathbb{C}^2)^{\otimes m}) \to \bL(H_A)$ and $\cD: \bL(H_B)  \to \bL((\mathbb{C}^2)^{\otimes m})$. We say it is an $(m, \zeta, \delta)$ local stochastic robust code if for any local stochastic channels $\cV_A:\bL(H_A) \to \bL(H_A)$ and $\cV_B:\bL(H_B) \to \bL(H_B)$, with parameter $\delta$ (see Def.~\ref{def:local-stoc}), the channels $\cE$ and $\cD$ give an $(m, \zeta)$ communication code for the combined channel $\cV_B \circ \Lambda \circ \cV_A$.  
\end{definition}

\medskip In the following theorem, we show that fault-tolerant communication codes can be obtained from local stochastic robust communication codes.

\begin{theorem} \label{thm:robust-ft-code-stc}
Consider $c, \delta \geq 0$ from Def.~\ref{def:ft_scheme_stoc} and let $\delta' := 3c\delta$. Let $(m, \zeta, \delta')$ be a local stochastic robust communication code for $\Lambda: \bL(H_A) \to \bL(H_B), \: H_{A} = (\mathbb{C}^2)^{\otimes n_A}, \: H_{B} = (\mathbb{C}^2)^{\otimes n_B}$, according to Def.~\ref{def:robust-can-code-stc}. Then there exists $k_0 > 0$, such that for any computational code $D_k, k \geq k_0$, one can obtain an $(m, 2\zeta, \delta)$ fault-tolerant communication code for $\Lambda$.
\end{theorem}
We skip the proof of Theorem~\ref{thm:robust-ft-code-stc} as it is essentially the same as Theorem~\ref{thm:robust-ft-code}. It can be proved similarly by replacing adversarial channels with local stochastic channels in Theorem~\ref{thm:robust-ft-code} and using Corollary~\ref{cor:enc-stc} instead of~\ref{cor:enc}, and Corollary~\ref{cor:dec-stc} instead of~\ref{cor:dec}. We note that the value of $k_0$ is such that 
\begin{equation} \label{eq:k-0-stc}
    4 ( |\Phi_E| + |\Phi_D| +2l(m, \zeta))  (c\delta)^{k_0}\leq \frac{\zeta}{2},
\end{equation}
where $|\Phi_E|$, $|\Phi_D|$ and $2l(m, \zeta)$ has the same meaning as in Eq.~(\ref{eq:k-o}).

\medskip In the following theorem, we consider the communication channel to be a local stochastic channel $\Lambda_\nu$, with parameter $\nu \in [0, 1]$, and show that fault-tolerant communication code on $\Lambda_\nu$ can be obtained using communication codes for local stochastic channels. Moreover, in  Appendix~\ref{sec:toric-code}, we show that the toric codes, whose minimum distance scales as the square root of the code length (contrary to the linear scaling required in Theorem~\ref{thm:ft-comm-code-general}) are reliable communication codes for the local stochastic noise; therefore, can be used as a fault-tolerant communication code. 

%A similar argument can be applied for codes with a linear rate of communication, for example, expander codes~\cite{tillich2013quantum, fawzi2018efficient}. 
%
\begin{theorem} \label{thm:ft-code-stc}
   Consider $c, \delta > 0$ from Def.~\ref{def:ft_scheme} and let $\delta' = 3c\delta$. Let $\Lambda_\nu: \bL(H) \to \bL(H), \: H = (\mathbb{C}^2)^{\otimes n}$ be a local stochastic channel with parameter $\nu \in [0, 1]$. Suppose  $\cE:\bL((\mathbb{C}^2)^{\otimes m}) \to \bL(H)$ and $\cD: \bL(H) \to \bL((\mathbb{C}^2)^{\otimes m})$ define an $(m, \zeta)$ communication code for all local stochastic channels with parameter $\alpha$. Then there exists $k_0 > 0$, such that for any computational code $D_k, k \geq k_0$, one can obtain an $(m, 2\zeta, \delta)$ fault-tolerant communication code for $\Lambda_\nu$ if $\alpha \geq  2\delta' + \nu$.
\end{theorem}

\begin{proof}
Lemma~\ref{lem:local-stochastic-composition} below shows that  $\overline{\Lambda}_\nu := \cV' \circ \Lambda_\nu \circ \cV$, for any local stochastic channels $\cV: \bL(H) \to \bL(H)$ and $\cV': \bL(H) \to \bL(H)$ with parameter $\delta'$, is a local stochastic channel with parameter $2\delta' + \nu$. Therefore, for $\alpha \geq 2\delta' + \nu$, $\cE$ and $\cD$ define an $(m, \zeta)$ communication code for $\overline{\Lambda}_\nu$. This means that  $\cE$ and $\cD$ define an $(m, \zeta, \delta')$ robust quantum communication code for $\Lambda_\nu$; therefore, an $(m, 2\zeta, \delta)$ fault-tolerant code can be obtained using Theorem~\ref{thm:robust-ft-code-stc}.
\end{proof}

\begin{lemma}
\label{lem:local-stochastic-composition}
Let $\cA$ and $\cB$ be local stochastic channels with parameters $\nu_{\cA}$ and $\nu_{\cB}$, respectively. Then $\cA \circ \cB$ is local stochastic with parameter $\nu_{\cA} + \nu_{\cB}$.
\end{lemma}
\begin{proof}
Consider the decomposition of $\cA$ and $\cB$ according to Eq.~(\ref{eq:local-stc-1}),
    \begin{align}
     \cA  = \sum_{A \subseteq [n]} \big( \otimes_{i \in [n] \setminus A} \cI_i \big) \otimes {\cA}_A,  \\
     \cB = \sum_{A \subseteq [n]} \big( \otimes_{i \in [n] \setminus A} \cI_i \big) \otimes {\cB}_A, 
    \end{align}
   where ${\cA}_A$ and  ${\cB}_A$ are completely positive maps.
 Note that, for any $A_1, A_2 \subseteq [n]$, we have 
 \begin{equation} \label{eq:W-AUA'}
    \left( ( \otimes_{i \in [n] \setminus A_1} \cI_i ) \otimes {\cA}_{A_1} \right) \circ \left(( \otimes_{i \in [n] \setminus A_2} \cI_i ) \otimes \cB_{A_2}  \right)  =   \big( \otimes_{i \in [n] \setminus (A_1  \cup A_2)} \cI_i \big) \otimes \cW_{A_1 \cup A_2}, 
 \end{equation}
for a completely positive map $\cW_{A_1 \cup A_2}$ acting on qubits in $A_1 \cup A_2$. This implies that 
\begin{equation} \label{eq:W-A-set}
  \cA \circ \cB = \sum_{A \subseteq [n]} \big( \otimes_{i \in [n] \setminus A} \cI_i \big) \otimes \cW_A,  
\end{equation}
where  $\cW_A = \sum_{\substack{A_1, A_2 \subseteq [n] \\ A_1  \cup A_2 = A}} \cW_{A_1 \cup A_2}$, with $\cW_{A_1 \cup A_2}$ being according to Eq.~(\ref{eq:W-AUA'}). 

\smallskip\noindent For any $T \subseteq
[n]$, we have the following equality,
\begin{multline} \label{eq:inclusion-T}
   \{ (A_1, A_2) : T \subseteq A_1 \cup A_2 \}  = \cup_{T_1 \subseteq T} \: \{ (A_1, A_2) : A_1: T_1 \subseteq A_1 \text{ and } A_2: (T \setminus T_1) \subseteq A_2\}.
\end{multline}
The inclusion $\subseteq$ can be seen by noting that for $(A_1, A_2) :  T \subseteq A_1 \cup A_2$, and for $T_1 = T \cap A_1$, we have $T \setminus T_1 \subseteq A_2$.  The inclusion $\supseteq$ can be seen by noting that for $A_1: T_1 \subseteq A_1 \text{ and } A_2: (T \setminus T_1) \subseteq A_2$, we have $T \subseteq A_1 \cup A_2$.

\medskip\noindent For any fixed $T$, we now have that
\begin{align}
    & \dnorm{\sum_{A: T \subseteq A}  \big( \otimes_{i \in [n] \setminus A} \cI_i \big) \otimes \cW_A} \nonumber \\
    & = \dnorm{   \sum_{\substack{A_1, A_2 \\ T \subseteq A_1  \cup A_2}} \left(( \otimes_{i \in [n] \setminus A_1} \cI_i) \otimes {\cA}_{A_1} \right) \circ \left(( \otimes_{i \in [n] \setminus A_2} \cI_i ) \otimes \cB_{A_2} \right)} \nonumber \\
    & = \dnorm{\sum_{T_1 \subseteq T} (\sum_{A_1: T_1 \subseteq A_1} (\otimes_{i \in [n] \setminus A} \cI_i ) \otimes {\cA}_{A_1}) \circ (\sum_{A_2: (T \setminus T_1) \subseteq A_2} ( \otimes_{i \in [n] \setminus A} \cI_i \big) \otimes \cB_{A_2} )} \nonumber \\
   & \leq \sum_{T_1 \subseteq T} \nu_{\cA}^{|T_1|} {\nu_{\cB}}^{|T| - |T_1|} 
    \leq (\nu_{\cA} + \nu_{\cB})^{|T|}, 
\end{align}
where the first equality follows Eqs.~(\ref{eq:W-AUA'}), and (\ref{eq:W-A-set}), the second equality follows form Eq.~(\ref{eq:inclusion-T}), and the first inequality follows from Eq.~(\ref{eq:local-stc-2}) corresponding to channels $\cA$ and $\cB$. 
\end{proof}

\section*{Data availability}
This manuscript has no associated data.

\section*{Conflict of Interest}
Authors declare no competing interests.

\section*{Acknowledgements}
We would like to thank Paula Belzig, Alexander M\"uller-Hermes, Robert K\"onig, Thomas Theurer and Freek Witteveen for useful discussions.
MC and AG acknowledge financial support from the European Research Council (ERC Grant Agreement No.~818761), VILLUM FONDEN via the QMATH Centre of Excellence (Grant No.~10059) and the Novo Nordisk Foundation (grant NNF20OC0059939 ‘Quantum for Life’).  Part of this work was completed while MC was Turing Chair for Quantum Software, associated to the QuSoft research center in Amsterdam, acknowledging financial support by the Dutch National Growth Fund (NGF), as part of the Quantum Delta NL visitor programme.
OF acknowledges financial support from the European Research Council (ERC Grant, Agreement No.~851716) and from a government grant managed by the Agence Nationale de la Recherche under the Plan France 2030 with the reference ANR-22-PETQ-0006. We also thank the National Center for Competence in Research SwissMAP of the Swiss National Science Foundation and the Section of Mathematics at the University of Geneva for their hospitality.

\printbibliography

\begin{appendix}
\section{Shor error correction} \label{app:shorEC}
 Shor error correction~\cite{shor1996fault, tansuwannont2023adaptive} for a stabilizer code (one-to-one) $C$, correcting $t$ errors, consists of the following two steps,
\begin{enumerate}
    \item[1.] \textbf{Syndrome Extraction:}  The error syndrome is extracted by doing stabilizer generator measurements, using a fault-tolerant method based on ancilla qubits prepared in cat states (see Fig.~\ref{fig:shor-error-corr}). This method provides a transversal way to extract syndrome; therefore, ensuring that errors do not propagate drastically during the implementation.
    
    However, an error during the generator measurement can still lead to recording a wrong syndrome. To remedy this, multiple rounds of stabilizer generator measurements are performed. More precisely, measurements are performed until $t + 1$ successive rounds of generator measurements agree, that is, for some $i$, $h_{i} = \cdots = h_{i + t}$, where $h_i$ denote the bit string containing outcome of all the generator measurements on some $i$th round of generator measurement. Assuming that no more than $t$ errors occurred, there exists at least one round $ i \leq i' \leq i + t$ of the generator measurement, where no error occurred. Therefore, correct syndrome is extracted for the round $i'$. Also it suffices to do at most $(t + 1)^2$ rounds of generator measurements to find a set of $t + 1$ agreeing rounds.

\begin{figure}
\centering
\begin{quantikz}
\lstick[wires = 4]{Data qubits} &\qw &  \ctrl{4} &\qw & \qw & \qw  & \qw  \\
&\qw & \qw &\ctrl{4} & \qw & \qw  & \qw  \\
&\qw & \qw &\qw & \ctrl{4}  & \qw  & \qw  \\
&\qw & \qw &\qw & \qw  & \ctrl{4}  & \qw  \\[0.5cm]
\lstick[wires = 4]{Ancilla qubits \\[0.2cm] $\frac{\ket{0000} + \ket{1111}} {\sqrt{2}}$} &\gate{H} & \targ{} &\qw & \qw  & \qw  & \meter{} & \cw  \\
&\gate{H} & \qw &\targ{} & \qw  & \qw  & \meter{} & \cw \\
&\gate{H} & \qw &\qw & \targ{}  & \qw & \meter{} & \cw  \\
&\gate{H} & \qw &\qw & \qw  & \targ{} & \meter{} & \cw  \\
\end{quantikz}
\caption{Measurement of a weight four Pauli $Z$ stabilizer generator according to Shor error correction.}
\label{fig:shor-error-corr}
\end{figure} 
\item[2.] \textbf{Error correction:} A classical decoder takes as input the extracted syndrome and outputs a Pauli error correction, whose syndrome matches the extracted syndrome. 
    
\end{enumerate}

\section{Example of a representation} \label{app:ex-rep}
In this section, we give an example of representation of quantum channels in codes from Section~\ref{sec:rep}.

\paragraph{Calderbank-Shor-Steane (CSS) codes.} A CSS code is a stabilizer code with a generating set $G = G_X \cup G_Z$ such that any element in $G_X$ is given by a tensor product of Pauli $X$ operators and any element in $G_Z$ is given by a tensor product of Pauli $Z$ operators. Any CSS code of type $(n, 1)$ has the property that the logical $\CNOT$ gate between two code words is equal to transverse $\CNOT$, that is, 
\begin{equation} \label{eq:log-CNOT}
    \CNOT (\ket{\overline{x}_1,  \overline{x}_2})  = \ket{\overline{x}_1,  \overline{x_1 \oplus x_2}}, \forall x_1, x_2 \in \{0, 1\},
\end{equation}
where $\ket{\overline{\psi}}$ denotes the encoded state corresponding to $\ket{\psi}$.

\paragraph{Representation.} 
We shall consider a CSS code  $C$ of type $(n, 1)$ that corrects weight $t = 2$ errors. We consider the many-to-one code $D := \mathrm{Der}(C, \mathscr{E}(n, 1))$. Recall that $\mathscr{E}(n, 1)$ corresponds to the span of Pauli operators of weight one. Recall that $D = (L, \mu)$, where $L = \mathscr{E}(n, 1) \overline{L} \subseteq N:= (\mathbb{C}^2)^{\otimes n}$, where  $\overline{L}$ is the code space of $C$ (see Eq.~(\ref{eq:codespace-L})), and $\mu^* = \tr_F \circ \: \cU^\dagger, \: \cU^\dagger = U(\cdot)U^\dagger$, where $U$ is according to~(\ref{eq:U-unitary}).

\medskip Let $\cP: \bL((\mathbb{C}^2)^{\otimes 2}) \to \bL((\mathbb{C}^2)^{\otimes 2}), \cP = \CNOT(\cdot)\CNOT$ be the quantum channel realized by the $\CNOT$ gate.  Let $\overline{\CNOT}$ be the transverse $\CNOT$ gate between two groups of $n$ qubits. We take the error correction gadget $\mathrm{EC}$ corresponding to the Shor error correction (see also Appendix~\ref{app:shorEC}). Consider the circuit $\overline{\CNOT} \circ \mathrm{EC}^{\otimes 2}$~\cite{aharonov1997fault, Aliferis2006quantum}, and let $\cT$ be the corresponding quantum channel. We now show that $\cT$ represents $\cP$ in $D$.

\smallskip Using that $\mathrm{EC}$ maps any input to a code word of $C$, it follows that $\overline{\CNOT} \circ \mathrm{EC}^{\otimes 2}$ outputs states in the code  $C$. Therefore, it follows that $(\cJ^*)^{\otimes 2} \circ \cT$ is a quantum channel. Therefore, by taking $\cS := (\cJ^*)^{\otimes 2} \circ \cT$, Eq.~(\ref{eq:comm-3}) is satisfied, that is,
\begin{equation}
  (\cJ)^{\otimes 2} \circ \cS = \cT.
\end{equation}
Consider now input to $\overline{\CNOT} \circ \mathrm{EC}^{\otimes 2}$ in code space $L \otimes L$. Firstly, $\mathrm{EC}^{\otimes 2}$ corrects the weight one error on each of the incoming encoded states (recall that a state in $L$ is obtained by applying a weight one error on a state in $\overline{L}$) and then the $\overline{\CNOT}$ performs the logical $\CNOT$ between the code words of $C$. Therefore, using Eq.~(\ref{eq:log-CNOT}), we get
\begin{equation}
    (\mu^*)^{\otimes 2} \circ (\cJ^*)^{\otimes 2} \circ \cT \circ \cJ^{\otimes 2} = \cP \circ (\mu^*)^{\otimes 2}.
\end{equation}
By taking $ \cR = (\cJ^*)^{\otimes 2} \circ \cT \circ \cJ^{\otimes 2}$, Eqs.~(\ref{eq:comm-1}) and~(\ref{eq:comm-2}) are satisfied, that is, 
\begin{align}
    (\mu^*)^{\otimes 2} \circ \cR = \cP \circ (\mu^*)^{\otimes 2}. \\
    \cJ^{\otimes 2} \circ \cR = \cT \circ (\cJ^*)^{\otimes 2}.
\end{align}
Therefore, $\cT$ represents $\cP$ in $\cD$.

%For a construction of fault-tolerant scheme for circuit-level stochastic noise according to~\ref{def:ft_scheme_stoc}, we refer the reader to~\cite[Section~3.3]{christandl2024fault}

\section{Construction of a fault-tolerant scheme} \label{sec:cons-ft-scheme}
In this section, we construct a fault-tolerant scheme for general circuit noise according to Definition~\ref{def:ft_scheme}. Our construction relies on the fault-tolerant construction based on so called rectangles (Recs), proposed in~\cite{aharonov1997fault} and also given later in~\cite{Aliferis2006quantum}, from a slightly different perspective. We note that construction with Recs works for codes correcting at least two errors.  The fault-tolerant construction based on extended rectangles (Exrecs), proposed in~\cite{Aliferis2006quantum}, works for codes correcting $1$ error. We have made the choice to work with Recs for the purpose of simplicity of presentation so as to avoid complications that arise with overlapping exRecs.

\subsection{Concatenated codes}
Here, we describe the concatenated codes and the many-to-one codes obtained from concatenated codes. We take these many-to-one concatenated codes to be the sequence of quantum codes in Point $(a)$ of Def.~\ref{def:ft_scheme}.
\begin{definition}[Concatenated code~\cite{aharonov1997fault, Aliferis2006quantum}] \label{def:concate-code}
Consider a one-to-one code $C$ of type (n, 1), defined using an isometry $V : \mathbb{C}^2 \to (\mathbb{C}^2)^{ \otimes n}$. A concatenated code $C_r, r= 0, 1, \dots$, where $C_0$ is the trivial code defined by the identity $I : \mathbb{C}^2 \to \mathbb{C}^2$ and $C_1 = C$, is defined recursively by the  following isometry,
\begin{equation}
    V_r : \mathbb{C}^2 \to (\mathbb{C}^2)^{ \otimes (n^r) } ; \: V_r =   (V_{r-1})^{\otimes n} V_1.
\end{equation}
We refer to $C$ as the base code. The code space $\overline{L}_r$ of the concatenated code is given by the image of $V_r$.
\end{definition}
A concatenated code recursively encodes a single qubit using the base code $C$, where at a given level of recursion, each physical qubit of the code is replaced by a block of $n$ qubits. For example, $C_2$ contains $n$ blocks of qubits, where each block contains $n$ qubits. Similarly, an \emph{$r$-block}, that is, the set of physical qubits of the code $C_r$, contains $n$ $(r-1)$-blocks, with the $0$-block being a single qubit. 

\paragraph{Stabilizer code correcting two errors.} We take the base code $C$ to be a stabilizer code of type $(n, 1)$ that is defined by a set of stabilizer generators $\{ S_1, \dots, S_{n-1}\}$ and it corrects any Pauli error of weight $t = 2$.  By linearity, it follows that $C$ corrects errors from $\mathscr{E}(n, 2)$. 

\smallskip Consider the code space $\overline{L}$ of $C$ and consider the derived code $\mathrm{Der}(C, \mathscr{E}(n, 1))$, with  the codespace $L := \mathscr{E}(n, 1) \overline{L} = \bigoplus_\mathbf{s} H_\mathbf{s}$, where the direct sum is over the syndromes $\mathbf{s} \in \{ 0, 1 \}^{n-1}$, indicating at most one error (see also Section~\ref{sec:stab-CSS-codes}). Recall that $\mathrm{Der}(C, \mathscr{E}(n, 1))$ is a many-to-one quantum code $(L, \mu)$ with $L = U (\mathbb{C}^2 \otimes F)$, where
 $$F = \mathrm{span}\{ \ket{\mathbf{s}}: \mathbf{s} \in \{ 0, 1 \}^{n-1} \text{ indicating at most one error}\},$$ 
 and the unitary $U : \mathbb{C}^2 \otimes F \to L $ acts as follows for any Pauli error $E$ of weight at most $1$,
 \begin{equation} \label{eq:unitary-U}
        U^\dagger E\ket{\overline{\psi}} = \ket{\psi} \otimes \ket{\mathbf{s}}, 
 \end{equation}
 where $\ket{\overline{\psi}} = V \ket{\psi}$ is the encoded version of $\ket{\psi} \in \mathbb{C}^2$ and $\mathbf{s}$ is the syndrome corresponding to $E$. We now extend the unitary $U$ to the full space. Let $L_{\perp} = \bigoplus_{\mathbf{s}'} H_{\mathbf{s}'}$, where $\{\mathbf{s}'\}$ is the set of syndromes indicating two or more errors and define $$F_\perp :=  \mathrm{span}\{ \ket{\mathbf{s}'} : \mathbf{s}' \in \{0, 1\}^{n -1} \text{ indicating two or more errors} \}.$$
We have that $F \bigoplus F_\perp = (\mathbb{C}^2)^{ \otimes n-1}$. For each syndrome $\mathbf{s}'$ indicating two or more errors, we fix an arbitrary Pauli operator $E(\mathbf{s}')$, which has the syndrome $\mathbf{s}'$ and define the unitary $U_\perp: \mathbb{C}^2 \otimes F_\perp \to L_{\perp}$ as follows, 
 \begin{equation} 
        U_\perp^\dagger E(\mathbf{s}')\ket{\overline{\psi}} = \ket{\psi} \otimes \ket{\mathbf{s}'}.
 \end{equation}
We will later use an extension of $U$ to the whole space $(\mathbb{C}^2)^{\otimes n}$ as follows,
\begin{equation} \label{eq:u-tilde}
    \tilde{U}: \mathbb{C}^2 \otimes (\mathbb{C}^2)^{ \otimes n-1} \to (\mathbb{C}^2)^{\otimes n}, \: \tilde{U} = U \textstyle\bigoplus U_\perp.
\end{equation}
Using $\tilde{U}$, we define
\begin{equation} \label{eq:mu-tilde}
    \tilde{\mu} : \bL(\mathbb{C}^2) \to \bL((\mathbb{C}^2)^{\otimes n}),  \: \tilde{\mu} := \tilde{U} (\cdot \otimes \ident_{(\mathbb{C}^2)^{ \otimes n-1}})\tilde{U}^\dagger.
\end{equation}
Note that $\tilde{\mu}^* = \tr_{[n-1]} \circ \: \tilde{\cU}^\dagger$,\: $\tilde{\cU} := \tilde{U}(\cdot) \tilde{U}^\dagger$.

\smallskip Consider now the concatenated code $C_r, r= 0, 1, 2, \dots$ obtained using the stabilizer code $C$ as the base code.  To define correctable errors for $C_r$, we need a notion of sparse errors.
\begin{definition}[$(r,t)$  sparse error~\cite{aharonov1997fault}] \label{def:sparse-err}
 Let $E_r$ be a Pauli error on an $r$-block of the concatenated code $C_r$. We have that $E_r = E_{r-1}^1 \otimes \cdots \otimes E_{r-1}^n$, where $E_{r-1}^i, i \in [n]$ is a Pauli error acting on an $(r-1)$-block.
 
 The error $E_r$ is said to be $(r,t)$ sparse if it contains no more than $t$ $(r-1)$ blocks, such that the corresponding $E_{r-1}^i$ are not $(r - 1, t)$ sparse. A $(0,t)$ sparse error corresponds to the identity operator. 
\end{definition}
The correctable set of errors for $C_r$ is given by the span of Pauli errors that are $(r, 2)$ sparse~\cite{aharonov1997fault}. Here, we consider the set of errors that are $(r, 1)$ sparse, that is, 
\begin{equation} \label{eq:sparse-err}
 \mathscr{E}_r := \mathrm{span}\{ E_r : E_r \text{ is $(r, 1)$ sparse} \},    
\end{equation}
and the corresponding derived code $D_r := \mathrm{Der}(C_r, \mathscr{E}_r)$. We consider $(r, 1)$ sparse errors (rather than full $(r, 2)$ sparse) as then the derived code $D_r$ can still correct some additional errors. For example, $D_1$ can correct one additional error in the sense that if we have a code state of $D_1$ with one additional Pauli error applied on it, it is possible to recover the encoded information since the base code corrects two errors. This property is needed for fault-tolerant constructions with rectangles~\cite{aharonov1997fault, Aliferis2006quantum}.

\paragraph{Many to one concatenated codes.}  For an $r = 0, 1, \cdots$, we define the many-to-one code $D_r = (L_r, \mu_r)$, where
\begin{equation}
  L_r = \mathscr{E}_r \overline{L}_r \subseteq N_r := (\mathbb{C}^2)^{\otimes n^r}.
  \end{equation}
The map $\mu_r^*: \bL(L_r) \to \bL(\mathbb{C}^2)$ is obtained using the map  $\tilde{\mu}_r^*: \bL(N_r) \to \bL(\mathbb{C}^2)$, which is recursively defined 
 as,
\begin{equation} \label{eq:mu-tilde-r}
    \tilde{\mu}_r^* :=  \tilde{\mu}^*\circ 
 (\tilde{\mu}^*_{r-1})^{\otimes n} ,
\end{equation}
where $\tilde{\mu} : \bL(\mathbb{C}^2) \to \bL((\mathbb{C}^2)^{\otimes n})$ is according to Eq.~(\ref{eq:mu-tilde}). Note that $\tilde{\mu}^*_r $ is the same as the ideal $r$-decoder from~\cite[Page 43]{Aliferis2006quantum}.
We take the quantum channel $\mu^*_r$ to be the restriction of $\tilde{\mu}^*_r$ on the code space $L_r$, that is, 
\begin{equation} \label{eq:mu*_r}
  \mu^*_r = \tilde{\mu}^*_r \circ \cJ_r,
\end{equation}
where $\cJ_r : \bL(L_r) \to \bL(N_r)$ is the natural embedding. By applying $\tilde{\mu}^*$ recursively, it can be seen that for any $(r, 1)$ sparse error $E_r$, we have
\begin{equation} \label{eq:mu-r-st}
  \tilde{\mu}^*_r (E_r \proj{\overline{\psi}_r}E_r^\dagger) = \proj{\psi}, \: \forall \ket{\overline{\psi}_r} = V_r \ket{\psi}. 
\end{equation}
Since $\mu_r^*$ is defined using $\tilde{\mu}_r^*$, it would be useful to have the definition of representation in Def.~\ref{def:rep_codes} directly in terms of $\tilde{\mu}_r^*$ instead of $\mu_r^*$. This is given in the following remark.  
\begin{remark} \label{rem:alt-rep}
Consider many-to-one concatenated codes $D_r$ and $D_{r'}$. A quantum channel $\cT: \bL(N_r^{\otimes q})  \to \bL(N_{r'}^{\otimes q'})$ represents $\cP: \bL((\mathbb{C}^2)^{\otimes q}) \to  \bL((\mathbb{C}^2)^{\otimes q'})$ in $D_r, D_{r'}$ if $(\cJ_{r'}^*)^{\otimes q'} \circ \cT : \bL(N_r^{\otimes q})  \to \bL(L_{r'}^{\otimes q'}) $ is a quantum channel and the following equality holds,
\begin{equation} \label{eq:weak-rep}
    (\tilde{\mu}^*_{r'})^{\otimes q'} \circ \cT \circ \cJ_r^{\otimes q} =   \cP \circ (\tilde{\mu}^*_{r})^{\otimes q} \circ \cJ_r^{\otimes q}.
\end{equation}
 We note that the definition of representation according to Def.~\ref{def:rep_codes} is satisfied by taking $\cR =  (\cJ^*_{r'})^{\otimes q'} \circ \cT \circ \cJ_r^{\otimes q} $, and $\cS = (\cJ^*_{r'})^{\otimes q'} \circ \cT$.
\end{remark}
 We take the many-to-one concatenated codes $D_r, r = 0, 1, \dots $, where $r$ is the level of concatenation, to be the sequence of codes in Point (a) of Def.~\ref{def:ft_scheme}. Note that $D_r$ is a code of type $(n^r, 1)$, with $M_r = \mathbb{C}^2$, $L_r \subseteq N_r$, where $N_r = (\mathbb{C}^2)^{\otimes n^r}$. 
%\medskip The rest of this section will be devoted to proving the following theorem, which provides a fault-tolerant scheme according to Def.~\ref{def:ft_scheme}.
 In the following theorem, we show that part $(b)$ and $(c)$ of Def.~\ref{def:ft_scheme} is satisfied for codes $D_r, r= 1, 2, \dots $, considering parameter $k$ in Def.~\ref{def:ft_scheme} such that $r \sim \log k$. 
 %The parameter $r$ is related with $k$ from Def.~\ref{def:ft_scheme} as $r \sim \log k$.  The construction of 
\begin{theorem} \label{thm:cons-ft-scheme}
    Consider the sequence of many-to-one concatenated codes $D_r, r = 0, 1, \dots $. For a universal gate set $\mathbf{A}$, containing unitary gates including an identity gate, state preparation and measurement gates, there exist a family of circuits $\Psi_{g, r}$, such that, any $\tilde{\cT}_{\Psi_{g, r}} \in \trans(\Psi_{g, r}, \delta)$ represents $\cT_g$ in codes $D_r$, with accuracy $\epsilon_1 = \epsilon_2 = \epsilon_3 = O((c\delta)^{2^r})$. Furthermore, there exist a family of circuits $\Gamma_{r, s}$ such any $\tilde{\cT}_{\Gamma_{r, s}} \in \trans(\Gamma_{r, s}, \delta)$ represents the identity channel in codes $D_r, D_s$, with accuracy $\epsilon_1 = \epsilon_2 =O((c\delta)^{\min\{r, s\}})$, and $\epsilon_3 = O((c\delta)^{2^s})$.
\end{theorem}

\begin{proof}
    Proof of this theorem follows from Lemma~\ref{lem:rep-gates}, and Lemma~\ref{lem:interface}.
\end{proof}

%%%%%%%%%%%%%%%%%%%%%%%%%%%%%%%%%%%%%%%%%%%%%%%
%%%% Recursive simulation
%%%%%%%%%%%%%%%%%%%%%%%%%%%%%%%%%%%%%%%%%%%%%%%

\subsubsection{Recursive simulation}
In this section, we provide $\Psi_{g, r}$ for gates $g \in \mathbf{A}$ according to Point (b) of Def.~\ref{def:ft_scheme} for $D_r$.

\smallskip For a gate $g \in \mathbf{A}$, with $q$ qubits as input and $q'$ qubits as output, we consider a quantum circuit $\Phi_g$, that takes as input $q$ blocks of qubits, each containing $n$ qubits and outputs $q'$ blocks of qubits,  each containing $n$ qubits. Let $\mathrm{EC}$ be an error correction gadget, for example, gadget corresponding to the Shor error correction~\cite{shor1996fault, tansuwannont2023adaptive} (see also Appendix~\ref{app:shorEC}). We obtain the quantum circuit $\Psi_{g, r}$, using the quantum circuit $\Phi_g \circ (\mathrm{EC})^{\otimes q}$, as described in the following paragraph.

\paragraph{r-Rec.} We take $\Psi_{g, 1} := \Phi_g \circ (\mathrm{EC})^{\otimes q}$. Consider the set of locations $\mathrm{Loc}(\Psi_{g, 1}) := \{ g_{i, j}\}_{i, j}$, and the size $|\Psi_{g, 1}| := |\mathrm{Loc}(\Psi_{g, 1})|$ of the circuit $\Psi_{g, 1}$. The quantum circuit $\Psi_{g, r}$ is obtained by replacing each $g_{i, j} \in \mathrm{Loc}(\Psi_{g, 1})$ by the corresponding circuit $\Psi_{g_{i, j}, r-1}$. The quantum circuit $\Psi_{g, r}$ is referred to as an $r$-rectangle or simply an $r$-Rec corresponding to the gate $g$. 

\smallskip Note that $\Psi_{g, r}$ contains $|\Psi_{g, 1}|$ number of ($r-1$)-Recs. Furthermore, $\Psi_{g, r}$ takes as input $q$ blocks of $n^r$-qubits, and outputs $q'$ blocks of $n^r$-qubits. Note that for the identity gate $g = I$, the corresponding $1$-Rec is $\Psi_{I, 1} = \mathrm{EC}$. In the following, we denote $r$-Rec for the identity gate as $\Psi_{I, r} := \mathrm{EC}_r$, where $\mathrm{EC}_0 = I$.  

\paragraph{Faulty $r$-Rec.} \label{para:faulty-path} We say a gate $g \in \mathbf{A}$ is \emph{faulty}, if it realizes an arbitrary channel instead of the intended channel $\cT_g$. A faulty $r$-Rec $\Psi_{g, r}$ is defined by a \emph{fault path}, which is a subset of locations $\Delta_{g,r} \subseteq \mathrm{Loc}(\Psi_{g, r})$, such that the gates applied at those locations are faulty.  We denote a faulty $r$-Rec by $(\Psi_{g, r}, \Delta_{g, r})$, where $\Delta_{g, r}$ is the set of faulty locations. The $(r-1)$-Recs contained in $(\Psi_{g, r}, \Delta_{g, r})$ are also faulty, and they are given by,  $(\Psi_{g_{i, j}, r-1}, \Delta_{g_{i, j}, r-1})$, where $\Delta_{g_{i, j}, r-1} = \Delta_{g, r} \cap \mathrm{Loc}(\Psi_{g_{i, j}, r-1}), \: \forall g_{i, j} \in \mathrm{Loc}(\Psi_{g, 1})$.  

\smallskip An important property of a faulty $r$-Rec is its goodness, which is defined as below.

\begin{definition}[A good faulty $r$-Rec~\cite{Aliferis2006quantum}] \label{def:g-r-rec}
 A faulty $1$-Rec corresponding to $\Psi_{g, 1}$ is said to be good if it contains no more than one faulty gate in it.  A faulty $r$-Rec $\Psi_{g, r}$ is said to be good if it contains no more than one bad (not good) $(r-1)$-Recs in it. 
\end{definition}
Another important property of a faulty $r$-Rec is its correctness~\cite{Aliferis2006quantum}, which is defined as follows.
\begin{definition}[A correct $r$-rec~\cite{Aliferis2006quantum}] \label{def:cor-rec}
A faulty $r$-Rec $(\Psi_{g, r}, \Delta_{g, r})$, corresponding to a gate $g$, with $q$ qubit input and $q'$ qubit output, is said to be correct if the following holds,
\begin{equation} \label{eq:correctness}
    (\tilde{\mu}_r^*)^{\otimes q'} \circ \cT_{\Delta_{g, r}} \circ \cJ_r^{\otimes q} = \cT_g \circ   (\mu_r^*)^{\otimes q},
\end{equation}
where $\cT_{\Delta_{g, r}}$ is the channel realized by $(\Psi_{g, r}, \Delta_{g, r})$, and $\tilde{\mu}_r^*$ is according to~(\ref{eq:mu-tilde-r}).
\end{definition}

\smallskip We note that Eq.~(\ref{eq:correctness}) is equivalent to Eq.~(\ref{eq:weak-rep}), considering $\cT =\cT_{\Delta_{g, r}}$ and $\cP = \cT_g$. It is shown in~\cite{aharonov1997fault, Aliferis2006quantum} that it is possible to construct $r$-Recs, for a universal gate set $\mathbf{A}$, in a way that the following properties hold~\cite[Rec-Cor, Page 44]{Aliferis2006quantum}.

\begin{lemma}{\normalfont(\cite[Rec-Cor]{Aliferis2006quantum}).} \label{lem:cor-rec}
A good $r$-Rec $(\Psi_{g, r}, \Delta_{g, r})$ is correct.  
\end{lemma}

\begin{lemma} {\normalfont(\cite[Rec-Val]{Aliferis2006quantum}).} \label{lem:valid-rec}
 The channel $\cT_{\Delta_{g, r}}$ corresponding to a good $r$-Rec $(\Psi_{g, r}, \Delta_{g, r})$ maps any state in $\bL(N_r)$ to a state in the subspace $\bL(L_r) \subseteq \bL(N_r)$, that is, $(\cJ_r^*)^{\otimes q'} \circ \cT_{\Delta_{g, r}}$ is a quantum channel, with $q'$ being the number of output qubits of $g$.
\end{lemma}
Lemmas~\ref{lem:cor-rec} and~\ref{lem:valid-rec} imply the following.
\begin{lemma} \label{lem:good-cor}
    For a good $r$-Rec $(\Psi_{g, r}, \Delta_{g, r})$, the corresponding quantum channel $\cT_{\Delta_{g, r}}$ represents $\cT_g$ in the code $D_r$. 
\end{lemma}
\begin{proof}
 The correctness condition in Lemma~\ref{lem:cor-rec} gives Eq.~(\ref{eq:weak-rep}). From Lemma~\ref{lem:valid-rec}, we have that $(\cJ_r^*)^{\otimes q'} \circ \cT_{\Psi_{g, r}}$ is a quantum channel. Therefore, conditions in Remark~\ref{rem:alt-rep} are satisfied.  
\end{proof}
The following lemma proves that an $r$-Rec $\Psi_{g, r}$ satisfies $(b)$ of Def.~\ref{def:ft_scheme}.

%A proof of this is provided in~\cite[Lemma 37]{christandl2024fault}, which uses similar techniques as in~\cite[Lemma 10.10]{kitaev} and~\cite[Lemma 11]{aharonov1997fault}. 

\begin{lemma} \label{lem:rep-gates}
Consider an $r$-Rec $\Psi_{g, r}$ corresponding to $g \in \mathbf{A}$, with $q, q'$ input and output qubits, respectively. There exists a constant $c > 1$  such that for all $\delta \in [0, \frac{1}{c})$, any $\tilde{\cT} \in \trans(\Psi_{g, r}, \delta)$ represents $\cT_g$ in the code $D_r$ with accuracy $\epsilon_1 = \epsilon_2 = \epsilon_3 = O((c \delta)^{2^r})$.
\end{lemma}
%%%%%%%%%%%%%%%%%%%%%%%%%%%%%%%%%%%%%%%%%%%%%%
% Proof
%%%%%%%%%%%%%%%%%%%%%%%%%%%%%%%%%%%%%%%%%%%%%%
\begin{proof}
 We will find maps $\cT, \cR$, and $\cS$ for $\tilde{\cT}_{\Psi_{g, r}} \in \trans(\Psi_{g, r}, \delta)$, such that the commutative diagram in Fig.~\ref{fig:strong-rep} holds exactly with respect to the code $D_r$ and $\cP = \cT_g$, and the superoperator $\cT $ satisfies, $\dnorm{\cT - \tilde{\cT}} \leq O((c \delta)^{2^r})$. Then, it follows from Part $(a)$ of Lemma~\ref{lem:rep-approx} that $\tilde{\cT}$ represents $\cT_g$ in the code $D_r$ with accuracy $\epsilon_1 = \epsilon_2 = \epsilon_3 = O((c \delta)^{2^r})$.

\smallskip  Consider the set of locations $\mathrm{Loc}(\Psi_{g,r}) = \{ g_{i, j} \}_{i , j}$ of the circuit $\Psi_{g,r}$. Consider an $\tilde{\cT}_{\Psi_{g, r}} \in \trans(\Psi_{g, r}, \delta)$. Note that $\tilde{\cT}_{\Psi_{g, r}}$ is obtained by concatenating $\{\tilde{\cT}_{g_{i, j}}\}$ in the structure of the quantum circuit $\Psi_{g, r}$, where each gate  $g_{i, j}$ is replaced by an $\tilde{\cT}_{g_{i, j}} \in \trans(g_{i, j}, \delta)$. We now write $\tilde{\cT}_{g_{i, j}}$  for all $g_{i, j} \in \mathrm{Loc}(\Psi_{g, r})$ as follows~(see also \cite[Lemma 10.10]{kitaev}):
\begin{equation}
    \tilde{\cT}_{g_{i, j}} = \cT_{g_{i, j}} + (\tilde{\cT}_{g_{i, j}} - \cT_{g_{i, j}}).
\end{equation}
In this way, we can write $\tilde{\cT}_{\Psi_{g, r}}$ as sum of terms, where each term contains either $\cT_{g_{i, j}}$ (perfect channel) or $(\tilde{\cT}_{g_{i, j}} - \cT_{g_{i, j}})$ (faulty map) at a location $(i, j)$. We note that $(\tilde{\cT}_{g_{i, j}} - \cT_{g_{i, j}})$ is not necessarily a quantum channel contrary to the definition of a faulty gate in Paragraph ``Faulty $r$-Rec" in Section~\ref{para:faulty-path}. Let $\cT_{\Delta_{g, r}}$ be the superoperator realized by the faulty $r$-Rec $(\Psi_{g, r}, \Delta_{g, r})$, where faulty locations $g_{i, j} \in \Delta_{g, r}$ realize $(\tilde{\cT}_{g_{i, j}} - \cT_{g_{i, j}})$ instead of $\cT_{g_{i, j}}$. Then, we have
\begin{equation} \label{eq:tilde-T-g-1}
  \tilde{\cT}_{\Psi_{g, r}}   = \sum_{ \Delta_{g, r} \subseteq \mathrm{Loc}(\Psi_{g, r})}  \cT_{\Delta_{g, r}}.
\end{equation}
Note that
\begin{equation} \label{eq:dnorm-g,r}
    \dnorm{ \cT_{\Delta_{g, r}}} \leq \delta^{|\Delta_{g, r}|}. 
\end{equation}
We now expand $\cT_{\Delta_{g, r}}$ in terms of good and bad fault paths,  
 \begin{equation} \label{eq:tilde-T-g}
  \tilde{\cT}_{\Psi_{g, r}}   = \sum_{ \text{good }\Delta_{g, r} }  \cT_{\Delta_{g, r}} + \sum_{\text{bad }\Delta_{g, r}}  \cT_{\Delta_{g, r}}.
\end{equation}
where a good $\Delta_{g, r}$ refers to $(\Psi_{g, r}, \Delta_{g, r})$ being good and a bad $\Delta_{g, r}$ refers to $(\Psi_{g, r}, \Delta_{g, r})$ being bad (not good) according to Def.~\ref{def:g-r-rec}. We take $\cT$ as follows,
\begin{equation} \label{eq:T-map}
\cT   =  \sum_{ \text{good }\Delta_{g, r}}  \cT_{\Delta_{g, r}}.
\end{equation}
We now prove that $\cT$ is $O((c \delta)^{2^r})$ close to $ \tilde{\cT}_{\Psi_{g, r}}$ in the diamond norm for a constant $c > 0$. The proof is similar to~\cite[Lemma 11]{aharonov1997fault}. Using Eqs.~(\ref{eq:tilde-T-g}) and~(\ref{eq:T-map}), and the triangle inequality, we get
\begin{equation} \label{eq:norm-diff-bound}
    \dnorm{\tilde{\cT}_{\Psi_{g, r}} - \cT} \leq  \sum_{ \text{bad }\Delta_{g, r}}  \dnorm{\cT_{\Delta_{g, r}}}.
\end{equation}
Let $\tau_{g, r} := \sum_{ \text{bad }\Delta_{g, r}}  \dnorm{\cT_{\Delta_{g, r}}}$. We upper bound $\tau_{g, r}$, recursively.  From Def.~\ref{def:g-r-rec}, we get
 \begin{equation} \label{eq:g-1-bad}
     \{ \Delta_{g, 1} \subseteq \mathrm{Loc}(\Psi_{g, 1}) : \Delta_{g, 1} \text{ is bad}  \} =    \{ \Delta_{g, 1} \subseteq \mathrm{Loc}(\Psi_{g, 1}) : |\Delta_{g, 1}| > 1\}
 \end{equation}
 Let $g_{\max} \in \mathbf{A}$ be the gate with the largest $1$-Rec  $\Psi_{g_{\max}, 1}$ in the gate set $\mathbf{A}$ and take $c :=  e \binom{|\Psi_{g_{\max}, 1}|}{2}$ ($e$ being the Euler's constant). Then, we have that
\begin{align} 
    \tau_{g, 1} &\leq \sum_{ \substack{\Delta_{g, 1} \subseteq \mathrm{Loc}(\Psi_{g, 1}) \\ |\Delta_{g, 1}| > 1}}  \delta^{|\Delta_{g, 1}|} \nonumber \\
    & = \sum_{ j > 1 } \binom{|\Psi_{g_{\max}, 1}|}{j} \delta^j \nonumber \\
    & \leq (1 + \delta)^{|\Psi_{g_{\max}, 1}|} \sum_{j > 1} \binom{|\Psi_{g_{\max}, 1}|}{j} \big(\frac{\delta}{ 1 + \delta} \big)^j  \big(\frac{1}{ 1 + \delta} \big)^{|\Psi_{g_{\max}, 1}| -j}\nonumber \\
    & \leq e^{|\Psi_{g_{\max}, 1}| \delta} \binom{|\Psi_{g_{\max}, 1}|}{2} \delta^2 \leq e \binom{|\Psi_{g_{\max}, 1}|}{2} \delta^2 = \frac{(c\delta)^2}{c}, \label{eq:t-g,1}
\end{align}
where the first inequality uses Eqs.~(\ref{eq:dnorm-g,r}) and (\ref{eq:g-1-bad}), the third inequality uses $(1 + \delta)^x \leq e^{\delta x}$, and $\sum_{j > t} \binom{n}{j} (1-\delta)^{n -j} \delta^j  \leq \binom{n}{t+1} \delta^{t+1}$, and the last inequality uses $|\Psi_{g_{\max}, 1}| \delta \leq 1$ (using $c\delta \leq 1$).

\medskip We have from Def.~\ref{def:g-r-rec},
 \begin{multline} \label{eq:g-r-bad}
     \{ \Delta_{g, r} \subseteq \mathrm{Loc}(\Psi_{g, r}) : \Delta_{g, r} \text{ is bad}  \} \\=    \{ \Delta_{g, r} \subseteq \mathrm{Loc}(\Psi_{g, r}) : (\Psi_{g, r} ,\Delta_{g, r}) \text{ contains more than one bad faulty $(r-1)$-Rec} \}
 \end{multline}
Note that $\dnorm{\cT_{\Delta_{g, r}}}  \leq \prod_{g_{i, j} \in  \mathrm{Loc}(\Psi_{g, 1})} \dnorm{\cT_{\Delta_{g_{i,j}, r-1}}}$, where $\Delta_{g_{i,j}, r-1} = \Delta_{g, r} \cap  \mathrm{Loc}(\Psi_{g_{i, j}, r-1})$. Therefore, using Eq.~(\ref{eq:g-r-bad}), we get
\begin{align}
    \tau_{g, r} &\leq  \sum_{ \text{bad }\Delta_{g, r}} \: \prod_{g_{i, j} \in  \mathrm{Loc}(\Psi_{g, 1})} \dnorm{\cT_{\Delta_{g_{i,j}, r-1}}} \nonumber\\
    & \leq \sum_{\substack{\Delta_{g, 1} \subseteq \mathrm{Loc}(\Psi_{g, 1}) \\ |\Delta_{g, 1}| > 1}} \: \prod_{ g_{i, j} \in \Delta_{g, 1}} \tau_{g_{i, j}, r-1} \label{eq:t-g-r}
\end{align}
For $r = 2$,  Eq.~(\ref{eq:t-g-r}) gives $\tau_{g, 2} \leq \frac{(c \delta)^{2^2}}{c}$. Applying Eq.~(\ref{eq:t-g-r}) recursively, we get
\begin{equation} \label{eq:bad-fault-path-r}
    \tau_{g, r} \leq \frac{(c \delta)^{2^r}}{c}.
\end{equation}
Therefore, from Eqs.~(\ref{eq:norm-diff-bound}) and (\ref{eq:bad-fault-path-r}), we have $\dnorm{\tilde{\cT}_{\Psi_{g, r}} - \cT} \leq O((c \delta)^{2^r})$. We now show that there exist maps $\cR$ and $\cS$ such that the commutative diagram in Fig.~\ref{fig:strong-rep} holds exactly, with $\cT$ from Eq.~(\ref{eq:T-map}),
and $\cP = \cT_g$. Similar to $\cT_{\Delta_{g, r}}$, we need another notation; for a fault path $\Delta_{g, r} \subseteq \mathrm{Loc}(\Psi_{g, r})$, we define $\tilde{\cT}_{\Delta_{g, r}}$ to be the channel realized by $(\Psi_{g, r}, \Delta_{g, r})$, where each faulty location $g_{i, j} \in \Delta_{g, r}$ realizes $\tilde{\cT}_{g_{i, j}} \in \trans(g_{i, j}, \delta)$ instead of $\cT_{g_{i, j}}$.  It can be seen that
\begin{equation}
    \cT_{\Delta_{g, r}} = \sum_{\Delta'_{g, r} \subseteq \Delta_{g, r}} (-1)^{|\Delta_{g, r}| - |\Delta'_{g, r}|} \: \tilde{\cT}_{\Delta'_{g, r}}.
\end{equation}
Therefore, From Eq.~(\ref{eq:T-map}),
\begin{equation} \label{eq:T-map-1}
\cT   =  \sum_{ \text{good }\Delta_{g, r}}  \: \sum_{\Delta'_{g, r} \subseteq \Delta_{g, r}} (-1)^{|\Delta_{g, r}| - |\Delta'_{g, r}|} \: \tilde{\cT}_{\Delta'_{g, r}}.
\end{equation}
Note that any subset $\Delta'_{g, r} \subseteq \Delta_{g, r}$ of a good $\Delta'_{g, r}$ is also good. Using Lemma~\ref{lem:good-cor}, $\tilde{\cT}_{\Delta'_{g, r}}$ corresponding to a good $\Delta'_{g, r}$ represents $\cT_g$. Therefore, there exist $\cR_{\Delta'_{g, r}}$ and $\cS_{\Delta'_{g, r}}$ such that the commutative diagram in Fig.~\ref{fig:strong-rep} holds exactly, with $\cT_{\Delta'_{g, r}}$ and $\cP = \cT_g$. We simply take,
\begin{align}
    \cR = \sum_{ \text{good }\Delta_{g, r}}  \: 
 \sum_{\Delta'_{g, r} \subseteq \Delta_{g, r}} (-1)^{|\Delta_{g, r}| - |\Delta'_{g, r}|} \: \cR_{\Delta'_{g, r}}. \label{eq:R}\\
      \cS =  \sum_{ \text{good }\Delta_{g, r}}  \:  \sum_{\Delta'_{g, r} \subseteq \Delta_{g, r}} (-1)^{|\Delta_{g, r}| - |\Delta'_{g, r}|} \: \cS_{\Delta'_{g, r}} \label{eq:S}.
\end{align}
Using the fact that Eq.~(\ref{eq:comm-2}) and  Eq.~(\ref{eq:comm-3}) are satisfied individually for $\cT_{\Delta'_{g, r}}$, $\cR_{\Delta'_{g, r}}$, and $\cS_{\Delta'_{g, r}}$, it follows that they are also satisfied for $\cT$, $\cR$ and $\cS$. We now show that Eq.~(\ref{eq:comm-1}), that is, ${\mu_r^*}^{\otimes q'} \circ \cR = \cT_g  \circ {\mu_r^*}^{\otimes q}$. Recall that $q, q'$ are the number of input and output qubits of $g$.

\medskip \noindent For a good $\Delta_{g, r}$, we have  
\begin{align}
    {\mu_r^*}^{\otimes q'} \circ \cR &= \sum_{ \text{good }\Delta_{g, r}}  \: 
 \sum_{\Delta'_{g, r} \subseteq \Delta_{g, r} } (-1)^{|\Delta_{g, r}| - |\Delta'_{g, r}|} {\mu_r^*}^{\otimes q'} \circ \cR_{\Delta'_{g, r}} \nonumber \\
    &= \Big(\sum_{ \text{good }\Delta_{g, r}}  \:  \sum_{\Delta'_{g, r} \subseteq \Delta_{g, r}} (-1)^{|\Delta_{g, r}| - |\Delta'_{g, r}|} \Big) \cT_g \circ {\mu_r^*}^{\otimes q}, \label{eq:R-g-r}
\end{align}
where in the second equality, we have used ${\mu_r^*}^{\otimes q'} \circ \cR_{\Delta'_{g, r}} = \cT_g \circ {\mu_r^*}^{\otimes q}$. We have that 
\begin{equation} \label{eq:sum-exp-1}
    \sum_{\Delta'_{g, r} \subseteq \Delta_{g, r}} (-1)^{|\Delta_{g, r}| - |\Delta'_{g, r}|} = \sum_{j = 0}^{|\Delta_{g, r}|} \binom{|\Delta_{g, r}|}{j} (-1)^{|\Delta_{g, r}| - j},
\end{equation}
which implies
\begin{equation} \label{eq:sum-exp-2}
  \sum_{\Delta'_{g, r} \subseteq \Delta_{g, r}} (-1)^{|\Delta_{g, r}| - |\Delta'_{g, r}|}   = \begin{cases}
        1, & \text{ if $|\Delta_{g, r}| = 0$} \\
        0, & \text{ otherwise. }
    \end{cases}
\end{equation}
Therefore, From Eqs.~(\ref{eq:R-g-r}) and~\ref{eq:sum-exp-2}, we get 
 ${\mu_r^*}^{\otimes q'} \circ \cR = \cT_g  \circ {\mu_r^*}^{\otimes q}$.
\end{proof}
%%%%%%%%%%%%%%%%%%%%%%%%%%%%%%%%%%%%%%%%%%%%%%
% Proof
%%%%%%%%%%%%%%%%%%%%%%%%%%%%%%%%%%%%%%%%%%%%%%

%%%%%%%%%%%%%%%%%%%%%%%%%%%%%%%%%%%%%%%%%%%%%
%%%% Recursive simulation
%%%%%%%%%%%%%%%%%%%%%%%%%%%%%%%%%%%%%%%%%%%%%

\subsection{Interface}
In this section, we provide the interfaces $\Gamma_{r, s}, r, s = 0, 1, \dots$ for $D_r, D_s$, as in Part $(c)$ of Def.~\ref{def:ft_scheme}. To this end, we will first construct $\Gamma_{r, r+1}$ and $\Gamma_{r+1, r}$, $r = 0, 1, \dots$, and then compose them to obtain $\Gamma_{r, s}$.

\subsubsection*{Interface $\Gamma_{r, r+1}$}
 Consider a quantum circuit $\Gamma$ in the gate set $\mathbf{A}$, with single qubit input and $n$-qubit output, such that it realises the following quantum channel, 
 \begin{equation}
     \cT_{\Gamma} = U (\cdot \otimes \proj{0}^{\otimes n-1})U^\dagger,
 \end{equation}
where $U$ is according to Eq.~(\ref{eq:unitary-U}). Noting that zero syndrome $\mathbf{s} = 0^{n-1}$ corresponds to $E = I$ in Eq.~(\ref{eq:unitary-U}), it follows $\cT_{\Gamma}$ maps any input state $\ket{\psi} \in \mathbb{C}^2$ to the corresponding encoded state in the base code $C$.

\smallskip  We define $\Gamma_{0, 1} := \mathrm{EC} \circ \Gamma$. It can be easily seen that $\Gamma_{0, 1}$ represents the identity in the code $D_0, D_1$. The quantum circuit $\Gamma_{r, r+1}$ is obtained by replacing each $g_{{i, j}} \in \mathrm{Loc}(\Gamma_{0, 1})$ by the corresponding $r$-Rec $\Psi_{g_{i,j}, r}$. Note that 
\begin{equation} \label{eq:gamma-r-r+1}
  \Gamma_{r, r+1} := \mathrm{EC}_{r+1} \circ \Gamma_{r}, 
\end{equation}
 where $\Gamma_{r}$ is obtained by replacing each location in the quantum circuit $\Gamma$ by the corresponding $r$-Rec. We denote by $(\Gamma_{r, r+1}, \Delta_{r, r+1})$, the interface $\Gamma_{r, r+1}$ with faults at locations in  $\Delta_{r, r+1} \subseteq \mathrm{Loc}(\Gamma_{r, r+1})$. 
 
 \paragraph{A good $(\Gamma_{r, r+1}, \Delta_{r, r+1})$.} We say $(\Gamma_{r, r+1}, \Delta_{r, r+1})$ is good if all the $r$-Recs contained in it are good according to Def.~\ref{def:g-r-rec}, that is, for all $g_{{i, j}} \in \mathrm{Loc}(\Gamma_{0, 1})$, the corresponding $(\Psi_{g_{i, j}, r}, \Delta_{g_{i,j}, r})$, where $\Delta_{g_{i,j}, r} := \Delta_{r, r+1} \cap \mathrm{Loc}(\Psi_{g_{i, j}, r})$, is good, otherwise it is bad (not good).
 
\medskip We have the following lemma for the interface $\Gamma_{r, r+1}$.
\begin{lemma} \label{lem:del-rep-r}
There exists a constant $c > 1$  such that for all $\delta \in [0, \frac{1}{c})$, any $\tilde{\cT}_{\Gamma_{r, r + 1}} \in \trans(\Gamma_{r, r + 1}, \delta)$ represents the identity channel in codes $D_r, D_{r+1}$ with accuracy $\epsilon_1 = \epsilon_2 =  O(|\Gamma_{0, 1}| \delta_r)$ and $\epsilon_3 = O(\delta_{r + 1})$, where $\delta_r := (c \delta)^{2^r}$.
\end{lemma}
\begin{proof}
 Firstly, using Remark~\ref{rem:alt-rep}, we show that $\cT_{\Delta_{r, r+1}}$ corresponding to a good $(\Gamma_{r, r+1}, \Delta_{r, r+1})$ represents the identity channel in $D_r, D_{r+1}$.  

\smallskip Using Eq.~(\ref{eq:gamma-r-r+1}), $(\Gamma_r, \Delta_r)$ and $(\mathrm{EC}_{r+1}, \Delta_{r+1})$ (see also Eq.~(\ref{eq:gamma-r-r+1})), where $\Delta_r =  \Delta_{r, r+1}  \cap \mathrm{Loc}(\Gamma_r)$, and $\Delta_{r+1} =  \Delta_{r, r+1}  \cap \mathrm{Loc}(\mathrm{EC}_{r+1})$.  
Since  $(\Gamma_{r, r+1}, \Delta_{r, r+1})$ is good, it follows that the all the rectangles in $(\Gamma_r, \Delta_r)$ and $(\mathrm{EC}_{r+1}, \Delta_{r+1})$ are good. Let $\cT_{\Delta_r}$ and $\cT_{\Delta_{r+1}}$ be realizations of $(\Gamma_r, \Delta_r)$, and $(\mathrm{EC}_{r+1}, \Delta_{r+1})$, respectively. Note that $\cT_{\Delta_{r, r+1}} =  \cT_{\Delta_{r+1}} \circ \cT_{\Delta_r}$.

\smallskip Using Lemma~\ref{lem:good-cor} for  $(\mathrm{EC}_{r+1}, \Delta_{r+1})$, it follows that $\cJ^*_{r + 1} \circ \cT_{\Delta_{r+1}} : \bL(N_{r+1}) \to \bL(L_{r+1})$ is a quantum channel. Therefore, $\cJ_{r + 1}^* \circ \cT_{\Delta_{r, r+1}} : \bL(N_{r}) \to \bL(L_{r + 1})$ is a quantum channel. It remains to show Eq.~(\ref{eq:weak-rep}), that is, 
 \begin{equation}
     \tilde{\mu}^*_{r + 1} \circ \cT_{\Delta_{r, r+1}} \circ \cJ_r =  \tilde{\mu}^*_{r}  \circ \cJ_r,
 \end{equation}
Applying Lemma~\ref{lem:good-cor} for all the rectangles in $(\Gamma_{r, r+1}, \Delta_{r, r+1})$ and using the composability of representation according to Lemma~\ref{lem:union-bound-rep}, it follows that $\cT_{\Gamma_{r, r+1}}$ represents $\cT_{\Gamma_{0, 1}}$ in $D_r$. Therefore, we have
\begin{equation} \label{eq:1-T}
    (\tilde{\mu}_r^*)^{\otimes n} \circ \cT_{\Delta_{r, r+1}} \circ \cJ_r = \cT_{\Gamma_{0, 1}} \circ \tilde{\mu}^*_r \circ \cJ_r.
\end{equation} 
Furthermore, as $\cT_{\Gamma_{0, 1}}$ represents identity channel in $D_0, D_1$, we get 
\begin{equation} \label{eq:2-T}
    \tilde{\mu}^* \circ \cT_{\Gamma_{0, 1}} = \cI.
\end{equation}
We now have,
\begin{align}
\tilde{\mu}^*_{r + 1} \circ \cT_{\Gamma_{r, r+1}} \circ \cJ_r &= \tilde{\mu}^* \circ (\tilde{\mu}_r^*)^{\otimes n} \circ \cT_{\Gamma_{r, r+1}} \circ \cJ_r \nonumber \\
&= \tilde{\mu}^* \circ \cT_{\Gamma_{0, 1}} \circ \tilde{\mu}^* \circ \cJ_r \nonumber \\
& = \tilde{\mu}^*_r \circ \cJ_r, 
\end{align}
where the first equality uses Eq.~(\ref{eq:mu-tilde-r}), the second equality uses Eq.~(\ref{eq:1-T}), and the last equality uses Eq.~(\ref{eq:2-T}).

\smallskip The remaining proof can be done similarly to Lemma~\ref{lem:rep-gates}, using that a good $(\Gamma_{r, r+1}, \Delta_{r, r + 1})$ represents identity channel in $D_r, D_{r+1}$.  Similarly to Lemma~\ref{lem:rep-gates}, we can show that there exist maps $\cT, \cR$, and $\cS$ for $\tilde{\cT}_{\Gamma_{r, r+1}} \in \trans(\Gamma_{r, r+1}, \delta)$, such that the commutative diagram in Fig.~\ref{fig:strong-rep} holds exactly for codes $D_r, D_{r+1}$ and $\cP = \cI$, and the superoperator $\cT $ satisfies, $\dnorm{\cT - \tilde{\cT}} \leq O(|\Gamma_{0, 1}|\delta_r)$. Then, using part $(a)$ of Lemma~\ref{lem:rep-approx}, it follows that $\tilde{\cT}_{\Gamma_{r, r + 1}}$ represents $\cI$ in codes $D_r, D_{r+1}$, with accuracy $\epsilon_1 = \epsilon_2 = \epsilon_3 = O(|\Gamma_{0, 1}|\delta_r)$.

\smallskip To obtain $\cT$, we expand $\tilde{\cT}_{\Gamma_{r, r + 1}}$ in terms of good $\Delta_{r, r+1}$ and bad $\Delta_{r, r+1}$ as in Eq.~(\ref{eq:tilde-T-g}), with $\cT_{\Delta_{r, r+1}}$ being the superoperator realized by $(\Gamma_{r, r+1}, \Delta_{r, r+1})$, where any faulty location $g_{i, j} \in \Delta_{r, r+1}$ realizes $\tilde{\cT}_{g_{i, j}} - \cT_{g_{i, j}}$ instead of $\cT_{g_{i, j}}$. We take $\cT$ to be the sum of terms corresponding to the good $\Delta_{r, r + 1}$ as in Eq.~(\ref{eq:T-map}).

\medskip Finally, we can show that $\epsilon_3 = O(\delta_{r+1})$. Using Eq.~(\ref{eq:gamma-r-r+1}), we write $\tilde{\cT}_{\Gamma_{r, r+1}} = \tilde{\cT}_{\mathrm{EC}_{r+1}} \circ \tilde{\cT}_{\Gamma_{r}}$, where $\tilde{\cT}_{\mathrm{EC}_{r+1}} \in \trans(\mathrm{EC}_{r+1}, \delta)$ and $\tilde{\cT}_{\Gamma_{r}} \in \trans(\Gamma_{r}, \delta)$. From Lemma~\ref{lem:rep-gates}, $\tilde{\cT}_{\mathrm{EC}_{r+1}}$ represents the identity channel in the code $D_{r+1}$ with accuracy $O(\delta_{r+1})$. Therefore, there exists an $\cS_ {\mathrm{EC}_{r+1}}: \bL(N_{r+1}) \to \bL(L_{r+1})$
such that $\dnorm{\cJ_{r + 1} \circ \cS_{\mathrm{EC}_{r+1}} - \tilde{\cT}_{\mathrm{EC}_{r+1}}} \leq  O( \delta_{r+1})$. We take $\cS_{\Gamma_{r, r+1}} = \cS_{\mathrm{EC}_{r+1}} \circ \tilde{\cT}_{\Gamma_{r}}$, which is a quantum channel from $\bL(N_{r})$ to $\bL(L_{r+1})$. It follows that
\begin{align}
    \dnorm{\cJ_{r + 1} \circ \cS_{\Gamma_{r, r+1}} - \cT_{\Gamma_{r, r+1}}} \leq \dnorm{\cJ_{r + 1} \circ \cS_{\mathrm{EC}_{r+1}} - \tilde{\cT}_{\mathrm{EC}_{r+1}}}  \leq O( \delta_{r+1}).
\end{align} 
\end{proof}

\subsubsection*{Interface $\Gamma_{r+1, r}$}
Consider a quantum circuit $\Gamma$ in the gate set $\mathbf{A}$, with $n$ qubit input and single qubit output, such that it realises the quantum channel $\tilde{\mu}^*$,
\begin{equation} \label{eq:gamma-dec}
    \cT_{\Gamma} := \tilde{\mu}^*.
\end{equation}
We compose $\Gamma$ with the quantum circuit $\mathrm{EC}$, and define $\Gamma_{1, 0} :=  \Gamma \circ \: (\otimes_{i = 1}^n \mathrm{EC}_0)$ (Recall that $\mathrm{EC}_0$ refers to the identity gate). Let $\Gamma_{r+1, r}$ be the quantum circuit obtained by replacing each location $g_{i, j}$ in $\Gamma_{1, 0}$ by the corresponding $\Psi_{g_{i, j}, r}$. We have
\begin{equation} \label{eq:gamma-r+1,r}
 \Gamma_{r+1, r} := \Gamma_r  \circ (\otimes_{i = 1}^n \mathrm{EC}^i_r),  
\end{equation}
where $\Gamma_r$ is obtained by replacing each location $g_{i, j}$ in $\Gamma$ by the corresponding $\Psi_{g_{i, j}, r}$.  We denote by $(\Gamma_{r+1, r}, \Delta_{r+1, r})$, the interface $\Gamma_{r + 1, r}$ with faults at location $\Delta_{r+1, r} \subseteq \mathrm{Loc}(\Gamma_{r + 1, r})$.  Similarly to the interface $\Gamma_{r, r+1}$ in the previous section, we say $(\Gamma_{r+1, r}, \Delta_{r+1, r})$ is good if all the $r$-Recs contained in it are good according to Def.~\ref{def:g-r-rec}, otherwise it is bad.

\medskip For the quantum circuit $\Gamma_{r + 1, r}$, we prove Lemma~\ref{lem:del-rep-r+1}, which is analogous to Lemma~\ref{lem:del-rep-r} for the interface $\Gamma_{r, r+1}$. Note that Lemma~\ref{lem:del-rep-r+1} considers representation in terms of codes $D_{r+1}, D_r$.  As good $(\mathrm{EC}^i_r, \Delta_r^i)$ represents identity in $D_r$, the channel $\cT$ corresponding to good $(\otimes_{i = 1}^n \mathrm{EC}^i_r, \: \cup_{i=1}^n \Delta^i_r)$ maps any state in $L_{r+1}$ to a state in $(\otimes_{i=1}^n L_r) \subseteq L_{r + 1}$. In the following lemma, we show that this mapping is faithful in the sense that the logical information encoded in the input in $L_{r + 1}$ is preserved. In another word, good $(\otimes_{i = 1}^n \mathrm{EC}^i_r, \: \cup_{i=1}^n \Delta^i_r)$ is a correct $r + 1$-Rec corresponding to the identity gate, according to Def.~\ref{def:cor-rec}. This property will be later needed in the proof of Lemma~\ref{lem:del-rep-r+1}.

\begin{lemma} \label{lem:n-EC}
Consider the faulty circuit $(\otimes_{i = 1}^n \mathrm{EC}^i_r, \: \cup_{i=1}^n \Delta^i_r)$, where $\Delta^i_r$ denotes the set of faulty locations in $\mathrm{EC}^i_r$, such that $(\mathrm{EC}^i_r, \: \Delta^i_r)$ is good for all $i \in [n]$. Then, the channel $\cT$ realized by $(\otimes_{i = 1}^n \mathrm{EC}^i_r, \: \cup_{i=1}^n \Delta^i_r)$   satisfies 
\begin{equation}
  \tilde{\mu}^*_{r+1} \circ \cT \circ \cJ_{r + 1} = \tilde{\mu}^*_{r+1} \circ \cJ_{r+1}.
\end{equation}
\end{lemma}
\begin{proof}
Consider an $(r+ 1, 1)$ sparse Pauli error $E_{r + 1} := E_r^1 \otimes \cdots \otimes E_r^n \in \mathscr{E}_{r + 1}$, where Pauli errors $E_i^r$ act on $r$-blocks. From Def.~\ref{def:sparse-err}, it follows that except for an $i_0 \in [n]$, all $E_i^r \in \cE_r$ (i.e., $(r, 1)$-sparse).

\medskip Consider now a state $\ket{\psi} \in \mathbb{C}^2$, and let $\ket{\overline{\psi}_{r+1}} = V_{r+1} \ket{\psi}$ be the encoding of $\ket{\psi}$ in code $C_{r+1}$. Note that 
\begin{equation}
    V_{r + 1} \ket{\psi} = V_r^{\otimes n} V_1 \ket{\psi}
    = V_r^{\otimes n} \ket{\overline{\psi}_1}.
\end{equation}
Therefore, $\ket{\overline{\psi}_{r+1}}$ is also the encoding of $\ket{\overline{\psi}_1}$ in the code $C_r$. Consider now $E_{r+1} \ket{\overline{\psi}_{r+1}} \in L_{r+1}$. We have 
\begin{align} \label{eq:i-not}
    E_{r+1} \ket{\overline{\psi}_{r+1}} = (\otimes_{i = 1}^n E^i_r V^i_r) \ket{\overline{\psi}_1} \in (N_r)_{i_0} \otimes_{i \in [n]\setminus {i_0}} L_r
\end{align}
For any $i \in [n]$, $\cT_{\Delta^i_r}$ corresponding to a good $(\mathrm{EC}^i_r, \: \Delta^i_r)$ represents identity in $D_r$. Therefore, from Remark~\ref{rem:alt-rep}, the channel $\cT_{\Delta^i_r}$ satisfies, 
\begin{align} \label{eq:EC-i-r}
    \tilde{\mu}_r^* \circ \cT_{\Delta^i_r} \circ \cJ_r &= \tilde{\mu}_r^* \circ \cJ_r. 
\end{align}
Let $\overline{\rho}_{r+1} := E_{r+1} \proj{\overline{\psi}_{r+1}} E_{r+1}$ and let $(\mu_r^*)_i$ denote $\mu_r$ corresponding to $i \in [n]$. Then, we have 
\begin{align}
    (\tilde{\mu}^*_r)^{\otimes n} \circ \cT (\overline{\rho}_{r+1}) &= \left( (\tilde{\mu}_r^*)_{i_0} \otimes (\otimes_{i \in [n]\setminus {i_0}}(\tilde{\mu}_r^*)_{i})  \circ (\otimes_{i = 1}^n \cT_{\Delta^i_r})\right)(\overline{\rho}_{r+1})  \nonumber \\
    & = \big(( (\tilde{\mu}^*_r)_{i_0} \circ \cT_{\Delta^{i_0}_r}) \otimes (\otimes_{i \in [n]\setminus {i_0}} (\tilde{\mu}_r^*)_{i})\big) (\overline{\rho}_{r+1}) \nonumber\\
    &=  \big( \cW_{i_0} \otimes (\otimes_{i \neq i_0} \cI_{i}) \big)  (\proj{\overline{\psi}_1}),  \label{eq:w-i-0}
\end{align}
where the second equality uses  $\overline{\rho}_{r+1} \in \bL((N_r)_{i_0}) \otimes  \bL(\otimes_{i \in [n]\setminus {i_0}} L_r)$ (see Eq.~(\ref{eq:i-not})) and Eq.~(\ref{eq:EC-i-r}). The last equality uses that $(\tilde{\mu}_r^*)_i, i \in [n]$ map the state $(\overline{\rho}_{r+1})$ to $\proj{\overline{\psi}_1}$, with the $i_0$th qubit being in error on which channel $\cW_{i_0} :=  (\tilde{\mu}^*_r)_{i_0} \circ \cT_{\Delta^{i_0}_r} \circ (E_r^{i_0} V_r^{i_0})(\cdot) (E_r^{i_0} V_r^{i_0})^\dagger$ acts.

\smallskip We now have
\begin{align}
 \tilde{\mu}^*_{r+1} \circ \cT(\overline{\rho}_{r+1}) 
 &=  \tilde{\mu}^* \circ (\tilde{\mu}^*_r)^{\otimes n} \circ \cT   (\overline{\rho}_{r+1}) \nonumber\\
 & = \tilde{\mu}^* \big((\cW_{i_0} \otimes (\otimes_{i \neq i_0} \cI_{i}))  (\proj{\overline{\psi}_1}) \big) \nonumber\\
 &= \proj{\psi} \nonumber \\
 & = \tilde{\mu}^*_{r+1} (\overline{\rho}_{r+1}), \label{eq:mu-psi}
\end{align}
where the first equality uses Eq.~(\ref{eq:mu-tilde-r}), the second equality uses Eq.~(\ref{eq:w-i-0}), and the third equality uses that  $\big((\cW_{i_0} \otimes (\otimes_{i \neq i_0} \cI_{i}))  (\proj{\overline{\psi}_1})$ is an encoding of $\ket{\psi}$ in the code $D_1$ ($r = 1$).

\medskip We have that Eq.~(\ref{eq:mu-psi}) holds for $\overline{\rho}_{r+1} := E_{r+1} \proj{\overline{\psi}_{r+1}} E_{r+1}$, corresponding to any $\ket{\psi}$ and any $(r+1, 1)$ sparse Pauli error $E_{r+1}$. We may extend Eq.~(\ref{eq:mu-psi}) for any (non-Pauli error) $ \sum_j\alpha_j E_{r+1}^j$, where $\{E_{r+1}^j\}_j$ are $(r+1, 1)$ sparse errors, using $\mu^*_{r+1} = \tr_{[n-1]} \circ \: \tilde{\cU}^\dagger$, where  $\tilde{\cU}^\dagger = \tilde{U}^\dagger (\cdot) \tilde{U}$ is a unitary and linearity of unitary operators. Therefore, Eq.~(\ref{eq:mu-psi}) holds for any state in $D_{r+1}$, implying
\begin{align}
  \tilde{\mu}^*_{r+1} \circ \cT \circ \cJ_{r + 1} = \tilde{\mu}^*_{r+1} \circ \cJ_{r+1}.
\end{align} 
\end{proof}

\begin{lemma} \label{lem:del-rep-r+1}
There exists a constant $c > 1$  such that for all $\delta \in [0, \frac{1}{c})$, any $\tilde{\cT}_{\Gamma_{r+1, r}} \in \trans(\Gamma_{r+1, r}, \delta)$ represents the identity channel in codes $D_{r+1}, D_r$, with accuracy $\epsilon_1 = \epsilon_2 = \epsilon_3 = |\Gamma_{0, 1}| O(\delta_r)$, where $\delta_r := (c \delta)^{2^r}$.
\end{lemma}

\begin{proof}
Let $\cT_{\Delta_{r+1, r}}$ be the channel realized by $(\Gamma_{r+1, r}, \Delta_{r+1, r})$, where locations in $\Delta_{r+1, r}$ realize arbitrary channels. We show that $\cT_{\Delta_{r+1, r}}$ corresponding to a good $(\Gamma_{r+1, r}, \Delta_{r+1, r})$ represents the identity channel in $D_{r+1}, D_{r}$, using which, the remaining proof can be done similarly to Lemma~\ref{lem:rep-gates}, as sketched in the proof of Lemma~\ref{lem:del-rep-r}.

\smallskip Using Lemma~\ref{lem:valid-rec} for $r$-Recs in $(\Gamma_{r+1, r}, \Delta_{r+1, r})$, it follows that $\cJ_r^* \circ \cT_{\Delta_{r+1, r}}$ is a quantum channel. From Remark~\ref{rem:alt-rep}, it remains to show that  Eq.~(\ref{eq:weak-rep}) holds for $D_r, D_{r+1}$, and $\cP := \cI$, that is,
\begin{equation}
 \tilde{\mu}^*_r \circ  \cT_{\Gamma_{r+1, r}} =  \tilde{\mu}^*_{r+1} \circ \cJ_{r+1}.   
\end{equation}
As $(\Gamma_{r+1, r}, \Delta_{r+1, r})$ is good, the corresponding $(\Gamma_r, \Delta_r)$ and $(\mathrm{EC}_r, \Delta^i_r), i \in [n]$ (see also Eq.~(\ref{eq:gamma-r+1,r})), where $\Delta_r =  \Delta_{r+1, r}  \cap \mathrm{Loc}(\Gamma_r)$ and $\Delta^i_r =  \Delta_{r+1, r}  \cap \mathrm{Loc}(\mathrm{EC}^i_r)$, are also good. Let $\cT_{\Delta_r}$ and $\cT_{\Delta^i_r}$ be channels corresponding to $(\Gamma_r, \Delta_r)$ and $(\mathrm{EC}^i_r, \Delta^i_r)$, respectively. Note that $\cT_{\Delta_{r+1, r}} = \cT_{\Delta_r} \circ \cT$, where $\cT := \otimes_{i = 1}^n \cT_{\Delta^i_r}$.

\medskip Using Lemma~\ref{lem:good-cor} for $r$-Recs in $(\Gamma_r, \Delta_r)$, it follows that $\cT_{\Delta_r}$ represents $\cT_\Gamma = \tilde{\mu}^*$ in $D_r$. Therefore, from Remark~\ref{rem:alt-rep}, we have
\begin{align}
    \tilde{\mu}^*_r \circ  \cT_{\Delta_r} \circ \cJ_r^{\otimes n} &= \cT_\Gamma \circ (\tilde{\mu}^*_{r})^{\otimes n} \circ \cJ_r^{\otimes n}\\
    &= \tilde{\mu}^*_{r+1} \circ \cJ_r^{\otimes n}. \label{eq:gamma-r-1}
\end{align}
Similarly using Lemma~\ref{lem:good-cor} for $r$-Recs in $((\otimes_{i = 1}^n\mathrm{EC}_r), \cup_{i = 1}^n \Delta^i_r)$,  it follows $\cT = \otimes_{i = 1}^n \cT_{\Delta^i_r}$ represents identity in the code $D_r$. Therefore,
\begin{equation}
    (\tilde{\mu}^*_r)^{\otimes n} \circ  \cT \circ \cJ_r^{\otimes n} = (\tilde{\mu}^*_r)^{\otimes n} \circ  \cJ_r^{\otimes n} \label{eq:EC-r-2}
\end{equation}
As the output of $\cT$ is in the subspace $L_r^{\otimes n} \subseteq N_{r}^{\otimes n}$, we get from Eq.~(\ref{eq:gamma-r-1}),
\begin{equation}
     \tilde{\mu}^*_r \circ  \cT_{\Delta_r}  \circ \cT = \tilde{\mu}^*_{r+1} \circ \cT,
\end{equation}
which implies 
\begin{equation}
     \tilde{\mu}^*_r \circ \cT_{\Delta_{r+1, r}} \circ \cJ_{r + 1}= \tilde{\mu}^*_{r+1} \circ \cT \circ \cJ_{r + 1}.
\end{equation}
Applying Lemma~\ref{lem:n-EC} on the RHS, we get
\begin{equation}
  \tilde{\mu}^*_r \circ  \cT_{\Delta_{r+1, r}} \circ \cJ_{r + 1} =  \tilde{\mu}^*_{r+1} \circ \cJ_{r+1}. 
\end{equation}
Therefore, $\cT_{\Delta_{r+1, r}}$ represents identity in $D_{r+1}, D_{r}$. 
\end{proof}

\subsubsection*{Interface $\Gamma_{r, s}$}
We define $\Gamma_{r, s}$ as follows:
\begin{equation} \label{eq:gamma-r-s}
 \Gamma_{r, s} :=    \begin{cases}
    \Gamma_{r, r+1} \circ \Gamma_{r+1, r+2} \circ \cdots \circ \Gamma_{s-1, s} & \text{ if $0 \leq r < s$} \\
    \Gamma_{r, r-1} \circ \Gamma_{r-1, r-2} \circ \cdots \circ \Gamma_{s + 1, s} & \text{ if $r > s \geq 0$}.
\end{cases}   
\end{equation}

\begin{lemma} \label{lem:interface}
There exists a constant $c > 1$  such that for all $\delta \in [0, \frac{1}{c})$, any  $\tilde{\cT} \in \trans(\Gamma_{r, s}, \delta)$ represents the identity channel in  $D_r, D_s$, with accuracy $\epsilon_1 = \epsilon_2 =  O( (c\delta)^{2^l} )$, where $l = \min\{r, s\}$, and $\epsilon_3 = O((c\delta)^{2^s})$.
\end{lemma}

\begin{proof}
Using Lemmas~\ref{lem:del-rep-r}, \ref{lem:del-rep-r+1}, and \ref{lem:union-bound-rep}, we get $\epsilon_3 = O((c\delta)^{2^s})$, and  for a constant $b > 0$, $\epsilon = \epsilon_1 = \epsilon_2$ satisfies,
\begin{align}
    \epsilon &\leq b(c\delta)^{2^l} (1 + (c\delta)^2 + (c\delta)^{2^2} + \cdots + (c\delta)^{2^{|r-s|}}  ) \nonumber \\
    & \leq \frac{b}{1 - (c\delta)^2}  (c\delta)^{2^l} = O((c\delta)^{2^l}),  \label{eq:eps-c-del} 
\end{align}
where the second inequality uses
 $1 + x^2 + x^4 + x^6 + \dots = \frac{1}{1-x^2}$, and the last equality uses that for any $c\delta < 1$.
\end{proof}
%%%%%%%%%%%%%%%%%%%%%%%%%%%%%%%%%%%%%%%%%%%%%%
%%% Fault-tolerance scheme for circuit-level stochastic noise
%%%%%%%%%%%%%%%%%%%%%%%%%%%%%%%%%%%%%%%%%%%%%%

\subsection{Fault-tolerance scheme for circuit-level stochastic noise} \label{sec:const-ft-scheme-stc}
In this section, we consider circuit-level stochastic noise from Section~\ref{sec:q_circuit-noise}. We will use the notation $\tilde{\cT}_{\Phi}$ to denote the quantum channel realized by a noisy quantum circuit $\Phi$, where each location $g_{i, j}$ in the circuit realizes $\tilde{\cT}_{g_{i, j}}$ according to circuit-level stochastic noise with noise rate $\delta$. Recall that $\tilde{\cT}_{g_{i, j}}$ applies the channel $\cT_{g_{i, j}}$ corresponding to $g_{i, j}$ with probability $(1-\delta)$, and with probability $\delta$, it applies an error channel (see Eq.~\ref{eq:noisy-g}).

\smallskip The following lemma gives Part $(b)$ of the fault-tolerant scheme from Def.~\ref{def:ft_scheme_stoc}.

\begin{lemma} \label{lem:rep-gates-stc}
Consider $r$-Rec $\Psi_{g, r}$ for $g \in \mathbf{A}$. There exists a constant $c > 1$  such that for all $\delta \in [0, \frac{1}{c})$, the channel $\tilde{\cT}_{\Psi_{g, r}}$ satisfies,
\begin{equation}
    \tilde{\cT}_{\Psi_{g, r}} = (1 - \epsilon) \overline{\cT}_{\Psi_{g, r}} + \epsilon \: \overline{\cZ}_{\Psi_{g, r}},
\end{equation}
where $\epsilon = O((c\delta)^{2^r})$,  $\overline{\cT}_{\Psi_{g, r}}$ represents $\cT_g$ in $D_r$, and $\overline{\cZ}_{\Psi_{g, r}}$ is some channel.
\end{lemma}

\begin{proof}
 For an $\Delta_{g, r} \subseteq \mathrm{Loc}(\Psi_{g, r})$, let $\cT_{\Delta_{g, r}}$ be the channel corresponding to the faulty $r$-Rec $(\Psi_{g, r}, \Delta_{g, r})$, such that any location in $g_{i, j} \in \Delta_{g, r}$ realizes some error channel instead of $\cT_{g_{i, j}}$. Then, $\tilde{\cT}_{\Psi_{g, r}}$ can be expanded as
\begin{equation}
    \tilde{\cT}_{\Psi_{g, r}} = \sum_{\Delta_{g, r} \subseteq \mathrm{Loc}(\Psi_{g, r})} (1 - \delta)^{|\Psi_{g, r}| - |\Delta_{g, r}|} \: \delta^{|\Delta_{g, r}|} \: \cT_{\Delta_{g, r}}.
\end{equation}
The probability associated with a fault path $\Delta_{g, r}$ is given by,
\begin{equation} \label{eq:p-delta}
    p (\Delta_{g, r}) = (1 - \delta)^{|\Psi_{g, r}| - |\Delta_{g, r}|} \: \delta^{|\Delta_{g, r}|} 
\end{equation}
Let $\epsilon$ be the probability that a bad fault path $\Delta_{g, r}$ occurs, that is, $\epsilon := \sum_{ \text{bad } \Delta_{g, r}} \: p(\Delta_{g, r})$. Then, from Eq.~(\ref{eq:p-delta}), we have
\begin{equation}
   \tilde{\cT}_{\Psi_{g, r}} = (1 - \epsilon) \overline{\cT}_{\Psi_{g, r}} + \epsilon \: \overline{\cZ}_{\Psi_{g, r}}, 
\end{equation}
where   
\begin{align}
  \overline{\cT}_{\Psi_{g, r}} &:= \sum_{ \text{good } \Delta_{g, r} } \frac{p(\Delta_{g, r})}{1 - \epsilon} \cT_{\Delta_{g, r}} \\  
  \overline{\cZ}_{\Psi_{g, r}} &:= \sum_{ \text{bad } \Delta_{g, r} } \frac{p(\Delta_{g, r})}{\epsilon} \cT_{\Delta_{g, r}}.
\end{align}
Using Lemma~\ref{lem:cor-rec}, $\cT_{\Delta_{g, r}}$ for good fault path $\Delta_{g, r}$ represents $\cT_g$ in  $D_r$; therefore, $\overline{\cT}_{\Psi_{g, r}}$ represents $\cT_g$ in  $D_r$. It remains to show that $\epsilon = O((c \delta)^{2^r})$. It can be seen using recursively the fact that $(\Psi_{g, r}, \Delta_{g, r})$ is bad if it contains more than one bad $(r-1)$ Recs, and choosing $c = \binom{ \max_{g \in \mathbf{A}} \{ \: |\Psi_{g, 1}| \}}{2}$.
\end{proof}
The following lemma gives Part $(c)$ of the fault-tolerant scheme from Def.~\ref{def:ft_scheme_stoc}.

\begin{lemma} 
Consider the circuit $\Gamma_{r, s}$ from Eq.~(\ref{eq:gamma-r-s}). There exists a constant $c > 1$  such that for all $\delta \in [0, \frac{1}{c})$, the noisy realization $\tilde{\cT}_{\Gamma_{r, s}}$ is given by,
\begin{equation}
    \tilde{\cT}_{\Gamma_{r, s}} = (1 - \epsilon) \overline{\cT}_{\Gamma_{r, s}} + \epsilon\: \overline{\cZ}_{\Gamma_{r, s}} ,
\end{equation}
where $\epsilon = O((c\delta)^{2^s})$, and 
$\cJ_s^* \circ \overline{\cT}_{\Gamma_{r, s}}$ is a quantum channel, and
\begin{equation}
     \overline{\cT}_{\Gamma_{r, s}} = (1 - \epsilon') \overline{\cT}'_{\Gamma_{r, s}} + \epsilon'\: \overline{\cZ}'_{\Gamma_{r, s}},
\end{equation}
where $\epsilon' = O((c\delta)^{2^l})$,
with $l = \min\{r, s\}$, and $\overline{\cT}_{\Gamma_{r, s}}$ represents identity channel in $D_r, D_s$. Here, $\overline{\cZ}_{\Gamma_{r, s}}$, and $\overline{\cZ}'_{\Gamma_{r, s}}$ are some channels.
\end{lemma}

\begin{proof}
We shall consider $\Gamma_{r, s}, 0 \leq r < s$, that is, $ \Gamma_{r, s} := \Gamma_{r, r+1} \circ \Gamma_{r+1, r+2} \circ \cdots \circ \Gamma_{s-1, s}$. The proof for $\Gamma_{r, s},  r > s \geq 0$ can be done similarly.

\smallskip Consider the quantum circuit $\Gamma_{r, r + 1} := \mathrm{EC}_{r+1} \circ \Gamma_{r}$, according to Eq.~(\ref{eq:gamma-r-r+1}).  Let $\cT_{\Delta_{r, r+1}}$ be the channel realized by $( \Gamma_{r, r+1}, \Delta_{r, r+1})$, where each location in $g_{i, j} \in \Delta$ realizes some error channel instead of $\cT_{g_{i, j}}$. Let $p(\Delta_{r, r+1}) := (1 - \delta)^{|\Gamma_{r, r + 1}| - |\Delta_{r, r + 1}|} \: \delta^{|\Delta_{r, r+1}|} $ be the probability associated with the fault path $\Delta_{r, r+1}$.

\smallskip  Let $\Delta_{r+1} := \Delta_{r, r+1} \cap \mathrm{Loc}(\mathrm{EC}_{r+1})$, and let $\epsilon$ be the probability that $\Delta_{r+1}$ (that is, $(\mathrm{EC}_{r+1}, \Delta_{r+1})$) is a bad $(r+1)$-Rec (meaning that it contains at least two bad (not good) $r$-Recs). We can expand $\tilde{\cT}_{\Gamma_{r, r + 1}}$, the noisy realization of $\Gamma_{r, r+1}$ under the stochastic noise as, 
\begin{equation} \label{eq:r-r+1-exp}
     \tilde{\cT}_{\Gamma_{r, r + 1}} = (1 - \epsilon) \overline{\cT}_{\Gamma_{r, r + 1}} +  \epsilon \: \overline{\cZ}_{\Gamma_{r, r + 1}}, 
\end{equation}
where
\begin{equation} 
    \overline{\cT}_{\Gamma_{r, r + 1}} := \sum_{\Delta_{r, r+1}: \Delta_{r+1} \text{ is good}} \: \frac{p(\Delta_{r, r +1})}{1 - \epsilon} \cT_{\Delta_{r, r +1}} 
\end{equation}
Using $\cJ_{r + 1}^* \circ \cT_{\Delta_{r, r+1}}$ is a quantum channel for a  good $(\mathrm{EC}_{r+1}, \Delta_{r+1})$, we have that $\cJ_{r + 1}^* \circ \overline{\cT}_{\Delta_{r, r+1}}$ is a quantum channel. From Lemma~\ref{lem:rep-gates-stc}, we get $\epsilon = O((c\delta)^{2^{r + 1}})$.

\smallskip Consider the sets $A := \{ \Delta_{r,r+1} : (\mathrm{EC}_{r+1}, \Delta_{r+1}) \text{ is a good } (r+1)\text{-Rec} \}$ and $B :=\{ \Delta_{r, r + 1}: ( \Gamma_{r, r+1}, \Delta_{r, r+1}) \text{ is good}\}$. We have $B \subseteq A$ since $( \Gamma_{r, r+1}, \Delta_{r, r+1})$ is good when all the $r$-Recs in  $( \Gamma_{r, r+1}, \Delta_{r, r+1})$ are good. Consider 
\begin{equation}
 C := \{\Delta_{r, r+1}: \text{$(\Gamma_{r, r+1}, \Delta_{r, r+1})$ is bad but $(\mathrm{EC}_{r+1}, \Delta_{r+1})$ is a good $(r+1)$-Rec}\} = A \setminus B.   
\end{equation}
 The associated probabilities satisfy $p(B) + p(C) = p(A) = 1 - \epsilon$. We expand $\overline{\cT}_{\Gamma_{r, r + 1}}$ as,
\begin{equation} \label{eq:r-r+1-exp-1}
     \overline{\cT}_{\Gamma_{r, r + 1}} = (1 -  \epsilon') \overline{\cT}'_{\Gamma_{r, r + 1}} +  \epsilon' \: \overline{\cZ}'_{\Gamma_{r, r + 1}}, 
\end{equation}
where 
\begin{equation}
\epsilon' = \frac{p(C)}{1 - \epsilon} \text{ and }    \overline{\cT}'_{\Gamma_{r, r + 1}} := \sum_{ \Delta_{r, r+1} \in B} \: \frac{p(\Delta_{r, r +1})}{(1 - \epsilon - p(C))} \cT_{\Delta_{r, r +1}}.
\end{equation}
As shown in Lemma~\ref{lem:del-rep-r}, $\cT_{\Delta_{r, r+1}}$ for a good $\Delta_{r, r +1}$ represents identity in $D_r, D_{r+1}$. Therefore, $\overline{\cT}'_{\Gamma_{r, r + 1}}$ represents identity in $D_r,  D_{r+1}$. Using  Lemma~\ref{lem:rep-gates-stc} and a union bound union bound over the $r$-Recs inside $ \Gamma_{r,r+1}$, we get $ p(C) \leq p(\Delta_{r, r + 1} \text{ is bad}) = O(|\Gamma_{0, 1}| (c\delta)^{2^r})$. This implies that $\epsilon' = O(|\Gamma_{0, 1}| (c\delta)^{2^r})$.

\smallskip Finally, doing the above procedure for $\Gamma_{r', r'+1}, \: r \leq r' < s$, we can find expansions for $\Gamma_{r, s}$ as in Eqs.~(\ref{eq:r-r+1-exp}) and~(\ref{eq:r-r+1-exp-1}), with $\epsilon = O((c\delta)^{2^s})$, and 
$$ \epsilon' \leq b |\Gamma_{0, 1}| (c\delta)^{2^r} (1 + (c\delta)^2 + (c\delta)^{2^2} + \cdots + (c\delta)^{2^{|r-s|}}  ) = O((c\delta)^{2^r}).$$ 
\end{proof}

%%%%%%%%%%%%%%%%%%%%%%%%%%%%%%%%%%%%%%%%%%%%%
%% Fault-tolerance scheme for circuit-level stochastic noise
%%%%%%%%%%%%%%%%%%%%%%%%%%%%%%%%%%%%%%%%%%%%%
\section{The toric code as a communication code} \label{sec:toric-code}
In this section, we give an example of a code that can correct local stochastic noise while having a sublinear minimum distance. We consider the family of toric codes~\cite{kitaev, kitaev2003fault, bravyi1998quantum, dennis2002topological} and  show that they define an $(m, \zeta)$ communication code for any local stochastic channel with parameter $\alpha \leq \frac{1}{36}$, such that $\zeta$ goes to zero as the code length increases~\cite{preskil2011toric}. Therefore, from Theorem~\ref{thm:ft-code-stc}, one can obtain a fault-tolerant code using toric code given $\alpha \geq (2\delta' + \nu)$, with parameters $\delta', \nu$ as in Theorem~\ref{thm:ft-code-stc}.
\begin{figure}[!b]
    \centering
\begin{tikzpicture}
    \draw (0, 0) grid (5, 5);
    \fill (canvas cs:x=1cm,y=1.5cm) node[left](){$Z$} circle (2pt);
    \fill (canvas cs:x=1.5cm,y=2.0cm) node[above](){$Z$} circle (2pt);
    \fill (canvas cs:x=2.0cm,y=1.5cm) node[right](){$Z$} circle (2pt);
    \fill (canvas cs:x=1.5cm,y=1.0cm) node[below](){$Z$} circle (2pt);

\fill (canvas cs:x=4.5cm,y=4cm) node[above](){$X$} circle (2pt);
\fill (canvas cs:x=4.0cm,y=4.5cm) node[left](){$X$} circle (2pt);
\fill (canvas cs:x=3.5cm,y=4.0cm) node[below](){$X$} circle (2pt);
\fill (canvas cs:x=4.0cm,y=3.5cm) node[right](){$X$} circle (2pt);

\draw[ thick] (3, 0) to (3, 5) ;
\draw[thick] (0, 3) to (5, 3) ;
\draw[dotted,  thick] (4.5, 0) to (4.5, 5);
\draw[dotted,  thick] (0, 4.5) to (5, 4.5);

\draw (3, 0) node[below](){$\overline{Z}_1$}
(5, 3) node[right](){$\overline{Z}_2$}
(5, 4.5) node[right](){$\overline{X}_1$}
(4.5, 0) node[below](){$\overline{X}_2$}
;
\end{tikzpicture}
\caption{Toric code : the figure shows $X$ and $Z$ type stabilizer generators of the toric code and logical operators $\overline{Z}_1, \overline{X}_1$ and $\overline{Z}_2, \overline{X}_2$ corresponding to two logical qubits.} 
\label{fig:toric-code}
\end{figure}

\paragraph{Encoding.} Consider logical states $\ket{\overline{x_1 x_2}} \in (\mathbb{C}^2)^{\otimes 2L^2}, \: x_1, x_2 \in \{0, 1\}$, of the toric code, with square lattice of length $L$ (see Fig.~\ref{fig:toric-code}). Note that $\ket{\overline{x_1 x_2}}$ is eigenstate of all the $X$ and $Z$ type generators with eigevalue $+1$, and eigenstate of logical operators $\overline{Z}_1$ and $\overline{Z}_2$, with eigenvalue $x_1$ and  $x_2$, respectively. Therefore, encoding 
of the toric code is given by $\cE := V (\cdot) V^\dagger$, where $V: ((\mathbb{C})^2)^{\otimes 2} \to (\mathbb{C}^2)^{\otimes 2L^2}$ is an isometry as follows,
\begin{align}
    V \ket{x_1 x_2} = \ket{\overline{x_1 x_2}}, \: \forall x_1 ,x_2 \in \{0, 1\}^2.
\end{align}

\paragraph{Decoding.} To define the decoding map $\cD: \bL((\mathbb{C}^2)^{\otimes 2L^2}) \to \bL((\mathbb{C}^2)^{\otimes 2})$ for the toric code, we first need to consider its error correction procedure for Pauli errors.  As the toric code is a CSS code, Pauli $Z$ and $X$ errors are corrected  separately.  Let $[2L^2]$ be the set of all edges on the toric code lattice. Consider a Pauli $Z$ error, acting on the edges in $E \subseteq [2L]^2$.  Then the error correction is done in the following two steps,
\begin{enumerate}
    \item[$(i)$] First the measurement corresponding to the $X$ type generators are performed. The generator measurement reveals the boundary $B$ of the $E$. 

    \item[$(ii)$] The classical minimum weight perfect matching (MWPM) algorithm then determines the minimum weight chain $E'$ with the boundary $B$. We then apply Pauli $Z$ operator on the qubits corresponding to the edges in $E'$.
\end{enumerate}
As the chains $E$ and $E'$ have the same boundary, $E + E'$ is a cycle, i.e, it has no boundary; therefore consists of loops.  This means that the state after the correction is again in the code space. The correction is successful if $E + E'$ consists of homologically trivial loops (i.e. can be continuously shrunk to a point, see also Fig.~\ref{fig:non-trivial-cycle}). If $E + E'$ contains a non-trivial cycle (see also Fig.~\ref{fig:non-trivial-cycle}), it corresponds to a logical $Z$ error.

\begin{figure}[!t]
    \centering
\begin{tikzpicture}
\draw (0, 0) grid (5, 5);
\draw[ultra thick, red] (2, 0) to (2, 1) to (1, 1) to (1, 3) to (2, 3) to (2, 4) to (3, 4) to (3, 5) to (2, 5) ; 
\draw[ultra thick, green] (3, 1) to (4, 1) to (4, 2) to (5, 2) to (5, 3) to (3, 3) to (3, 1) ; 
\end{tikzpicture}
\caption{The figure shows trivial and non-trivial loops of the toric code: the green loop is a trivial loop and the red loop is a non-trivial loop.}
\label{fig:non-trivial-cycle}
\end{figure}
The error correction corresponding to $X$ errors is done similarly; we first perform $Z$ type generator measurements, which reveal the boundary of the $X$ error pattern on the dual lattice, and then the error correction is performed based on MWPM.

Let $\cD_Z$ (and $\cD_X$) be the quantum channel corresponding to the $Z$ (and $X$) error correction according to the above two steps. Then, we define the decoding map to be  
\begin{equation}
   \cD = \cE^\dagger \circ \cD_X \circ \cD_Z,
\end{equation}
where $\cE^\dagger = V^\dagger (\cdot) V$. Let $\mathbb{P}([2L^2])$ be the power set of $[2L^2]$. Let $\mathbb{E} \subseteq \mathbb{P}([2L^2])$ be such that any Pauli error on any $E \in \mathbb{E}$ is correctable, that is,
\begin{equation} \label{eq:toric-code-rec}
 \cD \circ \cP(E) \circ \cE = \cI,
\end{equation}
where $\cP(E) = P (\cdot) P^\dagger$ for a Pauli operator $P$, acting non-trivially only on the qubits in $E$. Using the linearity of quantum error correction codes, and Eq.~(\ref{eq:toric-code-rec}), we have that Eq.~(\ref{eq:corr-cond}) is satisfied for any $E \in \mathbb{E}$. 

\smallskip In the following lemma, we show that toric code gives a reliable communication code for local stochastic channels with sufficiently small parameter. The proof of the lemma very heavily relies on the proof of error correction threshold of toric code given in~\cite{preskil2011toric}.
\begin{lemma}
    Consider a local stochastic channel $\Lambda_\alpha$ with parameter $\alpha \leq \frac{1}{36}$. The channels $\cE$ and $\cD$ corresponding to the toric code, with square lattice of length $L$, define an $(m, \zeta)$ communication code for $\Lambda_\alpha$, with $\zeta = \frac{16}{3} L^4 (36 \alpha)^{L/2}$.
\end{lemma}

\begin{proof}
Consider the set of non-correctable errors $\overline{\mathbb{E}}:= \mathbb{P}([2L^2]) \setminus \mathbb{E}$, which is given by, 
\begin{align}
    \overline{\mathbb{E}} &= \{ E \subseteq [2L^2] : E + E' \text{ contains a non-trivial cycle} \} \nonumber\\
    &\subseteq  \{ E \subseteq [2L^2] : E + E' \text{ contains a closed loop of size } l \geq L \}. \label{eq:E-bar}
\end{align}
We have that
\begin{align}
&\{  E \subseteq [2L^2] : E + E' \text{ is a closed loop of size $l$ }\}  \nonumber\\ 
&= \{  E \subseteq [2L^2] : E + E' \text{ is a closed loop of size $l$ and } |E| \geq \frac{l}{2}\} \nonumber\\
&\subseteq \{ E \subseteq [2L^2] : E \text{ is included in a closed loop of size $l$ and } |E| \geq \frac{l}{2} \} \nonumber\\
&\subseteq \cup_{S \in \cS_l} \{ E : E \subseteq S, |E| \geq l/2 \}, \label{eq:E-S-loop}
\end{align}
where $\cS_l$ is the set of all closed loops of length $l$ on the toric code lattice. 
Consider now the decomposition of $\Lambda_\alpha$ as follows (see Def.~\ref{def:local-stoc}): 
\begin{align} 
\Lambda_\alpha 
 &=  \sum_{E \in \mathbb{E}} \big( \otimes_{i \in [2L^2]\setminus E} \cI \big) \otimes {\Lambda_\alpha}_E + \sum_{E \in \overline{\mathbb{E}}} \big( \otimes_{i \in [2L^2]\setminus E} \cI \big) \otimes {\Lambda_\alpha}_E \label{eq:dec-lambda},
\end{align}
where maps ${\Lambda_\alpha}_E$ are completely positive. We have that 
\begin{multline}    \label{eq:comm-toric}
 \dnorm{ \cD \circ \Lambda_\alpha \circ \cE - \sum_{E \in \mathbb{E}}  \cD \circ \left(\big( \otimes_{i \in [2L^2]\setminus E} \cI \big)  \otimes {\Lambda_\alpha}_E\right) \circ \cE}  \\ =  \dnorm{ \cD \circ \Lambda_\alpha \circ \cE - \sum_{E \in \mathbb{E}} K({\Lambda_\alpha}_E) \cI} 
 \leq  \dnorm{\sum_{E \in \overline{\mathbb{E}}} \big( \otimes_{i \in [2L^2]\setminus E} \cI \big) \otimes {\Lambda_\alpha}_E}
\end{multline}
Furthermore,
\begin{align}
    \dnorm{\sum_{E \in \overline{\mathbb{E}}} \big( \otimes_{i \in [2L^2]\setminus E} \cI \big) \otimes {\Lambda_\alpha}_E} &\leq \dnorm{ \sum_{l = L}^{2L^2} \sum_{S \in \cS_l} \: \sum_{\substack{E: \tilde{E} \subseteq E \\ \tilde{E} \subseteq S, |\tilde{E}| \geq \frac{l}{2} }} \big( \otimes_{i \in [2L^2]\setminus E} \cI \big) \otimes {\Lambda_\alpha}_E }  \nonumber\\
& \leq \sum_{l = L}^{2L^2} \sum_{S \in \cS_l} \: \sum_{\tilde{E} \subseteq S, |\tilde{E}| \geq \frac{l}{2}} \alpha^{|\tilde{E}|} \nonumber\\
& \leq \sum_{l = L}^{2L^2}   (4 L^2 3^{l-1}) (2^l) \alpha^{l/2} \nonumber \\
& \leq \frac{8}{3} L^4 (36 \alpha)^{L/2}, \text{ for $\alpha \leq \frac{1}{36}$} \label{eq:upper-bound-bad}
\end{align}
where the first inequality uses Eq.~(\ref{eq:E-bar}), (\ref{eq:E-S-loop}) and the complete positivity of maps ${\Lambda_\alpha}_E$, and also the fact that the largest loop is of size $2L^2$. The second inequality uses the property of local stochastic channels from Eq.~(\ref{eq:local-stc-2}), and the last inequality uses that there are $2^l$ subsets of a loop of length $l$ and there are less than $4 L^2 3^{l-1}$ loops of length $l$ (see~\cite{preskil2011toric} for more details). From Eqs.~(\ref{eq:comm-toric}), and (\ref{eq:upper-bound-bad}), we get
\begin{equation}
    \dnorm{ \cD \circ \Lambda_\alpha \circ \cE - \sum_{E \in \mathbb{E}} K({\Lambda_\alpha}_E) \cI} \leq \frac{8}{3} L^4 (36 \alpha)^{L/2}, \label{eq:comm-toric-1}
\end{equation}
Eq.~(\ref{eq:comm-toric-1}) implies that $|1 -  \sum_{E \in \mathbb{E}} K({\Lambda_\alpha}_E)| \leq \frac{8}{3} L^4 (36 \alpha)^{L/2}$. We, therefore, have 
\begin{align}
    \dnorm{ \cD \circ \Lambda_\alpha\circ \cE - \cI} \leq \frac{16}{3} L^4 (36 \alpha)^{L/2}.
\end{align}
Therefore, channels $\cE$ and $\cD$ define an $(m, \zeta)$ communication code for the quantum channel $\Lambda_\alpha$, with $\zeta = \frac{16}{3} L^4 (36 \alpha)^{L/2}$.
\end{proof}

\end{appendix}
\end{document}